\documentclass[
  twocolumn,
  letterpaper,
  dvipsnames,
  allowtoday,
]{quantumarticle}

\pdfoutput=1

\usepackage[T1]{fontenc}
\usepackage{graphicx}
\usepackage[normalem]{ulem}

\usepackage[utf8]{inputenc}
\usepackage[english]{babel}

\usepackage{amsmath,amssymb,amsthm,mathtools}
\usepackage{latexsym}
\usepackage{mathrsfs}
\usepackage{extarrows}
\usepackage{slashed}
\usepackage{verbatim}
\usepackage{comment}
\usepackage{soul}
\usepackage[toc,page]{appendix}
\usepackage[vcentermath]{youngtab}
\usepackage{multirow}

\usepackage{tikz}
\usetikzlibrary{decorations.pathreplacing}
\usepackage{quantikz}

\usepackage{pgfplots}
\pgfplotsset{compat=1.18}

\usepackage{hyperref}
\usepackage[capitalise]{cleveref}
\usepackage[numbers]{natbib}

\usepackage{physics}

\newtheorem{theorem}{Theorem}

\definecolor{cothblue}{RGB}{31,119,180}
\definecolor{exporange}{RGB}{255,127,14}
\definecolor{invgreen}{RGB}{44,160,44}
\definecolor{invblue}{RGB}{23,70,162}
\definecolor{mybluecolor}{RGB}{0,0,139}

\def\:={\,\raisebox{0.85pt}{.}\hspace{-2.78pt}\raisebox{2.85pt}{.}\!\!=\,}
\def\=:{\,=\!\!\raisebox{0.85pt}{.}\hspace{-2.78pt}\raisebox{2.85pt}{.}\,}

\newcommand{\HSket}[1]{\left|#1\right)}
\newcommand{\HSbra}[1]{\left(#1\right|}
\newcommand{\HSinner}[2]{\left(#1\middle|#2\right)}
\newcommand{\HSnorm}[1]{\|#1\|_{\mathrm{HS}}}

\usepackage{xparse}

\makeatletter

\newif\ifinappendix
\inappendixfalse

\renewcommand{\appendixtocname}{Contents of the Appendix}

\newcommand{\printappendixtoc}{%
  \begingroup
  \setcounter{tocdepth}{2}
  \begin{center}
    \textbf{\large \appendixtocname}
  \end{center}
  \@starttoc{atoc}%
  \endgroup
}

\let\oldsection\section
\RenewDocumentCommand{\section}{s o m}{%
  \IfBooleanTF{#1}
    {\oldsection*{#3}}%
    {%
      \IfNoValueTF{#2}{\oldsection{#3}}{\oldsection[#2]{#3}}%
      \ifinappendix
        \addcontentsline{atoc}{section}{%
          \protect\numberline{\thesection}%
          \IfNoValueTF{#2}{#3}{#2}}%
      \fi
    }%
}

\let\oldsubsection\subsection
\RenewDocumentCommand{\subsection}{s o m}{%
  \IfBooleanTF{#1}
    {\oldsubsection*{#3}}%
    {%
      \IfNoValueTF{#2}{\oldsubsection{#3}}{\oldsubsection[#2]{#3}}%
      \ifinappendix
        \addcontentsline{atoc}{subsection}{%
          \protect\numberline{\thesubsection}%
          \IfNoValueTF{#2}{#3}{#2}}%
      \fi
    }%
}

\makeatother

\begin{document}
\title{Universal scaling of finite-temperature quantum adiabaticity in driven many-body systems}
\date{\today}

\author{Li-Ying Chou}
\email{liying729927@gmail.com}
\affiliation{Department of Physics, National Central University, Chungli 32001, Taiwan}
\orcid{0009-0007-7704-9776}

\author{Jyong-Hao Chen}
\email{jyonghaochen@gmail.com }
\affiliation{Department of Physics, National Central University, Chungli 32001, Taiwan}
\orcid{0000-0003-3951-049X}

\begin{abstract}
Establishing quantitative adiabaticity criteria at finite temperature remains substantially less developed than in the pure-state setting, even though realistic quantum systems are never at absolute zero.
Here, by combining a mixed-state quantum speed limit with mixed-state fidelity susceptibility within the Liouville-space formulation of quantum mechanics, we derive rigorous bounds on the Hilbert--Schmidt fidelity between mixed states.
Focusing on protocols that drive an initial Gibbs state toward a quasi-Gibbs target, these bounds yield an explicit threshold driving rate for the onset of nonadiabaticity.
For a broad class of local Hamiltonians in gapped phases, we show that, in the thermodynamic limit, the threshold driving rate factorizes into a system-size contribution that recovers the zero-temperature scaling and a universal temperature-dependent factor.
The latter is exponentially close to unity at low temperature, whereas at high temperature it is linear in temperature.
We verify the predicted scaling in several spin-\(1/2\) chains by obtaining closed-form expressions for the threshold driving rate.
Our results provide a practical and largely model-independent criterion for finite-temperature adiabaticity in driven many-body systems under closed-system unitary evolution.
\end{abstract}

\maketitle

\section{Introduction}

Slow driving is one of the few broadly useful routes to controlling the dynamics of quantum many-body systems.
At zero temperature, the quantum adiabatic theorem~\cite{Born27,Born28} asserts that a system initialized in an eigenstate remains close to the corresponding instantaneous eigenstate when the Hamiltonian varies sufficiently slowly.
However, in experiments and in many applications of quantum simulation and quantum information processing, systems inevitably start at finite temperature---even if very low---so it is essential to quantify when adiabatic following persists away from absolute zero.
While {\it zero-temperature} adiabaticity has been studied extensively for decades~\cite{Born27,Born28,Kato50,Messiah14,Nenciu93,Avron99,Hagedorn02,Teufel03,Ambainis06,Jansen07,Amin09,Lidar09,Cheung11,Elgart12,Ge16,Bachmann17,Albash18,Bachmann18,Bachmann20}, quantitative criteria at {\it finite temperature} remain much less developed.
This limitation is relevant not only for realistic driven experiments, but also for emerging directions such as thermal-state preparation in quantum information science~\cite{irmejs2026,granet2025,shirai2025} and finite-temperature topological phases in condensed matter physics~\cite{Hastings11,Unanyan20,Zhou25}.

Several recent works~\cite{Skelt_2020,Ilin_2021,Greenblatt_2024} have nevertheless advanced the study of finite-temperature adiabaticity.
First, Ref.~\cite{Skelt_2020} proposed metrics based on the quantum state and particle density for characterizing the degree of adiabaticity in finite-temperature many-body dynamics, and discussed the role of memory effects.
Their framework is insightful and potentially useful in both theory and experiment, but it is primarily diagnostic in nature.
Second, Ref.~\cite{Ilin_2021} extended the notion of adiabaticity to closed systems initially prepared in Gibbs states, identified the quasi-Gibbs state as the natural finite-temperature target, and proved a sufficient condition for adiabatic following.
They also showed that finite-temperature adiabaticity can, in some many-body settings, be more robust than its zero-temperature counterpart.
However, their general result is formulated as a trace-distance bound involving several auxiliary quantities and integral terms, so it does not immediately provide a simple, broadly applicable criterion for generic driven many-body systems.

More recently, Ref.~\cite{Greenblatt_2024} established a rigorous adiabatic theorem for low-temperature finite-range many-body fermionic lattice systems under suitable assumptions on the driving.
They proved closeness of the evolved state to the instantaneous Gibbs state in the sense of expectation values of local observables and obtained linear response as a corollary.
Their result is mathematically powerful, but it is restricted to low temperatures and a specific class of many-body systems.
Despite these important advances, a quantitatively sharp and broadly applicable finite-temperature criterion for adiabaticity breakdown in driven many-body systems is still lacking.
In particular, it remains unclear how to characterize the onset of adiabatic breakdown in a model-independent way and how such a characterization should connect to quantities that can be computed systematically in many-body systems.

In this work, we provide such a criterion by combining two concepts that are well developed for pure states but less exploited for mixed states: the quantum speed limit (QSL)~\cite{Mandelstam45,Vaidman92,Pfeifer93a,Pfeifer93b,Pfeifer95,Deffner17} and fidelity susceptibility~\cite{Zanardi_2006,Zanardi_2007,You_2007,GU_2010}.
Working in the Liouville-space formulation~\cite{Ernst87,Fick90,Gyamfi_2020}, we derive rigorous bounds on the Hilbert--Schmidt fidelity between mixed states under closed-system unitary evolution and apply them to protocols that drive an initial Gibbs state toward a quasi-Gibbs target~\cite{Skelt_2020,Ilin_2021}.
These bounds lead to an explicit threshold driving rate \(\Gamma^{\,}_{\mathrm{th}}\) that quantitatively characterizes the onset of adiabatic breakdown.

\begin{table}[t]
\begin{center}
\begin{tabular}{lcc}
\hline
\hline
 & low-temp.\ regime & high-temp.\ regime \\
\hline
\(f(\beta)\) &
\(\simeq 1 + c^{\,}_{1} e^{-\beta \Delta}\) &
\(\simeq c^{\,}_{2}/\beta\) \\
\hline
\hline
\end{tabular}
\end{center}
\caption{Universal scaling of the temperature-dependent factor \(f(\beta)\) in the threshold driving rate
\(\Gamma^{\,}_{\mathrm{th}}\sim\Gamma^{\,}_{N}f(\beta)\) for generic local quantum many-body systems in gapped phases.
Here, \(\beta\) is the inverse temperature and \(\Delta\) is the relevant excitation gap of the initial Hamiltonian.
The coefficient \(c^{\,}_{1}\in(0,2]\) is a model-dependent dimensionless constant, whereas \(c^{\,}_{2}>0\) has dimensions of inverse energy.
For the transverse-field Ising and quantum XY chains, \(f(\beta)=\coth(\beta\Delta/2)\) exactly.}
\label{tab:driving-rate-scaling}
\end{table}

Our central result is an explicit finite-temperature adiabatic threshold for closed-system driving from an initial Gibbs state: for a broad class of local Hamiltonians in gapped phases, the threshold factorizes in the thermodynamic limit (\(N\to\infty\)) as
\(
\Gamma^{\,}_{\mathrm{th}}\sim\Gamma^{\,}_{N}f(\beta).
\)
\footnote{Throughout, ``\(\sim\)'' denotes asymptotic equivalence as \(N\to\infty\): \(A^{\,}_N\sim B^{\,}_N\) means \(\lim_{N\to\infty}A^{\,}_N/B^{\,}_N=1\).}
Here, \(\Gamma^{\,}_{N}\) captures the familiar zero-temperature size dependence (typically decreasing with increasing \(N\)~\cite{Lychkovskiy:2016pfb}), while \(f(\beta)\) is a universal temperature-dependent factor summarized in Table~\ref{tab:driving-rate-scaling}.
In particular, \(f(\beta)\) is exponentially close to unity at low temperature, with the scale set by the relevant excitation gap, and is linear in temperature in the high-temperature regime.
To the best of our knowledge, this universal temperature dependence has not been quantitatively characterized in this form for generic driven many-body systems.
We verify this structure in several spin-chain models, including the transverse-field Ising, quantum XY, and mixed-field Ising chains, where closed-form expressions for \(f(\beta)\) can be obtained using transfer-matrix techniques~\cite{Huang1987,Nishimori2010}.

This paper is organized as follows.
In Sec.\ \ref{sec:liouville_space_formulation_and_mixed_state_fidelity}, we introduce the general Liouville-space formulation and the mixed-state fidelity measure used throughout.
In Sec.\ \ref{sec:mixed_state_quantum_speed_limit_and_fidelity_bounds}, we derive mixed-state fidelity bounds based on quantum speed limit ideas within this framework.
We then apply these bounds to finite-temperature adiabatic dynamics in Sec.\ \ref{sec:quantum_adiabaticity_at_finite_temperature}, where we introduce the threshold driving rate \(\Gamma^{\,}_{\mathrm{th}}\).
In Sec.\ \ref{sec:universal_temperature_scaling_of_the_threshold_driving_rate}, we determine the universal low- and high-temperature limits of \(\Gamma^{\,}_{\mathrm{th}}\) and clarify the physical origin of the corresponding scaling behavior.
In Sec.\ \ref{sec:spin_chain_models_as_illustrations}, we illustrate and verify the general results using representative spin-chain models.
Finally, in Sec.\ \ref{sec:summary_and_outlook}, we summarize the results, discuss broader implications, and outline possible extensions.
Additional technical details are presented in the Appendices.

\begin{figure}[t]
    \centering
    \begin{tikzpicture}[
        x=1.0cm,y=1.0cm,
        line cap=round,
        line join=round,
        tri/.style={draw=black, line width=0.75pt},
        statedot/.style={circle, draw=black, fill=white, minimum size=2.7mm, inner sep=0pt}
    ]

    \coordinate (L) at (0,0);
    \coordinate (R) at (5.4,0);
    \coordinate (B) at (2.7,-3.15);

    \draw[tri] (L)--(R);
    \draw[tri] (L)--(B);
    \draw[tri] (B)--(R);

    \node[statedot] at (L) {};
    \node[statedot] at (R) {};
    \node[statedot] at (B) {};

    \node[font=\normalsize] at (0,0.82) {target state};
    \node[font=\normalsize] at (5.4,0.82) {dynamical state};
    \node[font=\normalsize] at (2.7,-4.12) {initial state};

    \node[font=\large] at (0,0.4) {$\hat{\sigma}^{\,}_\lambda$};
    \node[font=\large] at (5.4,0.4) {$\hat{\rho}^{\,}_\lambda$};
    \node[font=\large] at (2.7,-3.55) {$\hat{\rho}^{\,}_0$};

    \node[font=\large] at (2.7,0.3)
        {$F[\hat{\sigma}^{\,}_\lambda,\hat{\rho}^{\,}_\lambda]=?$};

    \node[font=\large, rotate=-49]
        at ($(L)!0.53!(B)+(-0.4,0.02)$)
        {$F[\hat{\sigma}^{\,}_\lambda,\hat{\rho}^{\,}_0]$};

    \node[font=\large, rotate=49]
        at ($(B)!0.52!(R)+(0.4,0.01)$)
        {QSL};

    \end{tikzpicture}
    
    \caption{
    Schematic illustration of the relation between the initial state $\hat{\rho}^{\,}_0$, the dynamical state $\hat{\rho}^{\,}_\lambda$, 
    and the target state $\hat{\sigma}^{\,}_\lambda$. 
    The quantity of main interest is the mixed-state fidelity \(F[\hat{\sigma}^{\,}_{\lambda},\hat{\rho}^{\,}_{\lambda}]\). 
    The right edge is bounded by a quantum speed limit (QSL) on the dynamical Hilbert--Schmidt angle between \(\hat{\rho}^{\,}_{0}\) and \(\hat{\rho}^{\,}_{\lambda}\), whereas the left edge is given by the fidelity \(F[\hat{\sigma}^{\,}_{\lambda},\hat{\rho}^{\,}_{0}]\).
    }
    \label{fig:mixed_state_triangle}
\end{figure}

\section{Liouville-space formulation and mixed-state fidelity}
\label{sec:liouville_space_formulation_and_mixed_state_fidelity}

We consider a closed system with initial density matrix \(\hat{\rho}^{\,}_{0}\), driven by a Hamiltonian \(\hat{H}^{\,}_{\lambda}\) with time-dependent control parameter \(\lambda=\lambda(t)\), where \(\lambda(0)=0\) and \(\lambda(t)>0\) for \(t>0\).
Let \(\hat{\rho}^{\,}_{\lambda}\) denote the dynamical state obtained from \(\hat{\rho}^{\,}_{0}\) under the unitary evolution generated by \(\hat{H}^{\,}_{\lambda(t)}\), and let \(\hat{\sigma}^{\,}_{\lambda}\) denote a target family satisfying \(\hat{\sigma}^{\,}_{\lambda(0)}=\hat{\rho}^{\,}_{0}\).
We seek to bound how closely the dynamical state \(\hat{\rho}^{\,}_{\lambda}\) can track the target state \(\hat{\sigma}^{\,}_{\lambda}\) as the control parameter varies.
See Fig.~\ref{fig:mixed_state_triangle} for an illustration of the relation among the three states \(\hat{\rho}^{\,}_{0}\), \(\hat{\rho}^{\,}_{\lambda}\), and \(\hat{\sigma}^{\,}_{\lambda}\).

In the pure-state setting, closeness can be quantified by the overlap between the two state vectors, and corresponding bounds can be derived by projecting onto the subspace associated with the initial state~\cite{Lychkovskiy:2016pfb,Chen:2021gbb,Chen:2022qyb}.
For mixed states, however, there are several inequivalent notions of fidelity and distance~\cite{Jozsa94,Nielsen00,Wang_2008,Wilde17,Liang19}.
Common choices, such as the Uhlmann fidelity~\cite{uhlmann76} and the trace distance~\cite{Holevo73,Helstrom76,FuchsVanDeGraaf99}, are often difficult to evaluate for generic many-body states, since they involve matrix square roots or trace norms, typically requiring diagonalization or singular-value decomposition.

To generalize the projection-operator approach of Refs.~\cite{Chen:2021gbb,Chen:2022qyb}, we work in {\it Liouville-space}~\cite{Ernst87,Fick90,Gyamfi_2020}, where an operator \(\hat{A}\) is represented as a vector \(\HSket{A}\).
This space is endowed with the {\it Hilbert--Schmidt inner product}
\(\HSinner{A}{B}:=\Tr\big[\hat{A}^{\dag}\hat{B}\big]\)
and the induced {\it Hilbert--Schmidt norm} \(\HSnorm{A}:=\sqrt{\HSinner{A}{A}}\).
In particular, for a density matrix \(\hat{\rho}\), one has \(\HSnorm{\rho}^{2}=\Tr[\hat{\rho}^{2}]\), namely, the purity.

For any two density matrices \(\hat{\rho}\) and \(\hat{\sigma}\), we define the fidelity in Liouville-space, analogously to the pure-state case, as the squared overlap of the normalized Liouville vectors,
\begin{align}
F[\hat{\rho},\hat{\sigma}]
:=\left(\frac{|\HSinner{\rho}{\sigma}|}{\HSnorm{\rho}\HSnorm{\sigma}}\right)^2
=\frac{\big(\Tr[\hat{\rho}\hat{\sigma}]\big)^{2}}{\Tr[\hat{\rho}^{2}]\,\Tr[\hat{\sigma}^{2}]}\,,
\label{eq:mixed-state_fidelity}
\end{align}
which motivates the term \textit{Hilbert--Schmidt fidelity}~\footnote{The Hilbert--Schmidt fidelity defined in this work is the square of the {\it geometric mean fidelity} defined in Eq.\ (2.11) of Ref.~\cite{Liang19} and the square of the {\it operator fidelity} introduced in Eq.\ (5) of Ref.~\cite{Wang_2008}.}.
Note that, by the Cauchy--Schwarz inequality, one has \(0\le F[\hat{\rho},\hat{\sigma}]\le 1\), with \(F[\hat{\rho},\hat{\sigma}]=1\) if and only if \(\hat{\rho}=\hat{\sigma}\).

\section{Mixed-state quantum speed limit and fidelity bounds}
\label{sec:mixed_state_quantum_speed_limit_and_fidelity_bounds}

The Hilbert--Schmidt fidelity between the initial state \(\hat{\rho}^{\,}_{0}\) and the dynamical state \(\hat{\rho}^{\,}_{\lambda}\), namely
$F[\hat{\rho}^{\,}_{0},\hat{\rho}^{\,}_{\lambda}],$ plays a central role in what follows.
For convenience, we introduce the dynamical {\it Hilbert--Schmidt angle}
\(\Theta^{\,}_{\lambda} := \arccos\sqrt{F[\hat{\rho}^{\,}_{0},\hat{\rho}^{\,}_{\lambda}]}\),
which quantifies the separation between \(\hat{\rho}^{\,}_{0}\) and \(\hat{\rho}^{\,}_{\lambda}\).
Analogous to the role of the Fubini--Study angle~\cite{Anandan90} in the pure-state case, \(\Theta^{\,}_{\lambda}\) for a closed system is bounded from above by a mixed-state quantum speed limit (QSL) inequality (see App.~\ref{sec: derivation_mixed_QSL}):
\begin{subequations}
\label{eq: QSL mixed states}
\begin{align}
\Theta^{\,}_\lambda \leq \mathcal{R}(\lambda),
\;
\mathcal{R}(\lambda) :=
\int^{\lambda}_{0} 
\frac{d\lambda'}{\bigl|\partial^{\,}_{t}\lambda'\bigr|}
\sqrt{
2\,I^{\,}_{\mathrm{WY}}\!\left(\hat{\tilde{\rho}}^{\,}_{0},\hat{H}^{\,}_{\lambda'}\right)
},
\label{eq: QSL mixed states a}
\end{align}
where $\hat{\tilde{\rho}}^{\,}_{0}:=\hat{\rho}^{2}_{0}/\Tr[\hat{\rho}^{2}_{0}]$ is the escort density matrix of order-2~\cite{Naudts2005}, and
\begin{align}
I^{\,}_{\mathrm{WY}}\!\left(\hat{\tilde{\rho}}^{\,}_{0},\hat{H}^{\,}_{\lambda}\right)
&:= 
\frac{1}{2}\left\|\bigl[\hat{\tilde{\rho}}^{1/2}_{0},\hat{H}^{\,}_{\lambda}\bigr]\right\|^{2}_{\mathrm{HS}}
\nonumber\\
&=
\Tr\!\left[\hat{\tilde{\rho}}^{\,}_{0}\,\hat{H}^{2}_{\lambda}\right]
-\Tr\!\left[\hat{\tilde{\rho}}^{1/2}_{0}\hat{H}^{\,}_{\lambda}\hat{\tilde{\rho}}^{1/2}_{0}\hat{H}^{\,}_{\lambda}\right]
\label{eq: QSL mixed states b}
\end{align}
\end{subequations}
is the {\it Wigner--Yanase skew information}~\cite{Wigner1963} of the state $\hat{\tilde{\rho}}^{\,}_{0}$ with respect to the Hamiltonian $\hat{H}^{\,}_{\lambda}$.
In this formulation, \(I^{\,}_{\mathrm{WY}}(\hat{\tilde{\rho}}^{\,}_{0},\hat{H}^{\,}_{\lambda})\) quantifies the noncommutativity between the escort state 
\(\hat{\tilde{\rho}}^{\,}_{0}\) and the Hamiltonian \(\hat{H}^{\,}_{\lambda}\), while \(\sqrt{2\,I^{\,}_{\mathrm{WY}}(\hat{\tilde{\rho}}^{\,}_{0},\hat{H}^{\,}_{\lambda})}\) bounds the instantaneous ``speed'' associated with \(\Theta^{\,}_{\lambda}\).

Three remarks on Eq.~\eqref{eq: QSL mixed states} are in order.
First, Eq.~\eqref{eq: QSL mixed states} does not reduce to the standard Mandelstam--Tamm bound~\cite{Mandelstam45} in the pure-state limit: if \(\hat{\rho}^{\,}_{0}\) is pure, the integral \(\mathcal{R}(\lambda)\) exceeds the pure-state result by a factor of \(\sqrt{2}\), indicating that this mixed-state QSL need not be optimal.
This suboptimality is immaterial for our purposes, since we will use Eq.~\eqref{eq: QSL mixed states} only to extract the scaling form of the adiabaticity-breakdown 
condition.
Second, when the Hamiltonian is time-independent, Eq.~\eqref{eq: QSL mixed states} reduces to a form consistent with the result 
of Ref.~\cite{BolonekLason2021}, obtained there by a different approach.
Here, by contrast, Eq.~\eqref{eq: QSL mixed states} is derived within the Liouville-space formulation for general time-dependent Hamiltonians.
Third, while many mixed-state QSLs have been proposed (see, e.g., Refs.~\cite{Taddei13,Campo13,Deffner13,Uzdin16,Pires16,Pires18,Funo19,Ilin21b,Srivastav25}), adopting the Hilbert--Schmidt angle as our distance measure naturally leads to the corresponding Liouville-space QSL derived here.
To the best of our knowledge, this result does not appear to have been reported previously, apart from the time-independent special case of Ref.~\cite{BolonekLason2021} noted above and related QSL bounds in Ref.~\cite{Pires18} [see Eqs.~(4) and (6) therein].

The fidelity between the dynamical state \(\hat{\rho}^{\,}_{\lambda}\) and the target state \(\hat{\sigma}^{\,}_{\lambda}\), namely \(F[\hat{\sigma}^{\,}_{\lambda},\hat{\rho}^{\,}_{\lambda}]\), is the quantity we aim to bound.
Using a projection-operator approach similar to that in the pure-state case~\cite{Chen:2021gbb}, together with the QSL inequality~\eqref{eq: QSL mixed states}, we obtain (see App.~\ref{sec:derivation_of_mixed_state_fidelity_bounds_ref_eq_fidelity_bound})
\begin{subequations}
\label{eq: fidelity bound}
\begin{align}
    \bigl|F[\hat{\sigma}^{\,}_{\lambda},\hat{\rho}^{\,}_{\lambda}] 
          - F[\hat{\sigma}^{\,}_{\lambda},\hat{\rho}^{\,}_{0}]\bigr| 
    &\leq \sin\bigl(\widetilde{\mathcal{R}}(\lambda)\bigr),
    \label{eq: fidelity bound a}
\end{align}
and
\begin{align}
    \bigl|F[\hat{\sigma}^{\,}_{\lambda},\hat{\rho}^{\,}_{\lambda}] 
          - F[\hat{\sigma}^{\,}_{\lambda},\hat{\rho}^{\,}_{0}]\bigr|
    &\leq g(\lambda),
    \label{eq: fidelity bound b}
\end{align}
\end{subequations}
where \(F[\hat{\sigma}^{\,}_{\lambda},\hat{\rho}^{\,}_{0}]\) is the Hilbert--Schmidt fidelity between the initial state \(\hat{\rho}^{\,}_{0}\) and the target state \(\hat{\sigma}^{\,}_{\lambda}\).
Here, \(g(\lambda)\) is an auxiliary function defined as
\begin{subequations}
\begin{align}
    g(\lambda) &:= g^{\,}_{1}(\lambda)+g^{\,}_{2}(\lambda),\\
    g^{\,}_{1}(\lambda) &:= \sin^{2}\!\widetilde{\mathcal{R}}(\lambda)\,
        \bigl|1-2F[\hat{\sigma}^{\,}_{\lambda},\hat{\rho}^{\,}_{0}]\bigr|,\\
    g^{\,}_{2}(\lambda) &:= \sin\!\bigl(2\widetilde{\widetilde{\mathcal{R}}}\!(\lambda)\bigr)\,
        \sqrt{F[\hat{\sigma}^{\,}_{\lambda},\hat{\rho}^{\,}_{0}]}\,
        \sqrt{1-F[\hat{\sigma}^{\,}_{\lambda},\hat{\rho}^{\,}_{0}]},
    \label{eq: fidelity bound c}
\end{align}
with
\begin{align}
    \widetilde{\mathcal{R}}(\lambda) 
    := \min\!\left(\mathcal{R}(\lambda),\frac{\pi}{2}\right),
    \qquad
    \widetilde{\widetilde{\mathcal{R}}}(\lambda) 
    := \min\!\left(\mathcal{R}(\lambda),\frac{\pi}{4}\right).
    \nonumber
\end{align}
\end{subequations}
By construction, the bound in Eq.~\eqref{eq: fidelity bound a} is strictly weaker than that in Eq.~\eqref{eq: fidelity bound b}.

In the rest of this work, we focus on Hamiltonians of the form
\begin{align}
  \hat{H}^{\,}_{\lambda} = \hat{H}^{\,}_{0} + \lambda\,\hat{V},
  \label{eq: define interpolating Hamiltonian}
\end{align}
where \(\hat{H}^{\,}_{0}\) and \(\hat{V}\) are time-independent Hermitian operators with $[\hat{H}^{\,}_{0},\hat{V}]\neq0$.
Assuming that the initial state is stationary with respect to the initial Hamiltonian, i.e., \([\hat{\rho}^{\,}_{0},\hat{H}^{\,}_{0}]=0\), and that the driving rate \(\Gamma := \partial^{\,}_{t}\lambda\) is a positive constant, the QSL integral \(\mathcal{R}(\lambda)\) in Eq.~\eqref{eq: QSL mixed states} simplifies to
\begin{align}
\mathcal{R}(\lambda) =
\frac{\lambda^{2}}{2\Gamma}\,\delta V, 
\qquad 
\delta V := \sqrt{2\,I^{\,}_{\mathrm{WY}}\!\left(\hat{\tilde{\rho}}^{\,}_{0},\hat{V}\right)}.
\label{eq: QSL mixed states simplified}
\end{align}
Here, \(\delta V\) quantifies the quantum fluctuation of the driving term \(\hat{V}\) in the initial state \(\hat{\rho}^{\,}_{0}\) and will play a key role below.

\section{Quantum adiabaticity at finite temperature}
\label{sec:quantum_adiabaticity_at_finite_temperature}

Although the inequalities derived in Eq.~\eqref{eq: fidelity bound} are valid for arbitrary unitary dynamics,
we now specialize to finite-temperature adiabatic evolution and show how the mixed-state fidelity bounds in Eq.~\eqref{eq: fidelity bound} constrain the resulting dynamics.
We assume that the initial state $\hat{\rho}^{\,}_{0}$ is a Gibbs (thermal) state of the form
\begin{align}
    \hat{\rho}^{\,}_{0}(\beta)
   := \frac{1}{Z^{\,}_{0}} \sum_{n\geq 0} e^{-\beta E^{(0)}_n} |E^{(0)}_{n}\rangle \langle E^{(0)}_{n}|,
    \quad 
    Z^{\,}_{0} := \sum_{n\geq0} e^{-\beta E^{(0)}_{n}},
    \label{eq: initial Gibbs state}
\end{align}
where $\beta$ is the inverse temperature and $|E^{(0)}_{n}\rangle$ denotes an eigenstate of the initial Hamiltonian $\hat{H}^{\,}_{0}$ with eigenvalue $E^{(0)}_{n}$.

Since adiabatic evolution transports each initial eigenstate $|E^{(0)}_n\rangle$ to its instantaneous counterpart, we define the target state $\hat{\sigma}^{\,}_{\lambda}$ as the {\it quasi-Gibbs state}~\cite{Skelt_2020,Ilin_2021}.
It is obtained by adiabatically transporting the eigenbasis of the initial Gibbs state~\eqref{eq: initial Gibbs state} to the instantaneous eigenbasis of $\hat{H}^{\,}_{\lambda}$ [Eq.~\eqref{eq: define interpolating Hamiltonian}] while keeping the initial Boltzmann weights fixed:
\begin{align}
    \hat{\sigma}^{\,}_{\lambda}(\beta) 
    := \frac{1}{Z^{\,}_{0}} \sum_{n\geq0} e^{-\beta E^{(0)}_n} |E^{\,}_n(\lambda)\rangle\langle E^{\,}_n(\lambda)|,
    \label{eq: quasi-Gibbs state}
\end{align}
where \(\ket{E^{\,}_{n}(\lambda)}\) are the instantaneous eigenstates of \(\hat{H}^{\,}_{\lambda}\) [Eq.~\eqref{eq: define interpolating Hamiltonian}], and \(E^{(0)}_n \equiv E^{\,}_{n}(\lambda=0)\).

For finite-temperature driving, we quantify the closeness of the dynamical state \(\hat{\rho}^{\,}_{\lambda}\) to the quasi-Gibbs target state \(\hat{\sigma}^{\,}_{\lambda}(\beta)\) [Eq.~\eqref{eq: quasi-Gibbs state}] via their Hilbert--Schmidt fidelity [Eq.~\eqref{eq:mixed-state_fidelity}],
\begin{align}
\mathcal{F}(\lambda)
:=F[\hat{\sigma}^{\,}_{\lambda}(\beta),\hat{\rho}^{\,}_{\lambda}]
=
\frac{\big(\Tr[\hat{\rho}^{\,}_{\lambda}\hat{\sigma}^{\,}_{\lambda}(\beta)]\big)^2}
{\Tr[\hat{\rho}^{2}_{\lambda}]\,\Tr[\hat{\sigma}^{2}_{\lambda}(\beta)]},
\label{eq:adiabatic_fidelity_defined}
\end{align}
which we refer to as the {\it adiabatic fidelity}.
Following Refs.~\cite{Lychkovskiy:2016pfb,Chen:2022qyb}, we diagnose adiabaticity by introducing an adiabatic mean-free path $\lambda^{\,}_{\mathrm{ad}}$, 
defined by the condition that $\mathcal{F}(\lambda)\ge e^{-1}$ for \(0\le \lambda\le \lambda^{\,}_{\mathrm{ad}}\).
Combining this criterion with the fidelity bounds~\eqref{eq: fidelity bound} yields an upper bound on the driving rate, \(\Gamma\le \Gamma^{\,}_{\mathrm{th}}\), where the {\it threshold driving rate} \(\Gamma^{\,}_{\mathrm{th}}\) takes the form
\begin{align}
\Gamma^{\,}_{\mathrm{th}}
:= \frac{\delta V}{\chi^{\,}_{\mathrm{F}}}\,\alpha,
\label{eq: define threshold driving rate}
\end{align}
with \(\alpha=\mathcal{O}(1)\)~\cite{Chen:2021gbb} for both bounds in Eq.~\eqref{eq: fidelity bound}. 
In this expression, \(\delta V\) is the quantum fluctuation given by Eq.~\eqref{eq: QSL mixed states simplified}, and
\begin{align}
\chi^{\,}_{\mathrm{F}}
:= -\frac{\partial^{2}\ln\mathcal{C}(\lambda)}{\partial \lambda^{2}}\Big|_{\lambda=0},
\label{eq: define fidelity susceptiblity}
\end{align}
is the (Hilbert--Schmidt) mixed-state {\it fidelity susceptibility}.%
~\footnote{
Fidelity susceptibility serves as a standard probe of quantum phase transitions~\cite{Zanardi_2006,Zanardi_2007,You_2007,GU_2010}.
For mixed states, the Hilbert--Schmidt fidelity used in this work provides a tractable alternative to the Uhlmann fidelity~\cite{Zanardi_2007,You_2007}.
}
The quantity
\begin{align}
\mathcal{C}(\lambda)
:=F[\hat{\rho}^{\,}_{0}(\beta),\hat{\sigma}^{\,}_{\lambda}(\beta)]
=
\frac{\big(\Tr[\hat{\rho}^{\,}_{0}(\beta)\hat{\sigma}^{\,}_{\lambda}(\beta)]\big)^2}
{\Tr[\hat{\rho}^{2}_{0}(\beta)]\,\Tr[\hat{\sigma}^{2}_{\lambda}(\beta)]},
\label{eq: thermal state overlap defined}
\end{align}
is the Hilbert--Schmidt fidelity between the initial Gibbs state \(\hat{\rho}^{\,}_{0}(\beta)\) [Eq.~\eqref{eq: initial Gibbs state}] and the quasi-Gibbs state \(\hat{\sigma}^{\,}_{\lambda}(\beta)\) [Eq.~\eqref{eq: quasi-Gibbs state}], which we refer to as the {\it thermal-state overlap}.

For comparison, we denote the zero-temperature (pure-state) counterpart of Eq.~\eqref{eq: define threshold driving rate} by
\begin{align}
\Gamma^{\,}_{N} := 
\frac{\delta V^{(0)}}{\chi^{(0)}_{\mathrm{F}}}\,\alpha,
\label{eq: threshold driving rate at zero temperature}
\end{align}
where $\delta V^{(0)}$ and $\chi^{(0)}_{\mathrm{F}}$ are the ground-state counterparts of $\delta V$ and $\chi^{\,}_{\mathrm{F}}$, respectively.
They are obtained by replacing the initial state $\hat{\rho}^{\,}_{0}$ with $|E^{(0)}_{0}\rangle\langle E^{(0)}_{0}|$ in Eq.~\eqref{eq: QSL mixed states simplified} 
and the target state $\hat{\sigma}^{\,}_{\lambda}$ with $\dyad{E^{\,}_{0}(\lambda)}$ in Eq.~\eqref{eq: define fidelity susceptiblity}.%
~\footnote{
Explicitly,
\begin{subequations}
\label{eq: pure state counterpart}
\begin{align}
\delta V^{(0)}
&:= \sqrt{2}\,\sqrt{\langle E^{(0)}_{0}|\hat{V}^{2}|E^{(0)}_{0}\rangle
   -\langle E^{(0)}_{0}|\hat{V}|E^{(0)}_{0}\rangle^{2}},
\label{eq: pure state counterpart a}
\\ 
\chi^{(0)}_{\mathrm{F}}
&:= -\frac{\partial^{2} \ln\mathcal{C}^{(0)}(\lambda)}{\partial \lambda^{2}}
   \Big|_{\lambda=0},
\label{eq: pure state counterpart b}
\end{align}
\end{subequations}
where \(\mathcal{C}^{(0)}(\lambda)
:=|\langle E^{(0)}_{0}|E^{\,}_{0}(\lambda)\rangle|^{4}\) is the ground-state counterpart of the thermal-state overlap \(\mathcal{C}(\lambda)\) in Eq.~\eqref{eq: thermal state overlap defined}.
In the thermodynamic limit, \(\mathcal{C}^{(0)}(\lambda)\sim \exp(-\chi^{(0)}_{\mathrm{F}}\lambda^{2}/2)\) with \(\chi^{(0)}_{\mathrm{F}}\asymp N\), a behavior known as the \emph{generalized orthogonality catastrophe}~\cite{Lychkovskiy:2016pfb}, in analogy with Anderson’s orthogonality catastrophe~\cite{Anderson1967Infrared}.
}
For typical gapped systems with the driving term $\hat{V}$ being a sum of local operators, one has~\cite{Lychkovskiy:2016pfb} $\delta V^{(0)}\asymp \sqrt{N}$ and $\chi^{(0)}_{\mathrm{F}}\asymp N$, and thus $\Gamma^{\,}_{N}\asymp 1/\sqrt{N}$%
\footnote{
Throughout, we use ``\(\asymp\)'' to indicate the same scaling with \(N\) up to an \(N\)-independent prefactor; i.e., \(A^{\,}_N\asymp B^{\,}_N\) means that \(A^{\,}_N/B^{\,}_N\) remains bounded away from zero and infinity as \(N\to\infty\).
},
i.e., the zero-temperature threshold driving rate decreases with increasing system size \(N\).
This scaling provides a zero-temperature reference for assessing finite-temperature effects.

\section{Universal temperature scaling of the threshold driving rate}
\label{sec:universal_temperature_scaling_of_the_threshold_driving_rate}

The temperature dependence of the threshold driving rate $\Gamma^{\,}_{\mathrm{th}}$
[Eq.~\eqref{eq: define threshold driving rate}] is the main focus of this work.
For a broad class of models (specified below) in the setup defined by
Eqs.~\eqref{eq: define interpolating Hamiltonian}, \eqref{eq: initial Gibbs state}, and
\eqref{eq: quasi-Gibbs state},
we find that, in the thermodynamic limit ($N\to\infty$), the threshold driving rate factorizes as
\begin{align}
\Gamma^{\,}_{\mathrm{th}} \sim \Gamma^{\,}_{N} f(\beta),
\label{eq: factorization of threshold driving rate}
\end{align}
where \(\Gamma^{\,}_{N}\) is the zero-temperature (pure-state) threshold driving rate
defined in Eq.~\eqref{eq: threshold driving rate at zero temperature}, and
\(f(\beta)\) captures the finite-temperature dependence.
In the remainder of this section, we analyze the scaling of \(f(\beta)\) as a function of the inverse temperature \(\beta\).
Our main result is summarized in the following theorem (see also Table~\ref{tab:driving-rate-scaling}).

\begin{theorem}[Temperature scaling of \(\Gamma^{\,}_{\mathrm{th}}\)]
\label{thm:GammaThTempScaling}
Consider a broad class of local Hamiltonians in gapped phases, i.e., Hamiltonians that can be written as sums of local operators and possess a nonzero spectral gap above the ground state.
In the thermodynamic limit \(N\to\infty\), the threshold driving rate $\Gamma^{\,}_{\mathrm{th}}$ factorizes as in Eq.~\eqref{eq: factorization of threshold driving rate}, with the temperature-dependent factor \(f(\beta)\) obeying the asymptotic scaling forms
\begin{subequations}
\label{eq: model-independent scaling}
\begin{align}
&\text{Low-temperature regime:}\quad 
f(\beta)\simeq 1+c^{\,}_{1}e^{-\beta\Delta}, 
\label{eq: model-independent scaling low temp}
\\ 
&\text{High-temperature regime:}\quad 
f(\beta)\simeq c^{\,}_{2}/\beta.
\label{eq: model-independent scaling high temp}
\end{align}
\end{subequations}
Here, \(c^{\,}_{1}\in(0,2]\) is a model-dependent dimensionless constant, whereas \(c^{\,}_{2}>0\) is a model-dependent constant with dimensions of inverse energy.
Moreover, $\Delta$ denotes the smallest excitation energy among excited eigenstates that couple to the ground state via $\hat V$.
\end{theorem}

\begin{proof}[Proof (sketch).]
The low-temperature scaling form \eqref{eq: model-independent scaling low temp} is obtained by retaining the ground state together with the lowest excited sector that couples to it through $\hat V$, and by showing that $c^{}_{1} \in (0, 2]$ in the thermodynamic limit $N \to \infty$.
The high-temperature scaling form \eqref{eq: model-independent scaling high temp} follows from an expansion around the infinite-temperature (maximally mixed) state and showing that \(c^{\,}_{2}\) is positive and finite in the thermodynamic limit.
A complete proof is given in App.~\ref{sec:proof_of_theorem_ref_thm_gammathtempscaling}.
\end{proof}

These universal scaling forms [Eq.~\eqref{eq: model-independent scaling}] are consistent with the following physical picture.
In the infinite-temperature limit \(\beta\to 0\), the initial Gibbs state \eqref{eq: initial Gibbs state} approaches the maximally mixed state, which commutes with any Hamiltonian.
Therefore, under unitary evolution the dynamical state remains equal to the initial one, \(\hat{\rho}^{\,}_{\lambda}=\hat{\rho}^{\,}_{0}\) for all \(\lambda\).
At the same time, the quasi-Gibbs state \eqref{eq: quasi-Gibbs state} also tends to the maximally mixed state, so that \(\mathcal{F}(\lambda)=\mathcal{C}(\lambda)=1\) as \(\beta\to 0\).
Thus, in the infinite-temperature limit, both the adiabatic fidelity \(\mathcal{F}(\lambda)\) and the thermal-state overlap \(\mathcal{C}(\lambda)\) saturate their maximal value, independent of the driving rate, and the threshold \(\Gamma^{\,}_{\mathrm{th}}\) can be taken arbitrarily large.
For small but finite \(\beta\), deviations of the Boltzmann weights from their infinite-temperature values are of order \(\beta\), so one expects the maximal admissible driving rate to scale inversely with this parameter, consistent with the high-temperature behavior \(f(\beta)\propto \beta^{-1}\) in Eq.~\eqref{eq: model-independent scaling high temp}.
In other words, at high temperature the quasi-Gibbs target state becomes nearly maximally mixed and is therefore much less sensitive to the details of the driving.

In the opposite, zero-temperature limit \(\beta\to\infty\), the initial Gibbs state \eqref{eq: initial Gibbs state} and the quasi-Gibbs state \eqref{eq: quasi-Gibbs state} reduce to the initial and instantaneous ground states, respectively.
Accordingly, the finite-temperature threshold driving rate must reproduce the zero-temperature value, implying \(\lim_{\beta\to\infty} f(\beta)=1\).
In a gapped phase with excitation gap \(\Delta\), contributions of excited states to thermodynamic quantities at low but nonzero temperature are suppressed by Boltzmann factors \(e^{-\beta\Delta}\), so deviations of \(\delta V\) and \(\chi^{\,}_{\mathrm{F}}\) from their ground-state values are likewise exponentially small in \(\beta\Delta\).
It is therefore natural to expect the leading deviation of \(f(\beta)\) from unity at low but nonzero temperature to be proportional to \(e^{-\beta\Delta}\), consistent with Eq.~\eqref{eq: model-independent scaling low temp}.

\begin{figure}[t]
\centering
\begin{tikzpicture}
\begin{axis}[
  width=9cm,
  height=6.2cm,
  xmin=0, xmax=2,
  ymin=0, ymax=5,
  axis lines=box,
  xlabel={$ \beta J $},
  ylabel={},
  tick align=inside,
  minor tick num=1,
  legend style={
    at={(0.98,0.98)},
    anchor=north east,
    draw=none,
    fill=none,
    font=\small,
    row sep=2pt
  },
  samples=500,
  domain=0.001:2
]

\addplot[
  thick
] {cosh(2*x)/sinh(2*x)};
\addlegendentry{$\coth(2\beta J)$}

\addplot[
  color=exporange,
  line width=1.5pt,
  dash pattern=on 0pt off 3.6pt,
  line cap=round
] {1 + 2*exp(-4*x)};
\addlegendentry{$1 + 2e^{-4\beta J}$}

\addplot[
  color=invblue,
  line width=1.0pt,
  dash pattern=on 0pt off 2.0pt,
  line cap=rect
] {1/(2*x)};
\addlegendentry{$\dfrac{1}{2\beta J}$}

\end{axis}
\end{tikzpicture}
\caption{
Temperature-dependent factor $f(\beta)$ in the threshold driving rate $\Gamma^{\,}_{\mathrm{th}}$ [Eqs.~\eqref{eq: define threshold driving rate}, \eqref{eq: factorization of threshold driving rate}]
for the TFIC and QXYC [Eq.~\eqref{eq: V_models_def}] in the thermodynamic limit.
The solid curve shows the exact result $f(\beta)=\coth(2\beta J)$ [Eq.~\eqref{eq: exact Ising temperature dependence}],
while the dotted curves show the low- and high-temperature asymptotics,
$f(\beta)\simeq 1+2e^{-4\beta J}$ (low-temperature) and $f(\beta)\simeq 1/(2\beta J)$ (high-temperature).
}
\label{fig: Ising exact and approximation}
\end{figure}

\section{Spin-chain models as illustrations}
\label{sec:spin_chain_models_as_illustrations}

To interpolate between the low- and high-temperature scaling forms of $f(\beta)$ in Eq.~\eqref{eq: model-independent scaling} and to test our general predictions, we now turn to concrete spin-chain models: the transverse-field Ising chain (TFIC)~\cite{Lieb61,Pfeuty70,Sachdev11} and the quantum XY chain (QXYC)~\cite{Lieb61,Katsura62,Barouch70,Barouch71,Sachdev11,Franchini17}, both with periodic boundary conditions.
Both models take the form~\eqref{eq: define interpolating Hamiltonian} and share the same initial Ising Hamiltonian,
\(
\hat{H}^{\,}_{0} = -J\sum^{N}_{j=1} Z^{\,}_{j}Z^{\,}_{j+1},
\)
where \(J>0\) and \(Z^{\,}_{j}\) is the Pauli-\(Z\) operator acting on site \(j\).
The driving term \(\hat{V}\) differs between the two models:
\begin{subequations}
\label{eq: V_models_def}
\begin{align}
\hat{V}^{\,}_{\mathrm{TFIC}} &= -J\sum^{N}_{j=1} X^{\,}_{j},
\label{eq: V_TFIC_def}
\\
\hat{V}^{\,}_{\mathrm{QXYC}} &= -J\sum^{N}_{j=1}\!\left(X^{\,}_{j}X^{\,}_{j+1}-Z^{\,}_{j}Z^{\,}_{j+1}\right),
\label{eq: V_QXYC_def}
\end{align}
\end{subequations}
where \(X^{\,}_{j}\) is the Pauli-\(X\) operator acting on site \(j\).
For the TFIC we choose \(\lambda(t)=h(t)/J\) with initial transverse field \(h(0)=0\),
while for the QXYC we choose \(\lambda(t)=(1+\gamma(t))/2\) with initial anisotropy \(\gamma(0)=-1\).

Both models map, via a Jordan--Wigner transformation, to quadratic fermionic Hamiltonians and thus admit an analytic treatment.
In the fermionic formulation, finite-\(N\) expressions are typically most transparent as mode products, whereas compact closed forms often emerge only after taking \(N\to\infty\).
Here, we instead use a transfer-matrix method~\cite{Huang1987,Nishimori2010}, which yields closed-form expressions for the quantum fluctuation \(\delta V\) and the fidelity susceptibility \(\chi^{\,}_{\mathrm F}\) at finite \(N\).
Notably, the TFIC and QXYC give identical contributions to \(\delta V\) and \(\chi^{\,}_{\mathrm{F}}\), reflecting the close analogy between the corresponding site-flip and bond-flip terms.
As a result, the threshold driving rate [Eq.~\eqref{eq: define threshold driving rate}] can be written as (see App.~\ref{sec:tfic_qxyc_transfer_matrix}):
\(
\Gamma^{\,}_{\mathrm{th}}=\Gamma^{\,}_{N}\, f^{\,}_{N}(\beta),
\)
where \(\Gamma^{\,}_{N}:=\Gamma^{\,}_{\mathrm{th}}(\beta\to\infty)=4\sqrt{2}J\alpha/\sqrt{N}\) is the zero-temperature threshold driving rate and
\begin{align}
f^{\,}_{N}(\beta)=
\coth(2\beta J)\,
\left(
\frac{1+\tanh^N(2\beta J)}{1+\tanh^{N-2}(2\beta J)}
\right)^{1/2}
\label{eq: temp scaling finite N TFIC}
\end{align}
is the finite-temperature correction factor.
Since \(\tanh(2\beta J)<1\) for any fixed \(\beta>0\), taking \(N\to\infty\) at fixed \(\beta\) yields
\begin{align}
f(\beta):=\lim_{N\to\infty} f^{\,}_{N}(\beta)=\coth(2\beta J).
\label{eq: exact Ising temperature dependence}
\end{align}
Equation~\eqref{eq: exact Ising temperature dependence} further implies the low- and high-temperature expansions
\(f(\beta)=1+2e^{-4\beta J}+\cdots\) and \(f(\beta)=1/(2\beta J)+\cdots\), respectively.
This confirms the universal scaling forms in Eq.~\eqref{eq: model-independent scaling}, with excitation gap \(\Delta=4J\) (the energy cost of creating a pair of domain walls) and coefficients \(c^{\,}_{1}=2\) and \(c^{\,}_{2}=1/(2J)\).
A comparison between the exact factor \(f(\beta)\) [Eq.~\eqref{eq: exact Ising temperature dependence}] and its low- and high-temperature asymptotics is shown in Fig.~\ref{fig: Ising exact and approximation}.
Notably, the crossover regime not captured by either asymptotic expansion occurs in a relatively narrow temperature window.

While Eq.~\eqref{eq: exact Ising temperature dependence} implies that, for the two exactly solvable models (TFIC and QXYC), the temperature-dependent factor \(f(\beta)\) is monotonic for all \(\beta>0\), a further study (see App.~\ref{sec:mfic_transfer_matrix}) of a non-integrable model---the mixed-field Ising chain (MFIC)~\cite{Fogedby78,Sen00,Ovchinnikov03,Banuls11,Chiba24}---shows that this monotonicity need not hold in general.
The MFIC is defined by
\[
\hat{H}^{\,}_{0}=
\sum^{N}_{j=1}
\left(-JZ^{\,}_{j}Z^{\,}_{j+1}+B Z^{\,}_{j}\right),
\quad
\hat{V}^{\,}_{\mathrm{MFIC}}=-J\sum^{N}_{j=1}X^{\,}_{j},
\]
with \(\lambda(t)=h(t)/J\).
For arbitrary \(|B|/J\), the corresponding \(f(\beta)\) is not guaranteed to remain monotonic at intermediate temperatures.
Taken together, these three models verify the predicted low- and high-temperature scaling forms of \(f(\beta)\) in Eq.~\eqref{eq: model-independent scaling}, with \(f(\beta)\to 1\) as \(\beta\to\infty\) and \(f(\beta)\propto 1/\beta\) as \(\beta\to 0\).

Finally, as an additional application of the fidelity bounds in Eqs.~\eqref{eq: fidelity bound}, we show that they can be used to estimate the adiabatic fidelity $\mathcal{F}(\lambda)$ from the thermal-state overlap $\mathcal{C}(\lambda)$ and the quantum-speed-limit integral $\mathcal{R}(\lambda)$ [Eq.~\eqref{eq: QSL mixed states simplified}], without explicitly solving the unitary dynamics for $\hat{\rho}^{\,}_{\lambda}$, which is typically computationally costly and may be intractable for generic many-body systems.
We illustrate this idea using the TFIC model.
Figure~\ref{fig: fidelity bound} shows results for $N=10^{3}$ and $10^{4}$ under a linear ramp $h(t)=2J^2t$.
Throughout the evolution, the adiabatic fidelity $\mathcal{F}(\lambda)$ (cyan) is essentially coincident with the thermal-state overlap $\mathcal{C}(\lambda)$ (black), 
a phenomenon attributed to ``almost-orthogonality'' in large Hilbert spaces~\cite{Chen:2022qyb}.
The blue (resp., red) shaded region indicates the range of $\mathcal{F}(\lambda)$ allowed by inequality~\eqref{eq: fidelity bound a} (resp., inequality~\eqref{eq: fidelity bound b}).
These finite-temperature bounds are quantitatively similar to their zero-temperature counterparts~\cite{Lychkovskiy:2016pfb,Chen:2021gbb,Chen:2022qyb}, thereby 
showing that constraints on adiabatic fidelity for pure and mixed states can be formulated within a unified framework.

\begin{figure}[t]
    \centering
    \includegraphics[width=0.5\textwidth]{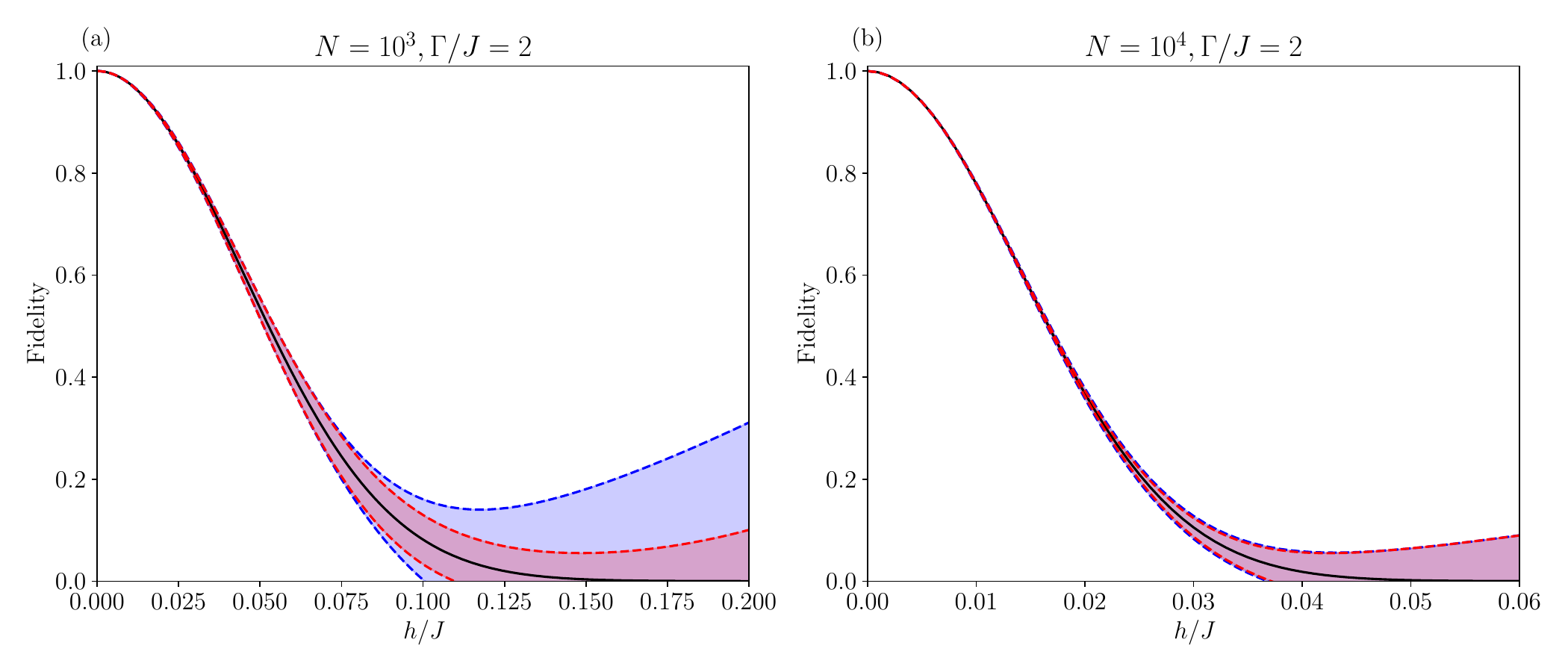}
    \caption{
    Adiabatic fidelity $\mathcal{F}(\lambda)$ [Eq.~\eqref{eq:adiabatic_fidelity_defined}] (cyan curve) and thermal-state overlap $\mathcal{C}(\lambda)$ [Eq.~\eqref{eq: thermal state overlap defined}] (black curve) for the driven transverse-field Ising chain
    $\hat{H}^{\,}_{\lambda}=\hat{H}^{\,}_{0}+\lambda \hat{V}^{\,}_{\mathrm{TFIC}}$ \eqref{eq: V_TFIC_def} at $\beta J=5$ and $\Gamma/J=2$, plotted as a function of $\lambda=h/J$.
    Panels (a) and (b) correspond to $N=10^{3}$ and $N=10^{4}$, respectively.
    Over the range shown, $\mathcal{F}(\lambda)$ and $\mathcal{C}(\lambda)$ are visually indistinguishable.
    The blue (resp., red) shaded region indicates the ranges of $\mathcal{F}(\lambda)$ allowed by inequality~\eqref{eq: fidelity bound a} (resp., \eqref{eq: fidelity bound b}).
    }
    \label{fig: fidelity bound}
\end{figure}

\section{Summary and outlook}
\label{sec:summary_and_outlook}

In summary, we developed a Liouville-space framework that combines a mixed-state quantum speed limit with fidelity susceptibility to derive explicit bounds on the mixed-state adiabatic fidelity. 
For protocols that start from a Gibbs state and drive toward a quasi-Gibbs target, these bounds yield a threshold driving rate \(\Gamma^{\,}_{\mathrm{th}}\) beyond which adiabatic following breaks down. 
For local Hamiltonians in gapped phases, we showed that \(\Gamma^{\,}_{\mathrm{th}}\) factorizes into the familiar zero-temperature system-size dependence and a universal temperature-dependent factor. 
The latter is exponentially close to unity at low temperature and scales linearly with temperature at high temperature. 
We validated this scaling in representative spin-\(1/2\) chains, for which \(\Gamma^{\,}_{\mathrm{th}}\) can be obtained in closed form. 
Our results provide a broadly applicable quantitative criterion for finite-temperature adiabaticity in driven many-body systems under unitary evolution, and offer a useful benchmark for assessing how the adiabatic window changes with temperature in finite-temperature experiments.

A natural next step is to extend the present framework to open systems governed by Lindblad dynamics. Such an extension could provide a rigorous tool for assessing adiabaticity in the presence of decoherence, with applications to adiabatic preparation of both ground and thermal states~\cite{granet2025,shirai2025}  in quantum simulators and quantum annealers.

\acknowledgments{
This work was funded by the National Science and Technology Council (NSTC) of Taiwan under Grant No.\ 113-2112-M-008-037-MY3.
}

\bibliographystyle{quantum-inits}
\bibliography{ref}

@article{Chen:2021gbb,
    author = "Chen, Jyong-Hao and Cheianov, Vadim",
    title = "{Bounds on quantum adiabaticity in driven many-body systems from generalized orthogonality catastrophe and quantum speed limit}",
    doi = "10.1103/PhysRevResearch.4.043055",
    journal = "Phys. Rev. Res.",
    volume = "4",
    number = "4",
    pages = "043055",
    year = "2022"
}

@article{Lychkovskiy:2016pfb,
    author = "Lychkovskiy, Oleg and Gamayun, Oleksandr and Cheianov, Vadim",
    title = "{Time scale for adiabaticity breakdown in driven many-body systems and orthogonality catastrophe}",
    archivePrefix = "arXiv",
    primaryClass = "cond-mat.quant-gas",
    doi = "10.1103/PhysRevLett.119.200401",
    journal = "Phys. Rev. Lett.",
    volume = "119",
    number = "20",
    pages = "200401",
    year = "2017",
    note = "[Erratum: Phys.Rev.Lett. 129, 119902 (2022)]"
}

@article{Chen:2022qyb,
   title={{Quantum adiabaticity in many-body systems and almost-orthogonality in complementary subspace}},
   author={Jyong-Hao Chen and Vadim Cheianov},
   journal={SciPost Phys. Core},
   volume={8},
   pages={084},
   year={2025},
   publisher={SciPost},
   doi={10.21468/SciPostPhysCore.8.4.084},
   url={https://scipost.org/10.21468/SciPostPhysCore.8.4.084},
}

@article{Uzdin16,
   title={Speed limits in Liouville space for open quantum systems},
   volume={115},
   ISSN={1286-4854},
   url={http://dx.doi.org/10.1209/0295-5075/115/40003},
   DOI={10.1209/0295-5075/115/40003},
   number={4},
   journal={EPL (Europhysics Letters)},
   publisher={IOP Publishing},
   author={Uzdin, Raam and Kosloff, Ronnie},
   year={2016},
   month=aug, pages={40003} }

@article{Funo19,
   title={Speed limit for open quantum systems},
   volume={21},
   ISSN={1367-2630},
   url={http://dx.doi.org/10.1088/1367-2630/aaf9f5},
   DOI={10.1088/1367-2630/aaf9f5},
   number={1},
   journal={New Journal of Physics},
   publisher={IOP Publishing},
   author={Funo, Ken and Shiraishi, Naoto and Saito, Keiji},
   year={2019},
   month=jan, pages={013006} }

@article{Ilin21b,
  author    = {Il'in, Nikolai and Lychkovskiy, Oleg},
  title     = {Quantum speed limit for thermal states},
  journal   = {Physical Review A},
  volume    = {103},
  number    = {6},
  pages     = {062204},
  year      = {2021},
  month     = jun,
  doi       = {10.1103/PhysRevA.103.062204},
  url       = {https://doi.org/10.1103/PhysRevA.103.062204},
  publisher = {American Physical Society}
}

@article{Srivastav25,
  title = {Family of exact and inexact quantum speed limits for completely positive and trace-preserving dynamics},
  author = {Srivastav, Abhay and Pandey, Vivek and Mohan, Brij and Pati, Arun Kumar},
  journal = {Phys. Rev. A},
  volume = {112},
  issue = {5},
  pages = {052204},
  numpages = {18},
  year = {2025},
  month = {Nov},
  publisher = {American Physical Society},
  doi = {10.1103/npkr-c4vf},
  url = {https://link.aps.org/doi/10.1103/npkr-c4vf}
}

@article{Ilin_2021,
  author    = {Il'in, Nikolai and Aristova, Anastasia and Lychkovskiy, Oleg},
  title     = {Adiabatic theorem for closed quantum systems initialized at finite temperature},
  journal   = {Physical Review A},
  volume    = {104},
  number    = {3},
  pages     = {L030202},
  year      = {2021},
  month     = sep,
  doi       = {10.1103/PhysRevA.104.L030202},
  url       = {https://doi.org/10.1103/PhysRevA.104.L030202},
  publisher = {American Physical Society}
}

@article{Skelt_2020,
  author    = {Skelt, Amy H. and D'Amico, Irene},
  title     = {Characterizing Adiabaticity in Quantum Many-Body Systems at Finite Temperature},
  journal   = {Advanced Quantum Technologies},
  volume    = {3},
  number    = {7},
  pages     = {1900139},
  year      = {2020},
  month     = may,
  doi       = {10.1002/qute.201900139},
  url       = {https://doi.org/10.1002/qute.201900139},
  publisher = {Wiley}
}

@article{Pfeuty70,
  author  = {Pfeuty, Pierre},
  title   = {The one-dimensional Ising model with a transverse field},
  journal = {Annals of Physics},
  volume  = {57},
  number  = {1},
  pages   = {79},
  year    = {1970},
  doi     = {10.1016/0003-4916(70)90270-8}
}

@article{Lieb61,
  author  = {Lieb, Elliott and Schultz, Theodore and Mattis, Daniel},
  title   = {Two soluble models of an antiferromagnetic chain},
  journal = {Annals of Physics},
  volume  = {16},
  number  = {3},
  pages   = {407},
  year    = {1961},
  doi     = {10.1016/0003-4916(61)90115-4}
}

@book{Sachdev11,
  author    = {Sachdev, Subir},
  title     = {Quantum Phase Transitions},
  edition   = {2},
  publisher = {Cambridge University Press},
  year      = {2011},
  doi       = {10.1017/CBO9780511973765}
}

@book{Ernst87,
  author    = {Ernst, Richard R. and Bodenhausen, Geoffrey and Wokaun, Alexander},
  title     = {Principles of Nuclear Magnetic Resonance in One and Two Dimensions},
  publisher = {Oxford University Press},
  year      = {1987},
  address   = {Oxford},
  isbn      = {978-0198556473},
  nolink  = {}
}

@book{Fick90,
  author    = {Fick, Eugen and Sauermann, G{\"u}nter},
  title     = {The Quantum Statistics of Dynamic Processes},
  series    = {Springer Series in Solid-State Sciences},
  volume    = {86},
  publisher = {Springer-Verlag},
  address   = {Berlin and Heidelberg},
  year      = {1990},
  isbn      = {978-3-642-83717-3},
  note      = {Translated by W. D. Brewer},
  nolink  = {}
}

@article{Gyamfi_2020,
   title={{Fundamentals of quantum mechanics in Liouville space}},
   volume={41},
   ISSN={1361-6404},
   url={http://dx.doi.org/10.1088/1361-6404/ab9fdd},
   DOI={10.1088/1361-6404/ab9fdd},
   number={6},
   journal={European Journal of Physics},
   publisher={IOP Publishing},
   author={Gyamfi, Jerryman A},
   year={2020},
   month=oct, pages={063002} }

@article{Wang_2008,
   title={An alternative quantum fidelity for mixed states of qudits},
   volume={373},
   ISSN={0375-9601},
   url={http://dx.doi.org/10.1016/j.physleta.2008.10.083},
   DOI={10.1016/j.physleta.2008.10.083},
   number={1},
   journal={Physics Letters A},
   publisher={Elsevier BV},
   author={Wang, Xiaoguang and Yu, Chang-Shui and Yi, X.X.},
   year={2008},
   month=dec, pages={58} }

@article{GU_2010,
  author    = {Gu, Shi-Jian},
  title     = {Fidelity approach to quantum phase transitions},
  journal   = {International Journal of Modern Physics B},
  volume    = {24},
  number    = {23},
  pages     = {4371},
  year      = {2010},
  month     = sep,
  doi       = {10.1142/S0217979210056335},
  url       = {https://doi.org/10.1142/S0217979210056335},
  publisher = {World Scientific}
}

@article{Zanardi_2006,
  author    = {Zanardi, Paolo and Paunkovi{\'c}, Nikola},
  title     = {Ground state overlap and quantum phase transitions},
  journal   = {Physical Review E},
  volume    = {74},
  number    = {3},
  pages     = {031123},
  year      = {2006},
  month     = sep,
  doi       = {10.1103/PhysRevE.74.031123},
  url       = {https://doi.org/10.1103/PhysRevE.74.031123},
  publisher = {American Physical Society}
}

@article{Zanardi_2007,
  author    = {Zanardi, Paolo and Quan, H. T. and Wang, Xiaoguang and Sun, C. P.},
  title     = {Mixed-state fidelity and quantum criticality at finite temperature},
  journal   = {Physical Review A},
  volume    = {75},
  number    = {3},
  pages     = {032109},
  year      = {2007},
  month     = mar,
  doi       = {10.1103/PhysRevA.75.032109},
  url       = {https://doi.org/10.1103/PhysRevA.75.032109},
  publisher = {American Physical Society}
}

@article{You_2007,
  author    = {You, Wen-Long and Li, Ying-Wai and Gu, Shi-Jian},
  title     = {Fidelity, dynamic structure factor, and susceptibility in critical phenomena},
  journal   = {Physical Review E},
  volume    = {76},
  number    = {2},
  pages     = {022101},
  year      = {2007},
  month     = aug,
  doi       = {10.1103/PhysRevE.76.022101},
  url       = {https://doi.org/10.1103/PhysRevE.76.022101},
  publisher = {American Physical Society}
}

@article{uhlmann76,
  author    = {Uhlmann, Armin},
  title     = {The ``transition probability'' in the state space of a *-algebra},
  journal   = {Reports on Mathematical Physics},
  volume    = {9},
  number    = {2},
  pages     = {273},
  year      = {1976},
  month     = apr,
  doi       = {10.1016/0034-4877(76)90060-4},
  url       = {https://doi.org/10.1016/0034-4877(76)90060-4},
  publisher = {Elsevier}
}

@article{Holevo73,
  author  = {Holevo, A. S.},
  title   = {Statistical decision theory for quantum systems},
  journal = {Journal of Multivariate Analysis},
  volume  = {3},
  number  = {4},
  pages   = {337},
  year    = {1973},
  doi     = {10.1016/0047-259X(73)90028-6}
}

@book{Helstrom76,
  author    = {Helstrom, Carl W.},
  title     = {Quantum Detection and Estimation Theory},
  publisher = {Academic Press},
  address   = {New York},
  year      = {1976},
  isbn      = {0123400503},
  nolink  = {}
}

@article{FuchsVanDeGraaf99,
  author  = {Fuchs, Christopher A. and van de Graaf, Jeroen},
  title   = {Cryptographic distinguishability measures for quantum-mechanical states},
  journal = {IEEE Transactions on Information Theory},
  volume  = {45},
  number  = {4},
  pages   = {1216},
  year    = {1999},
  doi     = {10.1109/18.761271}
}

@Article{Anderson1967Infrared,
  author    = {Anderson, P. W.},
  journal   = {Phys. Rev. Lett.},
  title     = {{Infrared Catastrophe in Fermi Gases with Local Scattering Potentials}},
  year      = {1967},
  month     = {Jun},
  pages     = {1049},
  volume    = {18},
  doi       = {10.1103/PhysRevLett.18.1049},
  issue     = {24},
  numpages  = {0},
  publisher = {American Physical Society},
  url       = {https://link.aps.org/doi/10.1103/PhysRevLett.18.1049},
}

@article{Greenblatt_2024,
  author    = {Greenblatt, Rafael L. and Lange, Markus and Marcelli, Giovanna and Porta, Marcello},
  title     = {Adiabatic Evolution of Low-Temperature Many-Body Systems},
  journal   = {Communications in Mathematical Physics},
  volume    = {405},
  number    = {3},
  pages     = {75},
  year      = {2024},
  month     = mar,
  doi       = {10.1007/s00220-023-04903-6},
  url       = {https://doi.org/10.1007/s00220-023-04903-6},
  publisher = {Springer Science and Business Media LLC}
}

@article{Jozsa94,
  title={Fidelity for mixed quantum states},
  author={Jozsa, Richard},
  journal={Journal of Modern Optics},
  volume={41},
  number={12},
  pages={2315},
  year={1994},
  publisher={Taylor \& Francis},
  doi={10.1080/09500349414552171},
  url={https://www.tandfonline.com/doi/abs/10.1080/09500349414552171}
}

@book{Nielsen00,
  title={Quantum Computation and Quantum Information},
  author={Nielsen, Michael A. and Chuang, Isaac L.},
  year={2000},
  publisher={Cambridge University Press},
  address={Cambridge},
  nolink  = {}
}

@book{Wilde17,
  title={Quantum Information Theory},
  author={Wilde, Mark M.},
  year={2017},
  publisher={Cambridge University Press},
  address={Cambridge},
  edition={2nd},
  isbn={9781107176164},
  doi={10.1017/9781316809976}
}

@article{Liang19,
doi = {10.1088/1361-6633/ab1ca4},
url = {https://doi.org/10.1088/1361-6633/ab1ca4},
year = {2019},
month = {jun},
publisher = {IOP Publishing},
volume = {82},
number = {7},
pages = {076001},
author = {Liang, Yeong-Cherng and Yeh, Yu-Hao and Mendonça, Paulo E M F and Teh, Run Yan and Reid, Margaret D and Drummond, Peter D},
title = {Quantum fidelity measures for mixed states},
journal = {Reports on Progress in Physics}
}

@article{BolonekLason2021,
  doi = {10.22331/q-2021-06-24-482},
  url = {https://doi.org/10.22331/q-2021-06-24-482},
  title = {Classical and quantum speed limits},
  author = {Bolonek-Laso{\'{n}}, Katarzyna and Gonera, Joanna and Kosi{\'{n}}ski, Piotr},
  journal = {{Quantum}},
  issn = {2521-327X},
  publisher = {{Verein zur F{\"{o}}rderung des Open Access Publizierens in den Quantenwissenschaften}},
  volume = {5},
  pages = {482},
  month = jun,
  year = {2021}
}

@misc{granet2025,
      title={Adiabatic preparation of thermal states and entropy-noise relation on noisy quantum computers}, 
      author={Etienne Granet and Henrik Dreyer},
      year={2025},
      eprint={2509.05206},
      archivePrefix={arXiv},
      primaryClass={quant-ph},
      url={https://arxiv.org/abs/2509.05206}, 
}

@article{Irmejs2026,
   title={Quasi-Adiabatic Processing of Thermal States},
   volume={10},
   ISSN={2521-327X},
   url={http://dx.doi.org/10.22331/q-2026-03-10-2018},
   DOI={10.22331/q-2026-03-10-2018},
   journal={Quantum},
   publisher={Verein zur Forderung des Open Access Publizierens in den Quantenwissenschaften},
   author={Irmejs, Reinis and Bañuls, Mari Carmen and Cirac, J. Ignacio},
   year={2026},
   month=mar, pages={2018} }

@misc{shirai2025,
      title={Quasi-adiabatic thermal ensemble preparation in the thermodynamic limit}, 
      author={Tatsuhiko Shirai},
      year={2025},
      eprint={2510.13555},
      archivePrefix={arXiv},
      primaryClass={cond-mat.stat-mech},
      url={https://arxiv.org/abs/2510.13555}, 
}

@article{Hastings11,
  title = {Topological Order at Nonzero Temperature},
  author = {Hastings, Matthew B.},
  journal = {Phys. Rev. Lett.},
  volume = {107},
  issue = {21},
  pages = {210501},
  numpages = {5},
  year = {2011},
  month = {Nov},
  publisher = {American Physical Society},
  doi = {10.1103/PhysRevLett.107.210501},
  url = {https://link.aps.org/doi/10.1103/PhysRevLett.107.210501}
}

@article{Unanyan20,
  title = {Finite-Temperature Topological Invariant for Interacting Systems},
  author = {Unanyan, Razmik and Kiefer-Emmanouilidis, Maximilian and Fleischhauer, Michael},
  journal = {Phys. Rev. Lett.},
  volume = {125},
  issue = {21},
  pages = {215701},
  numpages = {6},
  year = {2020},
  month = {Nov},
  publisher = {American Physical Society},
  doi = {10.1103/PhysRevLett.125.215701},
  url = {https://link.aps.org/doi/10.1103/PhysRevLett.125.215701}
}

@article{Zhou25,
  title = {Finite-Temperature Quantum Topological Order in Three Dimensions},
  author = {Zhou, Shu-Tong and Cheng, Meng and Rakovszky, Tibor and von Keyserlingk, Curt and Ellison, Tyler D.},
  journal = {Phys. Rev. Lett.},
  volume = {135},
  issue = {4},
  pages = {040402},
  numpages = {7},
  year = {2025},
  month = {Jul},
  publisher = {American Physical Society},
  doi = {10.1103/PhysRevLett.135.040402},
  url = {https://link.aps.org/doi/10.1103/PhysRevLett.135.040402}
}

@article{Katsura62,
  author  = {Katsura, Shigetoshi},
  title   = {Statistical Mechanics of the Anisotropic Linear Heisenberg Model},
  journal = {Phys. Rev.},
  volume  = {127},
  pages   = {1508},
  year    = {1962},
  doi     = {10.1103/PhysRev.127.1508}
}

@article{Barouch70,
  author  = {Barouch, Eytan and McCoy, Barry M. and Dresden, Max},
  title   = {Statistical Mechanics of the {X Y} Model. {I}},
  journal = {Phys. Rev. A},
  volume  = {2},
  pages   = {1075},
  year    = {1970},
  doi     = {10.1103/PhysRevA.2.1075}
}

@article{Barouch71,
  author  = {Barouch, Eytan and McCoy, Barry M.},
  title   = {Statistical Mechanics of the {X Y} Model. {II}. Spin-Correlation Functions},
  journal = {Phys. Rev. A},
  volume  = {3},
  pages   = {786},
  year    = {1971},
  doi     = {10.1103/PhysRevA.3.786}
}

@book{Franchini17,
  author    = {Franchini, Fabio},
  title     = {An Introduction to Integrable Techniques for One-Dimensional Quantum Systems},
  publisher = {Springer},
  year      = {2017},
  doi       = {10.1007/978-3-319-48487-7}
}

@article{Fogedby78,
  author  = {Fogedby, Hans C.},
  title   = {The Ising chain in a skew magnetic field},
  journal = {Journal of Physics C: Solid State Physics},
  volume  = {11},
  number  = {13},
  pages   = {2801},
  year    = {1978},
  doi     = {10.1088/0022-3719/11/13/025}
}

@article{Sen00,
  title = {Quantum phase transitions in the Ising model in a spatially modulated field},
  author = {Sen, Parongama},
  journal = {Phys. Rev. E},
  volume = {63},
  issue = {1},
  pages = {016112},
  numpages = {4},
  year = {2000},
  month = {Dec},
  publisher = {American Physical Society},
  doi = {10.1103/PhysRevE.63.016112},
  url = {https://link.aps.org/doi/10.1103/PhysRevE.63.016112}
}

@article{Ovchinnikov03,
  author  = {Ovchinnikov, A. A. and Dmitriev, D. V. and Krivnov, V. Ya. and Cheranovskii, V. O.},
  title   = {Antiferromagnetic Ising chain in a mixed transverse and longitudinal magnetic field},
  journal = {Phys. Rev. B},
  volume  = {68},
  pages   = {214406},
  year    = {2003},
  doi     = {10.1103/PhysRevB.68.214406}
}

@article{Banuls11,
  author  = {Ba{\~n}uls, Mari-Carmen and Cirac, J. Ignacio and Hastings, M. B.},
  title   = {Strong and weak thermalization of infinite nonintegrable quantum systems},
  journal = {Phys. Rev. Lett.},
  volume  = {106},
  pages   = {050405},
  year    = {2011},
  doi     = {10.1103/PhysRevLett.106.050405}
}

@article{Chiba24,
  title = {Proof of absence of local conserved quantities in the mixed-field Ising chain},
  author = {Chiba, Yuuya},
  journal = {Phys. Rev. B},
  volume = {109},
  issue = {3},
  pages = {035123},
  numpages = {15},
  year = {2024},
  month = {Jan},
  publisher = {American Physical Society},
  doi = {10.1103/PhysRevB.109.035123},
  url = {https://link.aps.org/doi/10.1103/PhysRevB.109.035123}
}

@article{Lieb1972,
  title        = {The finite group velocity of quantum spin systems},
  author       = {Lieb, Elliott H. and Robinson, Derek W.},
  journal      = {Communications in Mathematical Physics},
  volume       = {28},
  number       = {3},
  pages        = {251},
  year         = {1972},
  doi          = {10.1007/BF01645779}
}

@article{Hastings2006,
  title        = {Spectral Gap and Exponential Decay of Correlations},
  author       = {Hastings, Matthew B. and Koma, Tohru},
  journal      = {Communications in Mathematical Physics},
  volume       = {265},
  number       = {3},
  pages        = {781},
  year         = {2006},
  doi          = {10.1007/s00220-006-0030-4},
  url          = {https://doi.org/10.1007/s00220-006-0030-4}
}

@article{Nachtergaele2006,
  title        = {Propagation of Correlations in Quantum Lattice Systems},
  author       = {Nachtergaele, Bruno and Ogata, Yoshiko and Sims, Robert},
  journal      = {Journal of Statistical Physics},
  volume       = {124},
  number       = {1},
  pages        = {1},
  year         = {2006},
  doi          = {10.1007/s10955-006-9143-6},
  url          = {https://doi.org/10.1007/s10955-006-9143-6}
}

@article{Nachtergaele2006_CMP,
  title        = {{Lieb--Robinson} Bounds and the Exponential Clustering Theorem},
  author       = {Nachtergaele, Bruno and Sims, Robert},
  journal      = {Communications in Mathematical Physics},
  volume       = {265},
  number       = {1},
  pages        = {119},
  year         = {2006},
  doi          = {10.1007/s00220-006-1556-1},
  url          = {https://doi.org/10.1007/s00220-006-1556-1}
}

@article{Naudts2005,
  author       = {Naudts, Jan},
  title        = {Escort Density Operators and Generalized Quantum Information Measures},
  journal      = {Open Systems \& Information Dynamics},
  volume       = {12},
  pages        = {13},
  year         = {2005},
  doi          = {10.1007/s11080-005-0483-5},
  url          = {https://doi.org/10.1007/s11080-005-0483-5}
}

@article{Wigner1963,
  author       = {Wigner, E. P. and Yanase, M. M.},
  title        = {Information Contents of Distributions},
  journal      = {Proceedings of the National Academy of Sciences of the United States of America},
  volume       = {49},
  number       = {6},
  pages        = {910},
  year         = {1963},
  doi          = {10.1073/pnas.49.6.910},
  url          = {https://doi.org/10.1073/pnas.49.6.910}
}

@book{Ashcroft1976,
  title        = {Solid State Physics},
  author       = {Ashcroft, Neil W. and Mermin, N. David},
  publisher    = {Holt, Rinehart and Winston},
  address      = {New York},
  year         = {1976},
  isbn         = {9780030839931},
  nolink  = {}
}

@book{Nishimori2010,
  author    = {Nishimori, Hidetoshi and Ortiz, Gerardo},
  title     = {Elements of Phase Transitions and Critical Phenomena},
  publisher = {Oxford University Press},
  address   = {Oxford},
  year      = {2010},
  isbn      = {9780199577224},
  doi       = {10.1093/acprof:oso/9780199577224.001.0001},
  url       = {https://doi.org/10.1093/acprof:oso/9780199577224.001.0001}
}

@book{Huang1987,
  author    = {Huang, Kerson},
  title     = {Statistical Mechanics},
  edition   = {2nd ed.},
  year      = {1987},
  publisher = {Wiley},
  address   = {New York},
  isbn      = {9780471815181},
nolink  = {}
}

@article{Born27,
       author = {{Born}, Max},
        title = "{Das Adiabatenprinzip in der Quantenmechanik}",
      journal = {Zeitschrift fur Physik},
         year = 1927,
        month = mar,
       volume = {40},
       number = {3-4},
        pages = {167},
          doi = {10.1007/BF01400360},
       adsurl = {https://ui.adsabs.harvard.edu/abs/1927ZPhy...40..167B}
}

@article{Born28,
       author = {{Born}, M. and {Fock}, V.},
        title = "{Beweis des Adiabatensatzes}",
      journal = {Zeitschrift fur Physik},
         year = 1928,
        month = mar,
       volume = {51},
       number = {3-4},
        pages = {165},
          doi = {10.1007/BF01343193},
       adsurl = {https://ui.adsabs.harvard.edu/abs/1928ZPhy...51..165B}
}

@article{Kato50,
author = {Kato ,Tosio},
title = {On the Adiabatic Theorem of Quantum Mechanics},
journal = {Journal of the Physical Society of Japan},
volume = {5},
number = {6},
pages = {435},
year = {1950},
doi = {10.1143/JPSJ.5.435},
URL = {https://doi.org/10.1143/JPSJ.5.435},
}

@book{Messiah14,
  author    = {Messiah, Albert},
  title     = {Quantum Mechanics},
  series    = {Dover Books on Physics},
  publisher = {Dover Publications},
  year      = {2014},
  isbn      = {9780486784557},
  nolink = {}
}

@article{Nenciu93,
  title     = {Linear adiabatic theory. Exponential estimates},
  author    = {Nenciu, Gheorghe},
  journal   = {Communications in Mathematical Physics},
  volume    = {152},
  number    = {3},
  pages     = {479},
  year      = {1993},
  publisher = {Springer},
  doi       = {10.1007/BF02096616}
}

@ARTICLE{Avron99,
       author = {{Avron}, Joseph E. and {Elgart}, Alexander},
        title = "{Adiabatic Theorem without a Gap Condition}",
      journal = {Communications in Mathematical Physics},
         year = 1999,
        month = jan,
       volume = {203},
       number = {2},
        pages = {445},
          doi = {10.1007/s002200050620},
archivePrefix = {arXiv},
 primaryClass = {math-ph},
       adsurl = {https://ui.adsabs.harvard.edu/abs/1999CMaPh.203..445A}
}

@article{Hagedorn02,
title = {Elementary Exponential Error Estimates for the Adiabatic Approximation},
journal = {Journal of Mathematical Analysis and Applications},
volume = {267},
number = {1},
pages = {235},
year = {2002},
issn = {0022-247X},
doi = {https://doi.org/10.1006/jmaa.2001.7765},
url = {https://www.sciencedirect.com/science/article/pii/S0022247X01977650},
author = {George A. Hagedorn and Alain Joye}
}

@book{Teufel03,
  title={Adiabatic Perturbation Theory in Quantum Dynamics},
  author={Teufel, Stefan},
  year={2003},
  publisher={Springer-Verlag},
  series={Lecture Notes in Mathematics},
  volume={1821},
  address={Berlin, Heidelberg},
  doi={10.1007/b13735}
}

@misc{Ambainis06,
      title={An Elementary Proof of the Quantum Adiabatic Theorem}, 
      author={Andris Ambainis and Oded Regev},
      year={2006},
      eprint={quant-ph/0411152},
      archivePrefix={arXiv},
      primaryClass={quant-ph},
      url={https://arxiv.org/abs/quant-ph/0411152}, 
}

@article{Jansen07,
author = {Jansen,Sabine  and Ruskai,Mary-Beth  and Seiler,Ruedi },
title = {Bounds for the adiabatic approximation with applications to quantum computation},
journal = {Journal of Mathematical Physics},
volume = {48},
number = {10},
pages = {102111},
year = {2007},
doi = {10.1063/1.2798382},
URL = {https://doi.org/10.1063/1.2798382}   
}

@article{Amin09,
  title = {Consistency of the Adiabatic Theorem},
  author = {Amin, M. H. S.},
  journal = {Phys. Rev. Lett.},
  volume = {102},
  issue = {22},
  pages = {220401},
  numpages = {4},
  year = {2009},
  month = {Jun},
  publisher = {American Physical Society},
  doi = {10.1103/PhysRevLett.102.220401},
  url = {https://link.aps.org/doi/10.1103/PhysRevLett.102.220401}
}

@article{Lidar09,
  author    = {Lidar, Daniel A. and Rezakhani, Ali T. and Hamma, Alioscia},
  title     = {Adiabatic approximation with exponential accuracy for many-body systems and quantum computation},
  journal   = {Journal of Mathematical Physics},
  volume    = {50},
  number    = {10},
  pages     = {102106},
  year      = {2009},
  month     = oct,
  doi       = {10.1063/1.3236685},
  url       = {https://doi.org/10.1063/1.3236685},
  publisher = {AIP Publishing}
}

@article{Cheung11,
   title={Improved error bounds for the adiabatic approximation},
   volume={44},
   ISSN={1751-8121},
   url={http://dx.doi.org/10.1088/1751-8113/44/41/415302},
   DOI={10.1088/1751-8113/44/41/415302},
   number={41},
   journal={Journal of Physics A: Mathematical and Theoretical},
   publisher={IOP Publishing},
   author={Cheung, Donny and Høyer, Peter and Wiebe, Nathan},
   year={2011},
   month=sep, pages={415302} }

@article{Elgart12,
  author    = {Elgart, Alexander and Hagedorn, George A.},
  title     = {A note on the switching adiabatic theorem},
  journal   = {Journal of Mathematical Physics},
  volume    = {53},
  number    = {10},
  pages     = {102202},
  year      = {2012},
  month     = sep,
  doi       = {10.1063/1.4748968},
  url       = {https://doi.org/10.1063/1.4748968},
  publisher = {AIP Publishing}
}

@article{Ge16,
  title = {Rapid Adiabatic Preparation of Injective Projected Entangled Pair States and Gibbs States},
  author = {Ge, Yimin and Moln\'ar, Andr\'as and Cirac, J. Ignacio},
  journal = {Phys. Rev. Lett.},
  volume = {116},
  issue = {8},
  pages = {080503},
  numpages = {5},
  year = {2016},
  month = {Feb},
  publisher = {American Physical Society},
  doi = {10.1103/PhysRevLett.116.080503},
  url = {https://link.aps.org/doi/10.1103/PhysRevLett.116.080503}
}

@article{Bachmann17,
  title = {Adiabatic Theorem for Quantum Spin Systems},
  author = {Bachmann, S. and De Roeck, W. and Fraas, M.},
  journal = {Phys. Rev. Lett.},
  volume = {119},
  issue = {6},
  pages = {060201},
  numpages = {6},
  year = {2017},
  month = {Aug},
  publisher = {American Physical Society},
  doi = {10.1103/PhysRevLett.119.060201},
  url = {https://link.aps.org/doi/10.1103/PhysRevLett.119.060201}
}

@article{Albash18,
  title = {Adiabatic quantum computation},
  author = {Albash, Tameem and Lidar, Daniel A.},
  journal = {Rev. Mod. Phys.},
  volume = {90},
  issue = {1},
  pages = {015002},
  numpages = {64},
  year = {2018},
  month = {Jan},
  publisher = {American Physical Society},
  doi = {10.1103/RevModPhys.90.015002},
  url = {https://link.aps.org/doi/10.1103/RevModPhys.90.015002}
}

@article{Bachmann18,
   title={The Adiabatic Theorem and Linear Response Theory for Extended Quantum Systems},
   volume={361},
   ISSN={1432-0916},
   url={http://dx.doi.org/10.1007/s00220-018-3117-9},
   DOI={10.1007/s00220-018-3117-9},
   number={3},
   journal={Communications in Mathematical Physics},
   publisher={Springer Science and Business Media LLC},
   author={Bachmann, Sven and De Roeck, Wojciech and Fraas, Martin},
   year={2018},
   month=mar, pages={997} }

@incollection{Bachmann20,
  author    = {Bachmann, Sven and De Roeck, Wojciech and Fraas, Martin},
  title     = {The adiabatic theorem in a quantum many-body setting},
  booktitle = {Analytic Trends in Mathematical Physics},
  series    = {Contemporary Mathematics},
  volume    = {741},
  pages     = {43},
  year      = {2020},
  publisher = {American Mathematical Society},
  doi       = {10.1090/conm/741/14919}
}

@article{Anandan90,
  title = {Geometry of quantum evolution},
  author = {Anandan, J. and Aharonov, Y.},
  journal = {Phys. Rev. Lett.},
  volume = {65},
  issue = {14},
  pages = {1697},
  numpages = {0},
  year = {1990},
  month = {Oct},
  publisher = {American Physical Society},
  doi = {10.1103/PhysRevLett.65.1697},
  url = {https://link.aps.org/doi/10.1103/PhysRevLett.65.1697}
}

@article{Pfeifer93a,
  title = {How fast can a quantum state change with time?},
  author = {Pfeifer, Peter},
  journal = {Phys. Rev. Lett.},
  volume = {70},
  issue = {22},
  pages = {3365},
  numpages = {0},
  year = {1993},
  month = {May},
  publisher = {American Physical Society},
  doi = {10.1103/PhysRevLett.70.3365},
  url = {https://link.aps.org/doi/10.1103/PhysRevLett.70.3365}
}

@article{Pfeifer93b,
  title = {How Fast Can a Quantum State Change with Time?},
  author = {Pfeifer, Peter},
  journal = {Phys. Rev. Lett.},
  volume = {71},
  issue = {2},
  pages = {306},
  numpages = {0},
  year = {1993},
  month = {Jul},
  publisher = {American Physical Society},
  doi = {10.1103/PhysRevLett.71.306.2},
  url = {https://link.aps.org/doi/10.1103/PhysRevLett.71.306.2}
}

@article{Pfeifer95,
  title = {Generalized time-energy uncertainty relations and bounds on lifetimes of resonances},
  author = {Pfeifer, Peter and Fr\"ohlich, J\"urg},
  journal = {Rev. Mod. Phys.},
  volume = {67},
  issue = {4},
  pages = {759},
  numpages = {0},
  year = {1995},
  month = {Oct},
  publisher = {American Physical Society},
  doi = {10.1103/RevModPhys.67.759},
  url = {https://link.aps.org/doi/10.1103/RevModPhys.67.759}
}

@article{Vaidman92,
author = {Vaidman, Lev},
title = {Minimum time for the evolution to an orthogonal quantum state},
journal = {American Journal of Physics},
volume = {60},
number = {2},
pages = {182},
year = {1992},
doi = {10.1119/1.16940},
URL = {https://doi.org/10.1119/1.16940}
}

@article{Mandelstam45,
author = {Mandelstam, L. and Tamm, Ig.},
title = {The uncertainty relation between energy and time in non-relativistic quantum mechanics},
journal = {J. Phys. USSR},
volume = {9},
pages = {249},
year = {1945},
nolink  = {}
}

@article{Deffner17,
  Author = {Sebastian Deffner and Steve Campbell},
  Doi = {10.1088/1751-8121/aa86c6},
  Journal = {Journal of Physics A: Mathematical and Theoretical},
  Month = {oct},
  Number = {45},
  Pages = {453001},
  Publisher = {{IOP} Publishing},
  Title = {Quantum speed limits: from Heisenberg's uncertainty principle to optimal quantum control},
  Url = {https://doi.org/10.1088/1751-8121/aa86c6},
  Volume = {50},
  Year = {2017}
}

@article{Taddei13,
  author  = {Taddei, M. M. and Escher, B. M. and Davidovich, L. and de Matos Filho, R. L.},
  title   = {Quantum Speed Limit for Physical Processes},
  journal = {Phys. Rev. Lett.},
  volume  = {110},
  pages   = {050402},
  year    = {2013},
  doi     = {10.1103/PhysRevLett.110.050402}
}

@article{Campo13,
  author  = {del Campo, A. and Egusquiza, I. L. and Plenio, M. B. and Huelga, S. F.},
  title   = {Quantum Speed Limits in Open System Dynamics},
  journal = {Phys. Rev. Lett.},
  volume  = {110},
  pages   = {050403},
  year    = {2013},
  doi     = {10.1103/PhysRevLett.110.050403}
}

@article{Deffner13,
  author  = {Deffner, Sebastian and Lutz, Eric},
  title   = {Quantum Speed Limit for Non-Markovian Dynamics},
  journal = {Phys. Rev. Lett.},
  volume  = {111},
  pages   = {010402},
  year    = {2013},
  doi     = {10.1103/PhysRevLett.111.010402}
}

@article{Pires16,
  author  = {Pires, Diego Paiva and Cianciaruso, Marco and C{\'e}leri, Lucas C. and Adesso, Gerardo and Soares-Pinto, Diogo O.},
  title   = {Generalized Geometric Quantum Speed Limits},
  journal = {Phys. Rev. X},
  volume  = {6},
  pages   = {021031},
  year    = {2016},
  doi     = {10.1103/PhysRevX.6.021031}
}

@article{Pires18,
  author  = {Campaioli, Francesco and Pollock, Felix A. and Binder, Felix C. and Modi, Kavan},
  title   = {Tightening Quantum Speed Limits for Almost All States},
  journal = {Phys. Rev. Lett.},
  volume  = {120},
  pages   = {060409},
  year    = {2018},
  doi     = {10.1103/PhysRevLett.120.060409}
}

\clearpage
\onecolumn

\appendix
\inappendixtrue

\renewcommand{\theequation}{A\arabic{equation}}
\setcounter{equation}{0}
\renewcommand{\thefigure}{A\arabic{figure}}
\setcounter{figure}{0}

\begin{center}
\textbf{\LARGE Appendix}
\end{center}

This Appendix provides technical details that support the results in the main text.

\printappendixtoc

\section{Derivation of the mixed-state quantum speed limit [Eq.\ \eqref{eq: QSL mixed states}]}
\label{sec: derivation_mixed_QSL}

In this section, we derive the mixed-state quantum speed limit inequality [Eq.~\eqref{eq: QSL mixed states}] for closed, unitary dynamics.
We parametrize the evolution by a driving coordinate \(\lambda=\lambda(t)\) and assume that the Hamiltonian depends on time only through \(\lambda(t)\), i.e., \(\hat H(t)=\hat H^{\,}_{\lambda(t)}\).

\subsection{Liouville-space setup}

We consider unitary dynamics generated by a Hamiltonian $\hat H_\lambda$,
\begin{align}
  \mathrm{i}\hbar\,\Gamma\,\frac{\partial}{\partial\lambda}\hat\rho^{\,}_\lambda
  =[\hat H^{\,}_\lambda,\hat\rho^{\,}_\lambda],
  \qquad
  \Gamma:=\partial^{\,}_t\lambda .
  \label{eq: Liouville von Neumann eq}
\end{align}
We vectorize an operator $\hat A=\sum_{n,m}A^{\,}_{nm}\,|n\rangle\langle m|$ as
$\HSket{A}=\sum_{n,m}A^{\,}_{nm}\,|n\rangle\otimes|m\rangle$,
equipped with the Hilbert--Schmidt inner product
$\HSinner{A}{B}:=\Tr(\hat A^\dagger \hat B)$.
Then Eq.~\eqref{eq: Liouville von Neumann eq} becomes
\[
  \frac{\partial}{\partial\lambda}\HSket{\rho^{\,}_\lambda}
  =-\frac{\mathrm{i}}{\hbar}\frac{1}{\Gamma}\,\hat{\hat L}^{\,}_\lambda\,\HSket{\rho^{\,}_\lambda},
  \qquad
  \hat{\hat L}^{\,}_\lambda\HSket{\bullet}
  =\HSket{[\hat H^{\,}_\lambda,\bullet]}
  =\left(\hat H^{\,}_\lambda\otimes\mathbb{I}-\mathbb{I}\otimes\hat H_\lambda^{\mathsf T}\right)\HSket{\bullet},
\]
where the superoperator $\hat{\hat L}_\lambda$ is Hermitian,
$\hat{\hat L}_\lambda^\dagger=\hat{\hat L}_\lambda$.

For later convenience we introduce normalized Liouville-space vectors
\[
  \HSket{\Psi^{\,}_\lambda}:=\frac{\HSket{\rho^{\,}_\lambda}}{\HSnorm{\rho^{\,}_\lambda}^{\,}},
  \qquad
  \HSket{\Psi^{\,}_0}:=\frac{\HSket{\rho^{\,}_0}}{\HSnorm{\rho^{\,}_0}^{\,}},
\]
where \(\HSnorm{\rho^{\,}_0}^2=\Tr(\hat\rho_0^2)\).
Note that for closed systems the purity is conserved,
\(\Tr(\hat{\rho}^{2}_\lambda)=\Tr(\hat{\rho}^{2}_0)\), hence \(\HSnorm{\rho^{\,}_\lambda}^{\,}=\HSnorm{\rho^{\,}_0}^{\,}\).

\subsection{Hilbert--Schmidt angle and its rate of change}

Define the Hilbert--Schmidt fidelity amplitude \(M(\lambda)\) and Hilbert--Schmidt angle \(\Theta^{\,}_\lambda\) by
\[
  M(\lambda)
  :=\HSinner{\Psi_\lambda}{\Psi_0}
  =\frac{\Tr(\hat\rho^{\,}_\lambda\hat\rho_0)}{\Tr(\hat\rho_0^2)}\in[0,1],
  \qquad
  \Theta^{\,}_\lambda:=\arccos  M(\lambda)\in[0,\pi/2],
\]
so that \(M(\lambda)=\cos\Theta_\lambda\).
Differentiating \(M(\lambda)\) with respect to \(\lambda\) yields
\begin{align}
  \partial^{\,}_{\lambda}{M}
  =\HSinner{\partial_\lambda\Psi_\lambda}{\Psi_0}
  =\frac{\mathrm{i}}{\hbar}\frac{1}{\Gamma}\,\frac{1}{\Tr(\hat\rho_0^2)}\,
    \HSbra{\rho^{\,}_\lambda}\hat{\hat{L}}^{\,}_\lambda\HSket{\rho^{\,}_0}.
  \label{eq: gdot}
\end{align}

Introduce the rank-one superprojector onto the initial Hilbert--Schmidt state and its complement,
\begin{align}
  \hat{\hat{\mathcal{P}}}:=\HSket{\Psi_0}\HSbra{\Psi_0},
  \qquad
  \hat{\hat{\mathcal{Q}}}:=\openone-\hat{\hat{\mathcal{P}}},
\label{eq: def rank one superprojector}
\end{align}
so that \(\hat{\hat{\mathcal{P}}}^2=\hat{\hat{\mathcal{P}}}\), \(\hat{\hat{\mathcal{Q}}}^2=\hat{\hat{\mathcal{Q}}}\), and
\(\hat{\hat{\mathcal{P}}}\hat{\hat{\mathcal{Q}}}=0\).
Decompose \(\hat{\hat{L}}^{\,}_\lambda\HSket{\rho^{\,}_0}\) in Eq.~\eqref{eq: gdot} as
\[
  \hat{\hat{L}}^{\,}_\lambda\HSket{\rho^{\,}_0}
  =\hat{\hat{\mathcal{P}}}\hat{\hat{L}}^{\,}_\lambda\HSket{\rho^{\,}_0}
   +\hat{\hat{\mathcal{Q}}}\hat{\hat{L}}^{\,}_\lambda\HSket{\rho^{\,}_0}.
\]
The \(\hat{\hat{\mathcal{P}}}\)-term vanishes because
\[
  \HSbra{\Psi_0}\hat{\hat{L}}^{\,}_\lambda\HSket{\rho^{\,}_0}
  =\frac{1}{\HSnorm{\rho^{\,}_0}^{\,}}\HSbra{\rho^{\,}_0}\hat{\hat{L}}^{\,}_\lambda\HSket{\rho^{\,}_0}
  =\frac{1}{\HSnorm{\rho^{\,}_0}^{\,}}\Tr\!\left(\hat\rho^{\,}_0[\hat H^{\,}_\lambda,\hat\rho^{\,}_0]\right)=0,
\]
by cyclicity of the trace.
Hence, Eq.~\eqref{eq: gdot} implies
\[
  |\partial^{\,}_{\lambda}{M}|
  =\frac{1}{\hbar}\frac{1}{|\Gamma|}\frac{1}{\Tr(\hat\rho_0^2)}
   \left|\HSbra{\rho^{\,}_\lambda}\hat{\hat{\mathcal{Q}}}\hat{\hat{L}}^{\,}_\lambda\HSket{\rho^{\,}_0}\right|.
\]
Applying the Cauchy--Schwarz inequality gives
\begin{align}
  |\partial^{\,}_{\lambda}M|
  &\le
  \frac{1}{\hbar}\frac{1}{|\Gamma|}\frac{1}{\Tr(\hat\rho_0^2)}
  \HSnorm{\hat{\hat{\mathcal{Q}}}^{\,}\HSket{\rho^{\,}_\lambda}}\;
  \HSnorm{\hat{\hat{L}}^{\,}_\lambda\HSket{\rho^{\,}_0}}.
  \label{eq: bound after CS}
\end{align}
The two norms evaluate to
\[
  \HSnorm{\hat{\hat{\mathcal{Q}}}^{\,}\HSket{\rho^{\,}_\lambda}}
  =\HSnorm{\rho^{\,}_0}^{\,}\,\sin\Theta_\lambda,\qquad
  \HSnorm{\hat{\hat{L}}^{\,}_\lambda\HSket{\rho^{\,}_0}}
  =\HSnorm{[\hat H^{\,}_\lambda,\hat\rho^{\,}_0]}^{\,}.
\]
Substituting into Eq.~\eqref{eq: bound after CS} yields
\[
  |\partial^{\,}_{\lambda} M|
  \le \frac{1}{\hbar}\frac{1}{|\Gamma|}\,\frac{1}{\HSnorm{\rho^{\,}_0}^{\,}}\,
  \sin\Theta^{\,}_\lambda\;\HSnorm{[\hat H^{\,}_\lambda,\hat\rho^{\,}_0]}^{\,}.
\]
Since \(M(\lambda)=\cos\Theta^{\,}_\lambda\), we have
\(\partial^{\,}_{\lambda}{M}=-\sin\Theta^{\,}_\lambda\,\partial^{\,}_{\lambda}\Theta^{\,}_\lambda\).
Canceling \(\sin\Theta^{\,}_\lambda\) gives
\[
  \left|\partial^{\,}_{\lambda}\Theta^{\,}_\lambda\right|
  \le \frac{1}{\hbar}\frac{1}{|\Gamma|}\frac{1}{\HSnorm{\rho^{\,}_0}^{\,}}\,
  \HSnorm{[\hat H^{\,}_\lambda,\hat\rho^{\,}_0]}^{\,}.
\]
Integrating in \(\lambda\) and using the triangle inequality yields
\begin{align}
  \Theta^{\,}_\lambda  -\underbrace{\Theta^{\,}_{0}}^{\,}_{=0}
  =
  \int^{\lambda}_{0}\mathrm{d}\lambda' \,\partial^{\,}_{\lambda'}\Theta^{\,}_{\lambda'}
  &\leq
  \int^{\lambda}_{0}\mathrm{d}\lambda'\,\bigl| \partial^{\,}_{\lambda'}\Theta^{\,}_{\lambda'}\bigr|
  \nonumber\\
  &\leq
  \frac{1}{\hbar}\,
  \int_0^\lambda \frac{\mathrm{d}\lambda'}{|\partial^{\,}_{t}\lambda'|}\;
  \frac{1}{\HSnorm{\rho^{\,}_0}^{\,}}\,
  \HSnorm{[\hat H^{\,}_{\lambda'},\hat\rho^{\,}_0]}^{\,}.
  \label{eq: QSL_lambda_HS_norm}
\end{align}
Finally, the commutator norm can be written explicitly as
\[
  \HSnorm{[\hat H^{\,}_{\lambda},\hat\rho_0]}^2
  =2\left[
    \Tr(\hat\rho_0^2\hat H_\lambda^2)
    -\Tr(\hat\rho_0\hat H^{\,}_\lambda\hat\rho_0\hat H_\lambda)
  \right],
\]
so Eq.~\eqref{eq: QSL_lambda_HS_norm} is equivalent to Eq.~\eqref{eq: QSL mixed states} in the main text (with $\hbar\equiv 1$).

\section{Derivation of mixed-state fidelity bounds [Eq.\ \eqref{eq: fidelity bound}]}
\label{sec:derivation_of_mixed_state_fidelity_bounds_ref_eq_fidelity_bound}

This section derives Eq.~\eqref{eq: fidelity bound} of the main text.
The derivation is inspired by Ref.~\cite{Chen:2022qyb}.
Let $\HSket{\Psi^{\,}_{0}}$, $\HSket{\Phi_\lambda}$, and $\HSket{\Psi_\lambda}$ be \emph{normalized} vectors
in the Liouville (Hilbert-Schmidt) space introduced in Sec.~\ref{sec: derivation_mixed_QSL}.
We define the Hilbert--Schmidt fidelity
\[
F[\Phi_2,\Phi_1]:=\bigl|\HSinner{\Phi_2}{\Phi_1}\bigr|^2\in[0,1].
\]
Our goal is to bound $F[\Phi_\lambda,\Psi_\lambda]$ in terms of
$F[\Phi_\lambda,\Psi^{\,}_{0}]$ and the Hilbert--Schmidt angle $\Theta_\lambda$ defined by
\[
\cos\Theta_\lambda:=\bigl|\HSinner{\Psi^{\,}_{0}}{\Psi_\lambda}\bigr|
=\sqrt{F[\Psi^{\,}_{0},\Psi_\lambda]},
\qquad
\Theta_\lambda\in[0,\pi/2].
\]

Introduce the rank-one superprojector onto $\HSket{\Psi^{\,}_{0}}$ and its orthogonal complement as in
Eq.~\eqref{eq: def rank one superprojector}.
By orthogonal decomposition,
\[
\HSket{\Psi_\lambda} = \hat{\hat{\mathcal{P}}}\HSket{\Psi_\lambda} + \hat{\hat{\mathcal{Q}}}\HSket{\Psi_\lambda},
\qquad
\HSket{\Phi_\lambda} = \hat{\hat{\mathcal{P}}}\HSket{\Phi_\lambda} + \hat{\hat{\mathcal{Q}}}\HSket{\Phi_\lambda}.
\]
It follows that
\[
F[\Phi_\lambda,\Psi^{\,}_{\lambda}]
= \Bigl|\HSbra{\Phi_\lambda}\bigl(\hat{\hat{\mathcal{P}}}+\hat{\hat{\mathcal{Q}}}\bigr)\HSket{\Psi^{\,}_{\lambda}}\Bigr|^2
=
\abs{\HSbra{\Phi_\lambda} \hat{\hat{\mathcal{P}}} \HSket{\Psi_{\lambda}} }^2
+ \abs{\HSbra{\Phi_\lambda} \hat{\hat{\mathcal{Q}}}\HSket{\Psi_{\lambda}} }^2
+ 2\,\mathrm{Re}
\Big[\HSbra{\Phi_\lambda} \hat{\hat{\mathcal{Q}}}\HSket{\Psi_{\lambda}}\,
      \HSbra{\Psi_{\lambda}} \hat{\hat{\mathcal{P}}} \HSket{\Phi_\lambda}\Big].
\]
Therefore,
\begin{align}
& \abs{F[\Phi_\lambda,\Psi_{\lambda}] - F[\Phi_\lambda,\Psi^{\,}_{0}]}\nonumber\\
&= \Bigg|
\abs{\HSbra{\Phi_\lambda} \hat{\hat{\mathcal{P}}} \HSket{\Psi_{\lambda}} }^2
+ \abs{\HSbra{\Phi_\lambda} \hat{\hat{\mathcal{Q}}} \HSket{\Psi_{\lambda}} }^2
+ 2\,\mathrm{Re}\!\left[\HSbra{\Phi_\lambda} \hat{\hat{\mathcal{Q}}}\HSket{\Psi_{\lambda}}\,
\HSbra{\Psi_{\lambda}} \hat{\hat{\mathcal{P}}} \HSket{\Phi_\lambda}\right]
- F[\Phi_\lambda,\Psi^{\,}_{0}]
\Bigg|\nonumber\\
&\le
\Bigg|
\abs{\HSbra{\Phi_\lambda} \hat{\hat{\mathcal{P}}} \HSket{\Psi_{\lambda}} }^2
+ \abs{\HSbra{\Phi_\lambda} \hat{\hat{\mathcal{Q}}} \HSket{\Psi_{\lambda}} }^2
- F[\Phi_\lambda,\Psi^{\,}_{0}]
\Bigg|
+ 2\Bigg|
\mathrm{Re}\!\left[\HSbra{\Phi_\lambda} \hat{\hat{\mathcal{Q}}}\HSket{\Psi_{\lambda}}\,
\HSbra{\Psi_{\lambda}} \hat{\hat{\mathcal{P}}} \HSket{\Phi_\lambda}\right]
\Bigg|\nonumber\\
&\le
\Bigg|
\abs{\HSbra{\Phi_\lambda} \hat{\hat{\mathcal{P}}} \HSket{\Psi_{\lambda}} }^2
+ \abs{\HSbra{\Phi_\lambda} \hat{\hat{\mathcal{Q}}} \HSket{\Psi_{\lambda}} }^2
- F[\Phi_\lambda,\Psi^{\,}_{0}]
\Bigg|
+ 2\abs{\HSbra{\Phi_\lambda} \hat{\hat{\mathcal{Q}}}\HSket{\Psi_{\lambda}}}\,
      \abs{\HSbra{\Psi_{\lambda}} \hat{\hat{\mathcal{P}}} \HSket{\Phi_\lambda}}\nonumber\\
&=
\abs{-\sin^2\Theta_\lambda\,F[\Phi_\lambda,\Psi^{\,}_{0}] + \mathcal{D}_{\rm un}(\lambda)}
+ 2\cos\Theta_\lambda \sqrt{F[\Phi_\lambda,\Psi^{\,}_{0}]}\sqrt{\mathcal{D}_{\rm un}(\lambda)}.
\label{eq: derive the main inequalities}
\end{align}
Here, we used the triangle inequality to obtain the first inequality, the bound
$\mathrm{Re}(z)\le |z|$ for $z\in\mathbb{C}$ to obtain the second inequality, and the identities
\[
\|\hat{\hat{\mathcal{P}}}\HSket{\Psi_{\lambda}} \|  = \cos\Theta_\lambda,\qquad
\|\hat{\hat{\mathcal{P}}}\HSket{\Phi_\lambda} \| = \sqrt{F[\Phi_\lambda,\Psi^{\,}_{0}]},\qquad
\mathcal{D}_{\rm un}(\lambda) := \abs{\HSbra{\Phi_\lambda} \hat{\hat{\mathcal{Q}}} \HSket{\Psi_{\lambda}} }^2.
\]

By the Cauchy--Schwarz inequality,
\[
\sqrt{\mathcal{D}_{\rm un}(\lambda)}
\le
\|\hat{\hat{\mathcal{Q}}}\HSket{\Psi_{\lambda}}\|\,
\|\hat{\hat{\mathcal{Q}}}\HSket{\Phi_\lambda}\|
=
\sin\Theta_\lambda\,
\sqrt{1-F[\Phi_\lambda,\Psi^{\,}_{0}]},
\]
and hence Eq.~\eqref{eq: derive the main inequalities} implies
\begin{align}
\abs{F[\Phi_\lambda,\Psi_{\lambda}] - F[\Phi_\lambda,\Psi^{\,}_{0}]}
\leq
\sin^2\Theta^{\,}_\lambda \,\abs{1-2F[\Phi_\lambda,\Psi^{\,}_{0}]}
+ \sin\!\left(2\Theta^{\,}_\lambda\right)
\sqrt{F[\Phi_\lambda,\Psi^{\,}_{0}]}\,\sqrt{1-F[\Phi_\lambda,\Psi^{\,}_{0}] }.
\label{appeq: bound on fidelity}
\end{align}
Maximizing the right-hand side of Eq.~\eqref{appeq: bound on fidelity} with respect to $F[\Phi_\lambda,\Psi^{\,}_{0}]$
yields a simpler, though looser bound,
\begin{align}
\abs{F[\Phi_\lambda,\Psi_{\lambda}] - F[\Phi_\lambda,\Psi^{\,}_{0}]}
\leq
\sin\Theta_\lambda .
\label{appeq: bound on fidelity simpler}
\end{align}

Finally, Eq.~\eqref{eq: fidelity bound} of the main text follows by making the identifications
\[
\HSket{\Psi^{\,}_{\lambda}}\equiv \frac{\HSket{\rho^{\,}_{\lambda}}}{\|\rho^{\,}_{\lambda}\|},
\qquad
\HSket{\Phi^{\,}_{\lambda}}\equiv \frac{\HSket{\sigma^{\,}_{\lambda}}}{\|\sigma^{\,}_{\lambda}\|},
\qquad
\HSket{\Psi^{\,}_{0}}\equiv \frac{\HSket{\rho^{\,}_{0}}}{\|\rho^{\,}_{0}\|}
\]
in Eqs.~\eqref{appeq: bound on fidelity} and \eqref{appeq: bound on fidelity simpler}, and then applying the
mixed-state quantum speed limit inequality~\eqref{eq: QSL mixed states} to bound $\Theta^{\,}_{\lambda}$.

\section{Proof of Theorem~\ref{thm:GammaThTempScaling}}
\label{sec:proof_of_theorem_ref_thm_gammathtempscaling}

In this section, we provide a proof of Theorem~\ref{thm:GammaThTempScaling} in the main text.
To this end, we evaluate both $\delta V$ [Eq.~\eqref{eq: QSL mixed states simplified}]
and $\chi^{\,}_{\mathrm{F}}$ [Eq.~\eqref{eq: define fidelity susceptiblity}]
in the low- and high-temperature regimes.

\subsection{Useful representations for $\delta V$ and $\chi^{\,}_{\mathrm{F}}$}

\subsubsection{Useful representation for $\delta V$}

We begin by deriving a convenient expression for $\delta V$ \eqref{eq: QSL mixed states simplified}:
\begin{align}
\delta V
&:= \sqrt{2\,I^{\,}_{\mathrm{WY}}\!\left(\hat{\tilde{\rho}}^{\,}_{0},\hat{V}\right)},
\qquad
I^{\,}_{\mathrm{WY}}\!\left(\hat{\tilde{\rho}}^{\,}_{0},\hat{V}\right)
:=
\Tr\!\left[\hat{\tilde{\rho}}^{\,}_{0}\,\hat{V}^2\right]
-\Tr\!\left[\left(\hat{\tilde{\rho}}^{1/2}_{0}\hat{V}\right)^2\right].
\label{eq: WYSI in app}
\end{align}
When $\hat{\rho}^{\,}_{0}(\beta)$ is a Gibbs state at inverse temperature $\beta$
\eqref{eq: initial Gibbs state}, the order-$2$ escort state $\hat{\tilde\rho}^{\,}_0$ reads
\[
\hat{\tilde\rho}^{\,}_0
:=\frac{\hat\rho_0^{\,2}}{\Tr[\hat\rho_0^{\,2}]}
=
\frac{e^{-2\beta \hat H_0}}{Z^{\,}_{0}(2\beta)}
=\hat\rho^{\,}_{0}(2\beta).
\]
Therefore, Eq.~\eqref{eq: WYSI in app} becomes
\[
I^{\,}_{\mathrm{WY}}\!\left(\hat{\tilde{\rho}}^{\,}_{0},\hat{V}\right)
=
\Tr\!\left[\hat\rho^{\,}_{0}(2\beta)\,\hat{V}^2\right]
-\Tr\!\left[\hat\rho^{\,}_{0}(2\beta)\,\hat{V}(\beta)\hat{V}\right]
=: I^{\,}_{\mathrm{WY}}(\beta),
\]
where
$
\hat{V}(\beta):=e^{\beta \hat{H}^{\,}_{0}}\hat{V} e^{-\beta \hat{H}^{\,}_{0}}
$
is the imaginary-time Heisenberg operator at time $\beta$.
Introducing the shorthand for a thermal average at inverse temperature $2\beta$,
$
\langle\cdots\rangle^{\,}_{2\beta}:=\Tr[\hat{\rho}^{\,}_{0}(2\beta)\cdots],
$
we can write $\delta V$ \eqref{eq: WYSI in app} compactly as
\begin{align}
\delta V = \sqrt{2\,I^{\,}_{\mathrm{WY}}(\beta)},
\qquad
I^{\,}_{\mathrm{WY}}(\beta)=\langle\hat{V}^2\rangle^{\,}_{2\beta} - \langle\hat{V}(\beta)\hat{V}\rangle^{\,}_{2\beta}.
\label{eq: QSL mixed states simplified app}
\end{align}

\subsubsection{Useful representation for $\chi^{\,}_{\mathrm{F}}$}

We now turn to the fidelity susceptibility $\chi^{\,}_{\mathrm{F}}$ \eqref{eq: define fidelity susceptiblity},
\[
\chi^{\,}_{\mathrm{F}}
:=-\frac{\partial^{2} \ln\mathcal{C}(\lambda)}{\partial \lambda^{2}}\Big|^{\,}_{\lambda=0},
\]
where the thermal-state overlap $\mathcal{C}(\lambda)$ is defined in
Eq.~\eqref{eq: thermal state overlap defined}.
Since we consider a unitary evolution, the purity is conserved,
$\Tr[\hat{\rho}^{\,}_{0}(\beta)^2]=\Tr[\hat{\sigma}^{\,}_{\lambda}(\beta)^2]$,
and hence, the denominators in Eq.~\eqref{eq: thermal state overlap defined} are independent of $\lambda$.
Therefore, in computing $\chi^{\,}_{\mathrm{F}}$ we only need
\[
\chi^{\,}_{\mathrm F}
=
-2\,\frac{\partial^2 \ln S(\lambda)}{\partial\lambda^2}\Big|_{\lambda=0},
\qquad
S(\lambda):=\Tr[\hat{\rho}_0(\beta)\hat{\sigma}_\lambda(\beta)].
\]
Working out the derivatives gives
\[
\chi^{\,}_{\mathrm F}
=
-2
\left[
\frac{S''(0)}{S(0)}
-
\left(\frac{S'(0)}{S(0)}\right)^2
\right].
\]

Write the Boltzmann weights $p^{\,}_n:=e^{-\beta E_n^{(0)}}$ and $Z^{\,}_{0}(\beta)=\sum_n p^{\,}_n$.
Upon using the spectral decomposition
\[
\hat{\rho}_0(\beta)=\frac{1}{Z^{\,}_{0}(\beta)}\sum_m p^{\,}_m\,|E_m^{(0)}\rangle\langle E^{(0)}_{m}|,
\qquad
\hat{\sigma}_\lambda(\beta)=\frac{1}{Z^{\,}_{0}(\beta)}\sum_n p^{\,}_n\,\dyad{E_n(\lambda)},
\]
we obtain
\begin{align}
S(\lambda)
&=
\frac{1}{Z^{\,}_{0}(\beta)^2}\sum_{n} p_n^{2}\,|\langle E_n^{(0)}|E_n(\lambda)\rangle|^2
+
\frac{1}{Z^{\,}_{0}(\beta)^2}\sum_{\substack{m,n\\ m\neq n}} p^{\,}_m p^{\,}_n\,|\langle E_m^{(0)}|E_n(\lambda)\rangle|^2.
\label{eq:S_overlap_sum}
\end{align}

For $\hat H_\lambda=\hat H_0+\lambda \hat V$, standard (non-degenerate) perturbation theory yields
\[
|E_n(\lambda)\rangle
=
|E_n^{(0)}\rangle
+\lambda\sum_{m\neq n}
\frac{\langle E_m^{(0)}|\hat V|E_n^{(0)}\rangle}{E_n^{(0)}-E_m^{(0)}}\,|E_m^{(0)}\rangle
+\mathcal{O}(\lambda^2).
\]
For $m\neq n$, the overlap expands as
\[
|\langle E_m^{(0)}|E_n(\lambda)\rangle|^2
=
\lambda^2
\frac{|\langle E_m^{(0)}|\hat{V}|E_n^{(0)}\rangle|^2}{(E_n^{(0)}-E_m^{(0)})^2}
+\mathcal{O}(\lambda^3)
\qquad (m\neq n),
\]
while normalization implies
\[
|\langle E_n^{(0)}|E_n(\lambda)\rangle|^2
=
1-\lambda^2\sum_{m\neq n}
\frac{|\langle E_m^{(0)}|\hat{V}|E_n^{(0)}\rangle|^2}{(E_n^{(0)}-E_m^{(0)})^2}
+\mathcal{O}(\lambda^3).
\]
Substituting into Eq.~\eqref{eq:S_overlap_sum} and reorganizing the sums gives
\[
S(\lambda)
=
S(0)
-\frac{\lambda^2}{2Z^{\,}_{0}(\beta)^2}\sum_{\substack{m,n\\ m\neq n}}
(p^{\,}_m-p^{\,}_n)^2
\frac{|\langle E_m^{(0)}|\hat{V}|E_n^{(0)}\rangle|^2}{(E_m^{(0)}-E_n^{(0)})^2}
+\mathcal{O}(\lambda^3),
\]
with
\begin{equation}
S(0)=\Tr(\hat{\rho}_0^2)=\frac{Z^{\,}_{0}(2\beta)}{Z^{\,}_{0}(\beta)^2}.
\label{eq: purity using pf}
\end{equation}
It follows that $S'(0)=0$ and
\[
S''(0)=
-\frac{1}{Z^{\,}_{0}(\beta)^2}\sum_{\substack{m,n\\ m\neq n}}
(p^{\,}_m-p^{\,}_n)^2
\frac{|\langle E_m^{(0)}|\hat{V}|E_n^{(0)}\rangle|^2}{(E_m^{(0)}-E_n^{(0)})^2}.
\]
We thus obtain
\begin{align}
\chi^{\,}_{\mathrm F}
=
\frac{2}{Z^{\,}_{0}(2\beta)}
\sum_{\substack{m,n\\ m\neq n}}
\left(e^{-\beta E_m^{(0)}}-e^{-\beta E_n^{(0)}}\right)^2
\frac{|\langle E_m^{(0)}|\hat{V}|E_n^{(0)}\rangle|^2}{(E_m^{(0)}-E_n^{(0)})^2},
\label{eq: define fidelity susceptiblity simple form app}
\end{align}
where $Z^{\,}_0(2\beta)=\sum_m e^{-2\beta E_m^{(0)}}$.

\subsection{Low-temperature regime}

\subsubsection{$\delta V$ in the low-temperature regime}

We begin by evaluating the standard deviation $\delta V$ \eqref{eq: QSL mixed states simplified app}
in the low-temperature regime.
For $\beta\Delta\gg 1$, it is sufficient to retain the ground state (n=0) and the lowest excited state that
couples to it under $\hat V$.
Throughout, we choose $|E^{(0)}_{1}\rangle$ such that
$V^{\,}_{10}:=\langle E^{(0)}_{1}|\hat V|E^{(0)}_{0}\rangle\neq 0$ and
\begin{equation}
\Delta:=E^{(0)}_{1}-E^{(0)}_{0}
=
\min_{n>0:\, \langle E^{(0)}_{n}|\hat V|E^{(0)}_{0}\rangle\neq 0}
\bigl(E^{(0)}_{n}-E^{(0)}_{0}\bigr)
>0.
\end{equation}
In this approximation, the Gibbs state at inverse temperature $2\beta$ reads
\[
\hat{\rho}^{\,}_{0}(2\beta)
\simeq
\frac{e^{-2\beta E^{(0)}_0}}{Z^{\,}_{0}(2\beta)}
\left( |E^{(0)}_{0}\rangle\langle E^{(0)}_{0}|+e^{-2\beta \Delta}\,|E^{(0)}_{1}\rangle\langle E^{(0)}_{1}| \right),
\qquad
Z^{\,}_{0}(2\beta)
\simeq e^{-2\beta E^{(0)}_0}\left(1+e^{-2\beta \Delta}\right).
\]

Expanding to first order in $e^{-2\beta\Delta}$ yields
\begin{equation}
\hat{\rho}^{\,}_{0}(2\beta)
\approx \left(1-e^{-2\beta \Delta}\right) |E^{(0)}_{0}\rangle\langle E^{(0)}_{0}| +e^{-2\beta \Delta} |E^{(0)}_{1}\rangle\langle E^{(0)}_{1}|,
\label{eq: initial Gibbs state low T}
\end{equation}
which reduces to the ground-state projector as $\beta\to\infty$.

\paragraph*{Term 1: $\langle \hat V^2\rangle_{2\beta}$.}
Using Eq.~\eqref{eq: initial Gibbs state low T},
\begin{equation}
\langle\hat{V}^2\rangle^{\,}_{2\beta}
=\Tr[\hat{\rho}^{\,}_{0}(2\beta)\hat{V}^2]
\approx (1-e^{-2\beta\Delta})(V^{2})^{\,}_{00} + e^{-2\beta\Delta}(V^2)^{\,}_{11},
\label{eq:V2_lowT}
\end{equation}
where $(V^2)_{mn}:=\langle E^{(0)}_{m}|\hat{V}^2|E^{(0)}_{n}\rangle$.

\paragraph*{Imaginary-time Heisenberg evolution.}
The imaginary-time Heisenberg operator is
\[
\hat{V}(\beta):=e^{\beta \hat{H}^{\,}_{0}}\hat{V}e^{-\beta \hat{H}^{\,}_{0}}
=\sum_{m,n}e^{\beta(E_m^{(0)}-E_n^{(0)})}V_{mn}\,|E_m^{(0)}\rangle\langle E_n^{(0)}|.
\]
Retaining the two lowest levels gives
\begin{equation}
\hat V(\beta)\approx
V^{\,}_{00}|E^{(0)}_{0}\rangle\langle E^{(0)}_{0}|
+e^{-\beta\Delta}V^{\,}_{01}|E^{(0)}_{0}\rangle\langle E^{(0)}_{1}|
+e^{\beta\Delta}V^{\,}_{10}|E^{(0)}_{1}\rangle\langle E^{(0)}_{0}|
+V^{\,}_{11}|E^{(0)}_{1}\rangle\langle E^{(0)}_{1}|,
\label{eq:Vbeta_two_level}
\end{equation}
where $V_{mn}:=\langle E_m^{(0)}|\hat V|E_n^{(0)}\rangle$.

\paragraph*{Term 2: $\langle \hat V(\beta)\hat V\rangle_{2\beta}$.}
Using Eqs.~\eqref{eq: initial Gibbs state low T} and \eqref{eq:Vbeta_two_level}, and keeping terms up to
order $e^{-2\beta\Delta}$, we obtain
\begin{equation}
\langle\hat{V}(\beta)\hat{V}\rangle^{\,}_{2\beta}
=\Tr[\hat{\rho}^{\,}_{0}(2\beta)\hat{V}(\beta)\hat{V}]
\approx
(V^{\,}_{00})^2+2e^{-\beta\Delta}|V^{\,}_{10}|^2
+ e^{-2\beta\Delta}
\Bigl[(V_{11})^2-(V_{00})^2\Bigr].
\label{eq:VVbeta_lowT}
\end{equation}

\paragraph*{Low-$T$ expansion of $I_{\mathrm{WY}}(\beta)$ and $\delta V$.}
Recalling $I^{\,}_{\mathrm{WY}}(\beta)=\langle \hat{V}^2\rangle_{2\beta}-\langle \hat{V}(\beta)\hat{V}\rangle_{2\beta}$ [Eq.~\eqref{eq: QSL mixed states simplified app}],
Eqs.~\eqref{eq:V2_lowT} and \eqref{eq:VVbeta_lowT} give
\[
I^{\,}_{\mathrm{WY}}(\beta)
\approx
\Bigl[(V^{2})^{\,}_{00} - (V^{\,}_{00})^{2} \Bigr]
-2e^{-\beta\Delta}|V^{\,}_{10}|^2
+\mathcal O(e^{-2\beta\Delta}).
\]
Therefore,
\begin{align}
\delta V=\sqrt{2I^{\,}_{\mathrm{WY}}(\beta)}
&=
\delta V^{(0)}
\left(
1-2 e^{-\beta\Delta}\frac{|V^{\,}_{10}|^2}{(\delta V^{(0)})^2}
\right)
+\mathcal O(e^{-2\beta\Delta}),
\label{eq: QSL mixed states low T final app}
\end{align}
where
$
\delta V^{(0)}:=\sqrt{2}\,
\sqrt{(V^{2})^{\,}_{00} - (V^{\,}_{00})^{2} }.
$

Finally, we show in Sec.~\ref{app:Ns_deltav_LT} that
$|V^{\,}_{10}|=|\langle E^{(0)}_{1}|\hat{V}|E^{(0)}_{0}\rangle|$ scales as $\sqrt{N}$.
Since $\delta V^{(0)}$ also scales as $\sqrt{N}$, the ratio
$|V^{\,}_{10}|/\delta V^{(0)}$ in Eq.~\eqref{eq: QSL mixed states low T final app}
is independent of $N$ in the thermodynamic limit ($N\to\infty$).

\subsubsection{$\chi^{\,}_{\mathrm F}$ in the low-temperature regime}

We now evaluate the fidelity susceptibility $\chi^{\,}_{\mathrm F}$
\eqref{eq: define fidelity susceptiblity simple form app} in the low-temperature regime.
To simplify notation, define
\[
\Delta_m:=E_m^{(0)}-E_0^{(0)}\ge 0,
\qquad
A_{mn}:=\frac{\bigl|\bra{E_m^{(0)}}\hat V\ket{E_n^{(0)}}\bigr|^2}{\bigl(E_m^{(0)}-E_n^{(0)}\bigr)^2}.
\]
We again fix the label $m=1$ such that $V_{10}\neq 0$ and
\begin{equation}
\Delta:=\Delta_1
=
\min_{n>0:\,V_{n0}\neq 0}\Delta_n
>0.
\end{equation}
Factoring out $E_0^{(0)}$ from $Z^{\,}_0(2\beta)$ gives
\[
Z^{\,}_0(2\beta)
=e^{-2\beta E_0^{(0)}}\Bigl(1+\sum_{k>0}e^{-2\beta\Delta_k}\Bigr)
=:e^{-2\beta E_0^{(0)}}\,\tilde Z(2\beta),
\qquad
\tilde Z(2\beta)=1+\sum_{k>0}e^{-2\beta\Delta_k}.
\]
Likewise,
\[
\bigl(e^{-\beta E_m^{(0)}}-e^{-\beta E_n^{(0)}}\bigr)^2
=e^{-2\beta E_0^{(0)}}\bigl(e^{-\beta\Delta_m}-e^{-\beta\Delta_n}\bigr)^2.
\]
Therefore, Eq.~\eqref{eq: define fidelity susceptiblity simple form app} becomes
\begin{equation}
\chi^{\,}_{\mathrm F}(\beta)
=\frac{2}{\tilde Z(2\beta)}
\sum_{\substack{m,n\\ m\neq n}}
\bigl(e^{-\beta\Delta_m}-e^{-\beta\Delta_n}\bigr)^2\,A_{mn}.
\label{eq:chiF_reduced}
\end{equation}

\paragraph*{Step 1: Split the sum into ground--excited and excited--excited sectors.}
Write
\[
\sum_{\substack{m,n\\ m\neq n}}
=
\sum_{n>0}\Bigl[(m,n)=(0,n)+(n,0)\Bigr]
\;+\;
\sum_{\substack{m,n>0\\ m\neq n}}.
\]
Since $\Delta_0=0$, the ground--excited sector yields
\[
\sum_{n>0}\Bigl[(1-e^{-\beta\Delta_n})^2 A_{0n}+(e^{-\beta\Delta_n}-1)^2A_{n0}\Bigr]
=
2\sum_{n>0}(1-e^{-\beta\Delta_n})^2 A_{0n},
\]
where $A_{n0}=A_{0n}$.
Expanding,
\(
(1-e^{-\beta\Delta_n})^2 = 1-2e^{-\beta\Delta_n}+e^{-2\beta\Delta_n}.
\)
The excited--excited sector is
\begin{equation}
S_{\rm ex}(\beta):=\sum_{\substack{m,n>0\\ m\neq n}}
\bigl(e^{-\beta\Delta_m}-e^{-\beta\Delta_n}\bigr)^2 A_{mn}.
\label{eq:Sex_def}
\end{equation}
Since
\(
(e^{-\beta\Delta_m}-e^{-\beta\Delta_n})^2
\le e^{-2\beta\Delta_m}+e^{-2\beta\Delta_n}+2e^{-\beta(\Delta_m+\Delta_n)},
\)
one has $S_{\rm ex}(\beta)=\mathcal O(e^{-2\beta\Delta})$.

Putting these into Eq.~\eqref{eq:chiF_reduced} gives
\begin{equation}
\chi^{\,}_{\mathrm F}(\beta)
=\frac{2}{\tilde Z(2\beta)}
\left[
2\sum_{n>0}\Bigl(1-2e^{-\beta\Delta_n}+e^{-2\beta\Delta_n}\Bigr)A_{0n}
+S_{\rm ex}(\beta)
\right].
\label{eq:chiF_split}
\end{equation}

\paragraph*{Step 2: Expand the normalization $1/\tilde Z(2\beta)$.}
Since $\tilde Z(2\beta)=1+\sum_{k>0}e^{-2\beta\Delta_k}$,
\begin{equation}
\frac{1}{\tilde Z(2\beta)}
=
1-\sum_{k>0}e^{-2\beta\Delta_k}+\mathcal O(e^{-4\beta\Delta}).
\label{eq:Zinv_expand}
\end{equation}

\paragraph*{Step 3: Collect terms order by order.}
Define the zero-temperature fidelity susceptibility
\begin{equation}
\chi^{(0)}_{\mathrm F}
:=
\chi^{\,}_{\mathrm F}(\beta\to\infty)
=4\sum_{n>0}A_{0n}
=
4\sum_{n>0}\frac{\bigl|\bra{E_n^{(0)}}\hat V\ket{E_0^{(0)}}\bigr|^2}{\Delta_n^2}.
\label{eq:chiF0_def}
\end{equation}
Using Eqs.~\eqref{eq:chiF_split} and \eqref{eq:Zinv_expand}, we obtain
\begin{align}
\chi^{\,}_{\mathrm F}(\beta)
=
\chi^{(0)}_{\mathrm F}
-8\sum_{n>0}e^{-\beta\Delta_n}A_{0n}
+\mathcal O(e^{-2\beta\Delta}).
\label{eq:chiF_lowT_leading}
\end{align}
(At order $\mathcal O(e^{-2\beta\Delta})$, the contributions come from the $e^{-2\beta\Delta_n}$ term in $(1-e^{-\beta\Delta_n})^2$,
the correction from $1/\tilde Z(2\beta)$, and the excited--excited sector $S_{\rm ex}(\beta)$.)

\paragraph*{Factorizing $\chi^{(0)}_{\mathrm F}$.}
Keeping only the leading Boltzmann correction in Eq.~\eqref{eq:chiF_lowT_leading} and factorizing $\chi^{(0)}_{\mathrm F}$ yields
\begin{equation}
\chi^{\,}_{\mathrm F}(\beta)
=
\chi^{(0)}_{\mathrm F}
\left[
1
-8e^{-\beta\Delta}\,
\frac{1}{\chi^{(0)}_{\mathrm F}}\frac{|V^{\,}_{10}|^2}{\Delta^2}
\right]
+\mathcal O\!\left(e^{-\beta\Delta_2},\,e^{-2\beta\Delta}\right),
\label{eq: fidelity susceptiblity low T app}
\end{equation}
where $\Delta_2$ denotes the next smallest excitation gap among states with nonzero ground-state coupling (namely, those with $V_{n0}\neq 0$).

We show in Sec.~\ref{app:Ns_chiF_LT} that \(|V^{\,}_{10}|^{2}/\Delta^{2}\asymp N\).
Since \(\chi^{(0)}_{\mathrm{F}}\asymp N\) in a gapped phase, the ratio
\((1/\chi^{(0)}_{\mathrm{F}})\,(|V^{\,}_{10}|^{2}/\Delta^{2})\)
is asymptotically independent of \(N\) as \(N\to\infty\).

\subsubsection{Threshold driving rate $\Gamma^{\,}_{\mathrm{th}}$ in the low-temperature regime}

We now compute the threshold driving rate $\Gamma^{\,}_{\mathrm{th}}$ \eqref{eq: define threshold driving rate}
by substituting the low-temperature expansions of $\delta V$ \eqref{eq: QSL mixed states low T final app}
and $\chi^{\,}_{\mathrm{F}}$ \eqref{eq: fidelity susceptiblity low T app}:
\begin{align}
\Gamma^{\,}_{\mathrm{th}}
=
\Gamma^{\,}_{N}
\frac{
1
-
2e^{-\beta \Delta}\left(\frac{|V^{\,}_{10}|}{\delta V^{(0)}}\right)^{2}
}
{
1
-8
e^{-\beta \Delta}
\frac{1}{\chi^{(0)}_{\mathrm F}}
\frac{|V^{\,}_{10}|^2}{\Delta^2}
}
\equiv
\Gamma^{\,}_{N}\, f^{\,}_{\mathrm{low}\text{-}T}(\beta)
+\mathcal{O}(e^{-2\beta \Delta}),
\nonumber
\end{align}
where expanding the ratio to first order in $e^{-\beta\Delta}$ gives
\begin{align}
f^{\,}_{\mathrm{low}\text{-}T}(\beta)
=
1
+
2e^{-\beta \Delta}
\underbrace{\left[
-
\left(\frac{|V^{\,}_{10}|}{\delta V^{(0)}}\right)^{2}
+
\frac{4}{\chi^{(0)}_{\mathrm F}}
\frac{|V^{\,}_{10}|^2}{\Delta^2}
\right]}_{=:W}.
\label{eq: temperature scaling low T app}
\end{align}
In the thermodynamic limit, $f^{\,}_{\mathrm{low}\text{-}T}(\beta)$ is independent of $N$
because both ratios $|V_{10}|/\delta V^{(0)}$ and $|V_{10}|^2/(\chi_F^{(0)}\Delta^2)$ are $N$-independent.

\paragraph*{Bounds on $W$.}
We next show that $W\ge 0$ and $W\le 1$.
From Eq.~\eqref{eq: temperature scaling low T app},
\[
W
=
|V_{10}|^2
\left(
-
\frac{1}{(\delta V^{(0)})^2}
+\frac{4}{\chi_F^{(0)}\,\Delta^2}
\right),
\qquad
V_{n0}:=\bra{E_n^{(0)}}\hat V\ket{E_0^{(0)}}.
\]
Introduce the dimensionless ratios
\(
a:=\frac{|V_{10}|^2}{(\delta V^{(0)})^2}
\)
and
\(
b:=\frac{|V_{10}|^2}{\chi_F^{(0)}\,\Delta^2},
\)
so that
\begin{equation}
W=-a+4b.
\label{eq:W-ab}
\end{equation}
Now note that $\chi^{(0)}_{\mathrm{F}}$ defined in Eq.~\eqref{eq:chiF0_def} implies two simple spectral inequalities.
First,
\[
\chi_F^{(0)}
=4\sum_{n\neq 0}\frac{|V_{n0}|^2}{\Delta_n^2}
\ge 4\frac{|V_{10}|^2}{\Delta^2},
\]
which implies
\begin{equation}
b\le\frac{1}{4}.
\label{eq: bound on a and b first}
\end{equation}
Second, since $\Delta\le \Delta_n$ for every $n>0$ with $V_{n0}\neq 0$,
\[
\chi_F^{(0)}
=4\sum_{n\neq 0}\frac{|V_{n0}|^2}{\Delta_n^2}
\le 4\frac{1}{\Delta^2}\sum_{n\neq 0}|V_{n0}|^2
=2\frac{(\delta V^{(0)})^2}{\Delta^2},
\]
where we used $\sum_{n\neq 0}|V_{n0}|^2=(\delta V^{(0)})^2/2$.
This implies
\begin{equation}
a\le 2b.
\label{eq: bound on a and b second}
\end{equation}
Combining Eqs.~\eqref{eq: bound on a and b first} and \eqref{eq: bound on a and b second}, and using $a\ge 0$, we obtain
\(
0\le a\le 2b \le \frac{1}{2}.
\)
Together with Eq.~\eqref{eq:W-ab}, this yields the universal bound
\begin{equation}
0\le W\le 1.
\label{eq:W-bound}
\end{equation}
In our convention, $V_{10}\neq 0$ by construction, and hence, $W>0$.
The upper end $W\simeq 1$ is approached when $b\simeq 1/4$ (so that $\chi_F^{(0)}$ is dominated by the $n=1$ contribution)
while $a\ll 1$ (so that $(\delta V^{(0)})^2$ is dominated by matrix elements to higher excited states).

We emphasize that $W$ is generally model dependent, even though it always obeys the universal bound~\eqref{eq:W-bound}.
Indeed, its value is controlled by how the spectral weight of $\hat V$ is distributed across excitations:
\begin{equation}
a=\frac{|V_{10}|^2}{2\sum_{n\neq 0}|V_{n0}|^2},
\qquad
b=\frac{|V_{10}|^2/\Delta^2}{4\sum_{n\neq 0}|V_{n0}|^2/\Delta_n^2}.
\label{eq:a-b-spectral}
\end{equation}

Finally, applying Eq.~\eqref{eq:W-bound} to Eq.~\eqref{eq: temperature scaling low T app} gives the low-temperature form stated in
Theorem~\ref{thm:GammaThTempScaling}, with
\[
c^{\,}_{1}=2W
=2|V_{10}|^2
\left(
-
\frac{1}{(\delta V^{(0)})^2}
+\frac{4}{\chi_F^{(0)}\,\Delta^2}
\right).
\]

\subsection{High-temperature regime}

\subsubsection{$\delta V$ in the high-temperature regime}

Without loss of generality, we assume $\Tr\big[\hat H^{\,}_{0}\big]=0$
(this can always be achieved by shifting $\hat H_0\mapsto \hat H_0-\Tr\big[\hat H_0\big]\mathbb I/d$, which does not affect the dynamics)
and that the Hilbert-space dimension scales as $d\asymp 2^{N}$.

At infinite temperature, the state is maximally mixed ($\hat\rho=\mathbb I/d$), which is invariant under unitary time evolution;
hence, all normalized overlaps are trivially unity.
In the high-temperature regime, the appropriate small parameter is $\beta$ compared to microscopic local energy scales
(for example, $\beta\|\hat h_i\|\ll 1$ for the local terms in $\hat H_0=\sum_i \hat h_i$).
We therefore expand the Boltzmann weight to second order in $\beta$. The initial Gibbs state
\eqref{eq: initial Gibbs state} is
\[
\hat{\rho}^{\,}_{0}(\beta)
\approx
\frac{\mathbb{I}-\beta \hat{H}^{\,}_{0}+\frac{\beta^2}{2}\hat{H}^{2}_{0}}{Z^{\,}_{0}(\beta)}.
\]
The partition function $Z^{\,}_{0}(\beta)$ is fixed by $\Tr[\hat\rho_0(\beta)]=1$, hence
\begin{align}
Z^{\,}_{0}(\beta)
&=\Tr\!\left[\mathbb{I}-\beta \hat{H}^{\,}_{0}+\frac{\beta^2}{2}\hat{H}^{2}_{0}\right]
= d+ \frac{\beta^2}{2}\Tr\big[\hat{H}^{2}_{0}\big].
\label{eq: initial pf high T}
\end{align}
Therefore,
\[
\hat{\rho}^{\,}_{0}(\beta)
\approx
\frac{1}{d}
\left[
\mathbb{I}-\beta \hat{H}^{\,}_{0}
+\frac{\beta^2}{2}\left(\hat{H}^{2}_{0}-\frac{1}{d}\Tr\big[\hat{H}^{2}_{0}\big]\,\mathbb I\right)
\right].
\]
Likewise, the initial state at inverse temperature $2\beta$ is
\begin{align}
\hat{\rho}^{\,}_{0}(2\beta)
&\approx
\frac{1}{d}
\left[
\mathbb{I}-2\beta \hat{H}^{\,}_{0}
+2\beta^2\left(\hat{H}^{2}_{0}-\frac{1}{d}\Tr\big[\hat{H}^{2}_{0}\big]\,\mathbb I\right)
\right].
\label{eq: initial Gibbs state high T 2 beta}
\end{align}

\paragraph*{Term 1: $\langle \hat V^2\rangle_{2\beta}$.}
Using Eq.~\eqref{eq: initial Gibbs state high T 2 beta},
\begin{align}
\langle\hat{V}^2\rangle^{\,}_{2\beta}
=\Tr[\hat{\rho}^{\,}_{0}(2\beta)\hat{V}^2]
\approx
\frac{1}{d}
\Biggl[
\Tr\big[\hat{V}^2\big]
-2\beta\, \Tr[\hat{H}^{\,}_{0}\hat{V}^2]
+2\beta^2\left(\Tr[\hat{H}^{2}_{0}\hat{V}^2]-\frac{1}{d}\Tr\big[\hat{H}^{2}_{0}\big]\Tr[\hat{V}^{2}]\right)
\Biggr].
\label{eq:V2_highT}
\end{align}

\paragraph*{Imaginary-time Heisenberg evolution.}
Using
\(
e^{\pm \beta \hat H_0}
=\mathbb I\pm \beta \hat H_0+\frac{\beta^2}{2}\hat H_0^2+\mathcal O(\beta^3),
\)
we obtain, up to $\mathcal O(\beta^2)$,
\begin{align}
\hat V(\beta)
&=
e^{\beta \hat H_0}\hat V e^{-\beta \hat H_0}
=
\hat V
+\beta\bigl[\hat H_0,\hat V\bigr]
+\frac{\beta^2}{2}\bigl[\hat H_0,\bigl[\hat H_0,\hat V\bigr]\bigr]
+\mathcal O(\beta^3),
\label{eq:Vbeta_BCH}
\end{align}
where
\(
[\hat H_0,[\hat H_0,\hat V]]
=\hat H_0^2\hat V-2\hat H_0\hat V\hat H_0+\hat V\hat H_0^2.
\)

\paragraph*{Term 2: $\langle \hat V(\beta)\hat V\rangle_{2\beta}$.}
Using Eq.~\eqref{eq:Vbeta_BCH}, we expand
\[
\hat V(\beta)\hat V
=
\hat V^2
+\beta\,[\hat H_0,\hat V]\,\hat V
+\frac{\beta^2}{2}\,[\hat H_0,[\hat H_0,\hat V]]\,\hat V
+\mathcal O(\beta^3).
\]
Therefore,
\begin{align}
\langle \hat V(\beta)\hat V\rangle_{2\beta}
&:=\Tr[\hat\rho^{\,}_0(2\beta)\hat V(\beta)\hat V]
\nonumber\\
&\approx
\frac{1}{d}\Tr\big[\hat{V}^2\big]
-\frac{2\beta}{d}\Tr[\hat H_0\hat V^2]
+\frac{\beta}{d}\Tr\!\left([\hat H_0,\hat V]\hat V\right)
\nonumber\\
&\quad
+\frac{2\beta^2}{d}\left(\Tr[\hat H_0^2\hat V^2]-\frac{\Tr[\hat H_0^2]\Tr\big[\hat{V}^2\big]}{d}\right)
-\frac{2\beta^2}{d}\Tr\!\left(\hat H_0[\hat H_0,\hat V]\hat V\right)
+\frac{\beta^2}{2d}\Tr\!\left([\hat H_0,[\hat H_0,\hat V]]\hat V\right)
+\mathcal O(\beta^3).
\label{eq:VVbeta_highT_expand_raw}
\end{align}
Now simplify the trace terms using cyclicity:
\[
\Tr\!\left([\hat H_0,\hat V]\hat V\right)
=\Tr(\hat H_0\hat V^2)-\Tr(\hat V\hat H_0\hat V)=0,
\quad
\Tr\!\left(\hat H_0[\hat H_0,\hat V]\hat V\right)
=\Tr(\hat H_0^2\hat V^2)-\Tr(\hat H_0\hat V\hat H_0\hat V),
\]
\[
\Tr\!\left([\hat H_0,[\hat H_0,\hat V]]\hat V\right)
=2\Tr(\hat H_0^2\hat V^2)-2\Tr(\hat H_0\hat V\hat H_0\hat V).
\]
Substituting these identities into Eq.~\eqref{eq:VVbeta_highT_expand_raw} yields
\[
\langle \hat V(\beta)\hat V\rangle_{2\beta}
\approx
\frac{1}{d}\Tr\big[\hat{V}^2\big]
-\frac{2\beta}{d}\Tr[\hat H_0\hat V^2]
+\frac{\beta^2}{d}
\left(
\Tr[\hat H_0^2\hat V^2]
+\Tr[\hat H_0\hat V\hat H_0\hat V]
-\frac{2\,\Tr[\hat H_0^2]\Tr\big[\hat{V}^2\big]}{d}
\right)
+\mathcal O(\beta^3).
\]

\paragraph*{High-temperature expansion of $I_{\mathrm{WY}}(\beta)$ and $\delta V$.}
Taking the difference,
\[
I^{\,}_{\mathrm{WY}}(\beta)
=\langle \hat V^2\rangle_{2\beta}-\langle \hat V(\beta)\hat V\rangle_{2\beta}
\approx
\frac{\beta^2}{d}\left(
\Tr[\hat H_0^2\hat V^2]-\Tr[\hat H_0\hat V\hat H_0\hat V]
\right)
+\mathcal O(\beta^3).
\]
Upon using the identity (for Hermitian $\hat H^{\,}_0$ and $\hat V$)
\[
\|[\hat H_0,\hat V]\|^{2}_{\mathrm{HS}}
=
\Tr\!\left([\hat H_0,\hat V]^\dagger[\hat H_0,\hat V]\right)
=
2\left(
\Tr[\hat H_0^2\hat V^2]-\Tr[\hat H_0\hat V\hat H_0\hat V]
\right),
\]
we obtain the compact form
\begin{align}
I^{\,}_{\mathrm{WY}}(\beta)
&\approx
\frac{\beta^2}{2d}\|[\hat{H}^{\,}_{0},\hat{V}]\|^{2}_{\mathrm{HS}}
+\mathcal O(\beta^3),
\qquad
\delta V=\sqrt{2I^{\,}_{\mathrm{WY}}(\beta)}
\approx
\beta\frac{1}{\sqrt{d}}
\|[\hat{H}^{\,}_{0},\hat{V}]\|^{\,}_{\mathrm{HS}}
+\mathcal O(\beta^2).
\label{eq: QSL mixed states high T final app}
\end{align}
In Sec.~\ref{app:Ns_deltav_HT}, we show that
$\|[\hat{H}^{\,}_{0},\hat{V}]\|^{\,}_{\mathrm{HS}}\asymp \sqrt{Nd}$,
which implies $\delta V\asymp \beta\sqrt N$ as $\beta\to 0$ in the thermodynamic limit.

\subsubsection{$\chi^{\,}_{\mathrm{F}}$ in the high-temperature regime}

We now evaluate the fidelity susceptibility $\chi^{\,}_{\mathrm{F}}$
\eqref{eq: define fidelity susceptiblity simple form app} in the high-temperature regime.
Expanding the Boltzmann weights to second order in $\beta$ gives
\[
e^{-\beta E^{(0)}_{n}}
=
1-\beta E^{(0)}_{n} +\frac{\beta^2}{2} \left(E^{(0)}_{n}\right)^2+\mathcal{O}(\beta^3),
\]
hence
\begin{align}
\bigl(e^{-\beta E_m^{(0)}}-e^{-\beta E_n^{(0)}}\bigr)^2
&=
\beta^2\left(E^{(0)}_{m}-E^{(0)}_{n}\right)^2+\mathcal{O}(\beta^3).
\label{eq:Boltz_diff_sq_highT}
\end{align}

Moreover,
\begin{align}
Z^{\,}_{0}(2\beta)
=\Tr[e^{-2\beta \hat H_0}]
=
d+2\beta^2\Tr[\hat H_0^2]+\mathcal O(\beta^3),
\qquad (\Tr\big[\hat H_0\big]=0).
\label{eq:Z0_2beta_highT}
\end{align}
Substituting Eqs.~\eqref{eq:Boltz_diff_sq_highT} and \eqref{eq:Z0_2beta_highT} into
Eq.~\eqref{eq: define fidelity susceptiblity simple form app} yields
\begin{align}
\chi^{\,}_{\mathrm F}(\beta)
&=
\frac{2\beta^2}{Z^{\,}_{0}(2\beta)}
\sum_{\substack{m,n\\ m\neq n}}
|V^{\,}_{mn}|^2
+\mathcal{O}(\beta^3)
=
\frac{2\beta^2}{d}
\sum_{\substack{m,n\\ m\neq n}}
|V^{\,}_{mn}|^2
+\mathcal{O}(\beta^3).
\label{eq: fidelity susceptiblity high T app}
\end{align}
In Sec.~\ref{app:Ns_chiF_HT}, we show that
$\sum_{\substack{m,n\\ m\neq n}}|V^{\,}_{mn}|^2\asymp Nd$ for an extensive local $\hat V$,
which implies $\chi^{\,}_{\mathrm{F}}(\beta)\asymp \beta^2 N$ as $\beta\to 0$.

\subsubsection{Threshold driving rate $\Gamma^{\,}_{\mathrm{th}}$ in the high-temperature regime}

We finally evaluate the threshold driving rate $\Gamma^{\,}_{\mathrm{th}}$
\eqref{eq: define threshold driving rate}
using $\delta V$ \eqref{eq: QSL mixed states high T final app}
and $\chi^{\,}_{\mathrm{F}}$ \eqref{eq: fidelity susceptiblity high T app}:
\begin{align}
\Gamma^{\,}_{\mathrm{th}}
&=
\frac{
\beta\,d^{-1/2}\,\|[\hat{H}^{\,}_{0},\hat{V}]\|^{\,}_{2}
}{
(2\beta^2/d)\,
\sum_{\substack{m,n\\ m\neq n}}
|V^{\,}_{mn}|^2
}\,\alpha
=
\Gamma^{\,}_{N}\,
\frac{1}{\beta}\,
\frac{
\left(d^{-1/2}\,\|[\hat{H}^{\,}_{0},\hat{V}]\|^{\,}_{2}\right)/\delta V^{(0)}
}{
\left( (2/d)\sum_{\substack{m,n\\ m\neq n}}|V^{\,}_{mn}|^2 \right)/\chi^{(0)}_{\mathrm{F}}
}
\equiv
\Gamma^{\,}_{N}\,f^{\,}_{\mathrm{hi}\text{-}T}(\beta),
\nonumber
\end{align}
where
\[
f^{\,}_{\mathrm{hi}\text{-}T}(\beta)
=
\frac{c^{\,}_{2}}{\beta},
\qquad
c^{\,}_{2}:=
\frac{
\left(d^{-1/2}\,\|[\hat{H}^{\,}_{0},\hat{V}]\|^{\,}_{2}\right)/\delta V^{(0)}
}{
\left( (2/d)\sum_{\substack{m,n\\ m\neq n}}|V^{\,}_{mn}|^2 \right)/\chi^{(0)}_{\mathrm{F}}
}.
\]
This establishes the high-temperature scaling form in Theorem~\ref{thm:GammaThTempScaling}, with a model-dependent coefficient $c^{\,}_{2}>0$.

\subsection{System-size scaling of key quantities}
\label{app:NsScaling}

In this subsection, we derive the $N$-scaling of several quantities that appear above.
Though the derivation is quite general, one may imagine a system of $N$ spin-$1/2$ degrees of freedom, so that the Hilbert-space dimension is $d=2^N$.
We focus on the interpolating Hamiltonian
$
  \hat{H}^{\,}_{\lambda}=\hat{H}^{\,}_{0}+\lambda\hat{V}.
$
Unless stated otherwise, we assume a standard many-body setting in which
\begin{align}
\hat{H}^{\,}_0=\sum_{i=1}^N \hat{h}^{\,}_i,
\qquad
\hat{V}=\sum_{i=1}^N \hat{v}^{\,}_i,
\label{eq:Ns_local_sums}
\end{align}
with each $\hat{h}^{\,}_i$ and $\hat{v}^{\,}_i$ acting on $\mathcal{O}(1)$ sites (finite range) and having an $\mathcal{O}(1)$ operator norm.
We further assume that we are away from criticality so that low-energy gaps remain $\mathcal{O}(1)$,
and that no symmetry/selection-rule obstruction forces the relevant matrix elements to vanish.
Without loss of generality, we may take $\mathrm{Tr}[\hat{h}^{\,}_i]=\Tr[\hat{v}^{\,}_i]=0$ for all $i$ (by shifting each local term by a multiple of the identity).

Note that, for translation-invariant systems away from criticality,
the extensivity of \((\delta V)^{2}\) [Eq.~\eqref{eq: QSL mixed states simplified app}] and
\(\chi^{\,}_{\mathrm{F}}\) [Eq.~\eqref{eq: define fidelity susceptiblity simple form app}] can be established rigorously 
via exponential clustering of correlations%
~\cite{Hastings2006,Nachtergaele2006,Nachtergaele2006_CMP}, which in turn follows from the Lieb--Robinson bound%
~\cite{Lieb1972}.
Here, we instead give a more elementary, intuitive argument.

\subsubsection{Scaling of $|V^{\,}_{10}|$ in the low-temperature regime}
\label{app:Ns_deltav_LT}

We consider $|V^{\,}_{10}|=|\langle E_1^{(0)}|\hat{V}|E_0^{(0)}\rangle|$,
which enters the low-temperature expansion of $\delta V$ [Eq.~\eqref{eq: QSL mixed states low T final app}].
At low temperature, a convenient physical picture is that the lowest excited state that couples to the ground state
under $\hat V$ is a delocalized single-quasiparticle superposition of local excitations:
\[
|E_1^{(0)}\rangle\simeq \ket{k}
=
\frac{1}{\sqrt N}\sum_{j=1}^N e^{\mathrm{i}kj}\ket{j},
\]
where $\ket{j}$ is localized near site $j$ and $k$ is a crystal momentum.
Assuming that $\hat{v}^{\,}_j$ creates $\ket{j}$ from the ground state with an $\mathcal O(1)$ amplitude $a$
(namely, $\hat{v}^{\,}_j|E_0^{(0)}\rangle \simeq a\ket{j}$), 
we obtain
\begin{align}
\langle E_1^{(0)}|\hat{V}|E_0^{(0)}\rangle
&\simeq
\bra{k}\hat{V}|E_0^{(0)}\rangle
=
\sum^{N}_{j=1} \bra{k}\hat{v}_j|E_0^{(0)}\rangle
\simeq
a\sum^{N}_{j=1} \braket{k}{j}
=
a\sum^{N}_{j=1}
\left(
\frac{1}{\sqrt N}\sum^{N}_{\ell=1}e^{-\mathrm{i}k\ell}\braket{\ell}{j}
\right)
=
a\frac{1}{\sqrt N}\sum^{N}_{j=1} e^{-\mathrm{i}k j}.
\nonumber
\end{align}
The remaining factor is a standard geometric lattice sum. For a periodic chain, the crystal momentum is quantized, and one has the exact identity (see, e.g., Ref.~\cite{Ashcroft1976}):
\(
\sum_{j=1}^{N}e^{-\mathrm{i}kj}
=
N\,\delta^{\,}_{k,0}.
\)
Since we have chosen $|E_1^{(0)}\rangle$ such that $V^{\,}_{10}\neq 0$, the corresponding lattice sum does not cancel and therefore contributes a factor proportional to $N$. Hence,
\begin{align}
\left|\langle E_1^{(0)}|\hat{V}|E_0^{(0)}\rangle\right|
\asymp
\sqrt N.
\label{eq:Ns_M_scaling}
\end{align}

\paragraph*{Remark.}
A crude upper bound is $|V^{\,}_{10}|\le \|\hat V\|\asymp N$.
The estimate~\eqref{eq:Ns_M_scaling} is the anticipated scaling in generic gapped phases when the lowest relevant
excitation is a delocalized quasiparticle and the local creation amplitude $a$ is $\mathcal O(1)$.

\subsubsection{Scaling of $|V^{\,}_{10}|^{2}/\Delta^2$ in the low-temperature regime}
\label{app:Ns_chiF_LT}

The ratio $|V^{\,}_{10}|^2/\Delta^2$, which appears in the low-temperature expansion of
$\chi^{\,}_{\mathrm{F}}$ [Eq.~\eqref{eq: fidelity susceptiblity low T app}],
can be estimated using Eq.~\eqref{eq:Ns_M_scaling} together with $\Delta=\mathcal O(1)$ in a gapped phase:
\(
|V^{\,}_{10}|^2/\Delta^2\asymp N.
\)

\subsubsection{Scaling of $\|[\hat{H}^{\,}_{0},\hat{V}]\|^{\,}_{\mathrm{HS}}$ in the high-temperature regime}
\label{app:Ns_deltav_HT}

The Hilbert--Schmidt norm $\|[\hat{H}^{\,}_{0},\hat{V}^{\,}]\|_{\mathrm{HS}}$ appears in the high-temperature expansion of $\delta V$ 
[Eq.~\eqref{eq: QSL mixed states high T final app}].
To estimate its $N$-scaling, we use the locality assumptions in Eq.~\eqref{eq:Ns_local_sums}.
Write the commutator as
$
[\hat{H}^{\,}_{0},\hat{V}^{\,}]
=\sum_{i,j}[\hat{h}^{\,}_{i},\hat{v}^{\,}_{j}].
$
If $\hat{h}^{\,}_{i}$ and $\hat{v}^{\,}_{j}$ have disjoint supports, then they commute and $[\hat{h}^{\,}_{i},\hat{v}^{\,}_{j}]=0$.
Because the terms are finite range, for each $i$ only $\mathcal O(1)$ values of $j$ overlap, so there are only $\mathcal O(N)$ nonzero local
commutators in the sum. Denote these nonzero terms by $\{\hat{c}^{\,}_{a}\}_{a=1}^{\mathcal O(N)}$, where each $\hat{c}^{\,}_{a}$ is supported on
$\mathcal O(1)$ sites and has $\mathcal O(1)$ norm.
Then,
$
\|[\hat{H}^{\,}_{0},\hat{V}^{\,}]\|_{\mathrm{HS}}^{2}
=\sum_{a,b}\Tr\!\left((\hat{c}^{\,}_{a})^{\dagger}\hat{c}^{\,}_{b}\right).
$
By locality and tracelessness of the building blocks, $\Tr((\hat{c}^{\,}_{a})^{\dagger}\hat{c}^{\,}_{b})$ vanishes unless the supports
of $\hat{c}^{\,}_{a}$ and $\hat{c}^{\,}_{b}$ overlap to form an $\mathcal O(1)$-sized connected cluster; hence only $\mathcal O(N)$ pairs $(a,b)$ contribute.
For each such nonzero pair, the operator $(\hat{c}^{\,}_{a})^{\dagger}\hat{c}^{\,}_{b}$ acts nontrivially on only $\mathcal O(1)$ sites, so tracing over the
remaining $N-\mathcal O(1)$ inactive sites produces an overall multiplicative factor
\(
\Tr(\mathbb{I})=2^{\,N-\mathcal O(1)}\propto d.
\)
Therefore,
$
\|[\hat{H}^{\,}_{0},\hat{V}^{\,}]\|_{\mathrm{HS}}^{2}\asymp N d,$ so that 
$
\|[\hat{H}^{\,}_{0},\hat{V}^{\,}]\|_{\mathrm{HS}}\asymp \sqrt{N d}.
$

\subsubsection{Scaling of $\sum_{m\neq n}|V^{\,}_{mn}|^2$ in the high-temperature regime}
\label{app:Ns_chiF_HT}

Consider
\(
\sum_{\substack{m,n\\ m\neq n}} |V_{mn}|^{2}
=\mathrm{Tr}[\hat V^{\,2}]-\sum_{n}|V_{nn}|^{2},
\)
which enters the high-temperature expansion of $\chi^{\,}_{\mathrm F}$ [Eq.~\eqref{eq: fidelity susceptiblity high T app}].
We first estimate the basis-independent quantity $\mathrm{Tr}[\hat V^{\,2}]=\|\hat V\|_{\mathrm{HS}}^{2}$.
Write $\hat V=\sum_{i=1}^{N}\hat v^{\,}_{i}$ as a sum of local terms with $\mathcal O(1)$ norm, and shift each local term by the identity so that $\mathrm{Tr}[\hat v^{\,}_{i}]=0$.
Expanding $\mathrm{Tr}[\hat V^{\,2}]$ then yields a sum over pairs $(i,j)$.
If $\hat v^{\,}_{i}$ and $\hat v^{\,}_{j}$ have disjoint supports, the trace factorizes and vanishes because the local terms are traceless.
Hence, only $\mathcal O(N)$ pairs $(i,j)$ with overlapping supports contribute.
Moreover, each nonzero contribution acts nontrivially on only $\mathcal O(1)$ sites, so tracing over the remaining $N-\mathcal O(1)$ inactive sites produces an overall factor proportional to $d=2^{N}$.
Therefore, generically,
\(
\mathrm{Tr}[\hat V^{\,2}]\asymp N d.
\)

Next, the diagonal contribution satisfies the general bound
\(
0\le \sum_{n}|V^{\,}_{nn}|^{2}\le \mathrm{Tr}[\hat V^{\,2}],
\)
and therefore the off-diagonal sum is also bounded above by $\mathrm{Tr}[\hat V^{\,2}]$.
For a generic local drive (i.e., excluding symmetry-enforced selection rules or fine-tuned cases where $\hat V$ is (block-)diagonal in the $\hat H^{\,}_0$ eigenbasis), we do not expect the diagonal part $\sum_{n}|V_{nn}|^{2}$ alone to account for essentially all of $\mathrm{Tr}[\hat V^{\,2}]=\sum_{m,n}|V_{mn}|^{2}$.
Under this genericity assumption, the off-diagonal contribution remains extensive, so
\(
\sum_{\substack{m,n\\ m\neq n}} |V_{mn}|^{2}\asymp Nd.
\)


\section{Transverse-field Ising chain and quantum XY chain: exact results via the transfer-matrix method}
\label{sec:tfic_qxyc_transfer_matrix}

We consider driven spin chains with periodic boundary conditions, whose Hamiltonian takes the form
\(
\hat{H}^{\,}_{\lambda}=\hat{H}^{\,}_{0}+\lambda(t)\,\hat{V},
\)
where the unperturbed Hamiltonian is of Ising type,
\(
\hat{H}^{\,}_{0}=-J\sum^{N}_{j=1}Z^{\,}_{j}Z^{\,}_{j+1},
\)
with \(J>0\). The control parameter \(\lambda(t)\) is time-dependent and monotonic, with \(\lambda(0)=0\).

We study two representative choices of the driving term \(\hat V\).
For the transverse-field Ising chain (TFIC), we take
\(
\hat{V}^{\,}_{\mathrm{TFIC}}= -J\sum^{N}_{j=1}X^{\,}_{j}
\)
with
\(
\lambda(t)= \frac{h(t)}{J},
\)
where \(h(t)\) is the time-dependent transverse field, and we assume \(h(0)=0\).
For the quantum XY chain (QXYC), we take
\(
\hat{V}^{\,}_{\mathrm{QXYC}}= -J\sum^{N}_{j=1}
\left(
X^{\,}_{j}X^{\,}_{j+1}
-
Z^{\,}_{j}Z^{\,}_{j+1}
\right)
\)
with
\(
\lambda(t)= \frac{1+\gamma(t)}{2},
\)
where \(\gamma(t)\) is the time-dependent anisotropy parameter, and we assume \(\gamma(0)=-1\).

We assume that the system is initially in a thermal state at inverse temperature \(\beta\),
$
\hat{\rho}^{\,}_{0}(\beta)=\frac{e^{-\beta\hat{H}^{\,}_{0}}}{Z^{\,}_{0}(\beta)}
$
with
$
Z^{\,}_{0}(\beta)\equiv \Tr\!e^{-\beta\hat{H}^{\,}_{0}}.
$
The Boltzmann weight can be factorized as
$
e^{-\beta\hat{H}^{\,}_{0}}
= e^{\beta J\sum^{N}_{j=1}Z^{\,}_{j}Z^{\,}_{j+1}}
=
\prod^{N}_{j=1}
e^{\beta JZ^{\,}_{j}Z^{\,}_{j+1}}.
$

For later convenience, we evaluate \(Z^{\,}_{0}(\beta)\).
Since all \(Z^{\,}_{j}Z^{\,}_{j+1}\) commute, \(\hat H_0\) is diagonal in the \(Z\)-product basis
\(|\mathbf{s}\rangle=|s_1,\ldots,s_N\rangle\), where \(s^{\,}_{j}=\pm1\) are the eigenvalues of \(Z^{\,}_{j}\).
Hence,
\begin{align}
Z^{\,}_{0}(\beta)
=\Tr\!\left(e^{-\beta \hat H_0}\right)
=\sum_{\{s^{\,}_{j}=\pm1\}}
\exp\!\Big(\beta J\sum_{j=1}^N s^{\,}_{j} s_{j+1}\Big).
\nonumber
\end{align}
Introduce the transfer matrix \(T\in\mathbb{C}^{2\times2}\) with entries
\begin{align}
T^{\,}_{s,s'}=\exp(\beta J\, s s')
=
\begin{pmatrix}
e^{\beta J} & e^{-\beta J}\\
e^{-\beta J} & e^{\beta J}
\end{pmatrix},
\nonumber
\end{align}
so that \(Z^{\,}_{0}(\beta)=\Tr(T^N)\).
The eigenvalues of \(T\) are
$
\lambda_{\pm}=e^{\beta J}\pm e^{-\beta J},
$
i.e.
$
\lambda_{+}=2\cosh(\beta J)
$
and
$
\lambda_{-}=2\sinh(\beta J).
$
Therefore,
\begin{align}
Z^{\,}_{0}(\beta)=\lambda_{+}^{N}+\lambda_{-}^{N}
=(2\cosh(\beta J))^{N}+(2\sinh(\beta J))^{N}
=2^{N}\Big(\cosh^{N}(\beta J)+\sinh^{N}(\beta J)\Big).
\label{eq:Z0_transfer_matrix}
\end{align}

With these preliminaries, we now evaluate both \(\delta V\) [Eq.~\eqref{eq: QSL mixed states simplified}]
and \(\chi^{\,}_{\mathrm{F}}\) [Eq.~\eqref{eq: define fidelity susceptiblity}] for the two models at arbitrary temperature.

\subsection{Evaluation of $\delta V$}

Our first goal is to compute \(\delta V\) [Eq.~\eqref{eq: QSL mixed states simplified}] for both models:
\begin{align}
\delta V
= \frac{\|[\hat{\rho}^{\,}_{0},\hat{V}]\|^{\,}_{2}}{\|\hat{\rho}^{\,}_{0}\|^{\,}_{2}}
= \sqrt{2}\,
\sqrt{\frac{\Tr\!\left(\hat{\rho}^{2}_{0}\hat{V}^2\right)
-\Tr\!\left[\left(\hat{\rho}^{\,}_{0}\hat{V}\right)^2\right]}
{\Tr\!\left(\hat{\rho}^{2}_{0}\right)}}\, .
\label{eq: recall deltaV}
\end{align}
It is convenient to rewrite \(\delta V\) in terms of unnormalized Boltzmann weights.
Define \(K\equiv \beta J\) and use
\(\hat{\rho}_0=\frac{1}{Z^{\,}_{0}(\beta)}e^{K\sum_j Z^{\,}_{j}Z^{\,}_{j+1}}\). One finds
\begin{align}
\Tr(\hat{\rho}_0^2)&=\frac{Z^{\,}_{0}(2\beta)}{Z^{\,}_{0}(\beta)^2},\\
\Tr(\hat{\rho}_0^2\hat V^2)&=\frac{1}{Z^{\,}_{0}(\beta)^2}\Tr\!\left(e^{2K\sum_j Z^{\,}_{j}Z^{\,}_{j+1}}\hat V^2\right),\\
\Tr\!\left[(\hat{\rho}_0\hat V)^2\right]
&=\frac{1}{Z^{\,}_{0}(\beta)^2}\Tr\!\left(e^{K\sum_i Z^{\,}_{i}Z^{\,}_{i+1}}\hat V\,e^{K\sum_j Z^{\,}_{j}Z^{\,}_{j+1}}\hat V\right).
\nonumber
\end{align}
Therefore,
\begin{align}
\left(\delta V\right)^2
=
\frac{2}{Z^{\,}_{0}(2\beta)}
\Bigg[
\Tr\!\left(e^{2K\sum_i Z^{\,}_{i}Z^{\,}_{i+1}}\hat V^2\right)
-
\Tr\!\left(e^{K\sum_i Z^{\,}_{i}Z^{\,}_{i+1}}\hat V\,e^{K\sum_j Z^{\,}_{j}Z^{\,}_{j+1}}\hat V\right)
\Bigg].
\label{eq:deltaV2_unnormalized}
\end{align}

\subsubsection{Evaluation of $\delta V$ in TFIC}

For the TFIC drive,
$
\hat V^{\,}_{\rm TFIC}=-J\sum_{j=1}^N X^{\,}_{j}.
$

\paragraph*{Step 1: Compute } \(\Tr(e^{2K\sum^{\,}_{i} Z^{\,}_{i}Z^{\,}_{i+1}}\hat V^2)\).
Expanding \(\hat V^2_{\rm TFIC}\) gives
$
\hat V^{2}_{\rm TFIC}
=J^2\sum_{j,k=1}^N X^{\,}_{j}X^{\,}_{k}.
$
Since \(e^{2K\sum Z^{\,}_{i}Z^{\,}_{i+1}}\) is diagonal in the \(Z\)-basis,
only ``on-site'' terms of \(X^{\,}_{j}X^{\,}_{k}\) contribute to the trace.
For \(j\neq k\), \(X^{\,}_{j}X^{\,}_{k}\) flips two distinct spins and has vanishing diagonal elements,
while for \(j=k\), \((X^{\,}_{j})^2=\mathbb{I}\).
Hence,
\begin{align}
\Tr\!\left(e^{2K\sum_{i} Z^{\,}_{i}Z^{\,}_{i+1}}\hat V_{\rm TFIC}^2\right)
=J^2\sum_{j=1}^N \Tr\!\left(e^{2K\sum_{i} Z^{\,}_{i}Z^{\,}_{i+1}}\mathbb{I}\right)
=NJ^2 Z^{\,}_{0}(2\beta).
\label{eq:TFIC_term1}
\end{align}

\paragraph*{Step 2: Compute } \(\Tr(e^{K\sum_{i} Z^{\,}_{i}Z^{\,}_{i+1}}\hat V\,e^{K\sum_{j} Z_{j}Z^{\,}_{j+1}}\hat V)\).
Similarly,
\begin{align}
A:= \Tr\!\left(e^{K\sum_{i} Z^{\,}_{i}Z^{\,}_{i+1}}\hat V_{\rm TFIC}\,e^{K\sum_{j} Z_{j}Z^{\,}_{j+1}}\hat V_{\rm TFIC}\right)
&=
J^2\sum_{k,\ell=1}^N
\Tr\!\left(e^{K\sum_{i} Z^{\,}_{i}Z^{\,}_{i+1}}X^{\,}_{k}\,e^{K\sum_{j} Z_{j}Z^{\,}_{j+1}}X^{\,}_{\ell}\right)
\nonumber\\
&=J^2\sum_{k=1}^N
\Tr\!\left(e^{K\sum_{i} Z^{\,}_{i}Z^{\,}_{i+1}}X^{\,}_{k}\,e^{K\sum_{j} Z_{j}Z^{\,}_{j+1}}X^{\,}_{k}\right),
\nonumber
\end{align}
where we have used the fact that, in the \(Z\)-basis, the trace is nonzero only for \(k=\ell\).
By translational invariance, each term has the same value, thus
\begin{align}
A
=&\, NJ^2
\Tr\!\left(e^{K\sum_{i} Z^{\,}_{i}Z^{\,}_{i+1}}X^{\,}_{1}\,e^{K\sum_{j} Z_{j}Z^{\,}_{j+1}}X^{\,}_{1}\right)
\nonumber\\
=&\,
NJ^2 \Tr\!\left[
\left(e^{KZ_{1}\left(Z_{2}+Z^{\,}_{N}\right)}X^{\,}_{1}\,e^{KZ_{1}\left(Z_{2}+Z^{\,}_{N}\right)}X^{\,}_{1}\right)
\left(e^{2K\sum^{N-1}_{i=2} Z^{\,}_{i}Z^{\,}_{i+1}}\right)
\right]
\nonumber.
\end{align}
To proceed, recall that $|s^{\,}_{1}\rangle$ (with $s^{\,}_{1}=\pm 1$) is the eigenstate of $Z^{\,}_{1}$ with eigenvalue $s^{\,}_{1}$,
and $X^{\,}_{1}|s^{\,}_{1}\rangle=|-s^{\,}_{1}\rangle$.
We use this basis to perform the trace over the first spin:
\begin{align}
\Tr^{\,}_{1}
\left(e^{KZ_{1}\left(Z_{2}+Z^{\,}_{N}\right)}X^{\,}_{1}\,e^{KZ_{1}\left(Z_{2}+Z^{\,}_{N}\right)}X^{\,}_{1}\right)
=&\, \sum^{\,}_{s^{\,}_{1}=\pm 1}
\langle s^{\,}_{1}|
\left(e^{KZ_{1}\left(Z_{2}+Z^{\,}_{N}\right)}X^{\,}_{1}\,e^{KZ_{1}\left(Z_{2}+Z^{\,}_{N}\right)}X^{\,}_{1}\right)
|s^{\,}_{1}\rangle
\nonumber\\
=&\, \sum^{\,}_{s^{\,}_{1}=\pm 1}
e^{Ks^{\,}_{1}\left(Z_{2}+Z^{\,}_{N}\right)}
e^{-Ks^{\,}_{1}\left(Z_{2}+Z^{\,}_{N}\right)}
=2.
\nonumber
\end{align}
Therefore, we obtain
$
A
=
2NJ^2
\Tr^{\,}_{2,3,\cdots,N}\left(e^{2K\sum^{N-1}_{i=2} Z^{\,}_{i}Z^{\,}_{i+1}}\right)
=
2NJ^{2} Q^{\,}_{N-1}(2\beta),
$
where
\begin{align}
Q^{\,}_{N-1}(2\beta)\equiv
\sum^{\,}_{s^{\,}_{2},\cdots, s^{\,}_{N}=\pm1}
\exp\!\Big(2K\sum^{N-1}_{i=2} s^{\,}_{i}s^{\,}_{i+1}\Big).
\label{eq: define Q TFIC}
\end{align}
One recognizes that $Q^{\,}_{N-1}(2\beta)$ is the partition function of the open Ising chain with $N-1$ spins at inverse temperature $2\beta$.
Its value can be obtained using a recursive method (see, e.g., Ref.~\cite{Nishimori2010}):
$
Q^{\,}_{N-1}(2\beta)
=
2\left[
2\cosh(2K)
\right]^{N-2}.
$
Therefore,
\begin{align}
A=\Tr\!\left(e^{K\sum_{i} Z^{\,}_{i}Z^{\,}_{i+1}}\hat V_{\rm TFIC}\,e^{K\sum_{j} Z_{j}Z^{\,}_{j+1}}\hat V_{\rm TFIC}\right)
=
2NJ^2
Q^{\,}_{N-1}(2\beta)
=
4NJ^{2}\left[
2\cosh(2K)
\right]^{N-2}.
\label{eq:Qx_closed}
\end{align}

\paragraph*{Step 3: Assemble \((\delta V)^2\).}
Plugging \eqref{eq:TFIC_term1} and \eqref{eq:Qx_closed} into \eqref{eq:deltaV2_unnormalized} yields
\begin{align}
(\delta V^{\,}_{\rm TFIC})^2
=
2NJ^2\left(1-2\frac{Q^{\,}_{N-1}(2\beta)}{Z^{\,}_{0}(2\beta)}\right).
\nonumber
\end{align}
Using \(Z^{\,}_{0}(2\beta)=(2\cosh(2\beta J))^N+(2\sinh(2\beta J))^N\) [Eq.~\eqref{eq:Z0_transfer_matrix} with $\beta\to 2\beta$], we obtain the closed form
\begin{align}
\delta V^{\,}_{\rm TFIC}
=
\sqrt{2N}J
\tanh(2\beta J)\,
\left(
\frac{1+\tanh^{N-2}(2\beta J)}{1+\tanh^N(2\beta J)}
\right)^{1/2}.
\label{eq:deltaV_TFIC_closed}
\end{align}
In particular, for fixed finite \(\beta\) and large \(N\) (so that \(\tanh^N(2\beta J)\to 0\)),
\begin{align}
\delta V_{\rm TFIC}\xrightarrow[N\to\infty]{}
\sqrt{2N}\,J\,\tanh(2\beta J).
\label{eq:deltaV_TFIC_closed large N}
\end{align}

\subsubsection{Evaluation of $\delta V$ in QXYC}

We now consider the quantum XY drive,
$
\hat V^{\,}_{\rm QXYC}
= \hat V^{\,}_x+\hat V^{\,}_z,
$
with
$
\hat V_x=-J\sum_{j=1}^N X^{\,}_{j}X^{\,}_{j+1},
$
and 
$
\hat V_z=+J\sum_{j=1}^N Z^{\,}_{j}Z^{\,}_{j+1}.
$
Since \(\hat V^{\,}_z\) is diagonal in the \(Z\)-basis, it commutes with \(\hat{\rho}^{\,}_0\).
Therefore,
\begin{align}
\delta V^{\,}_{\mathrm{QXYC}}
=
\frac{\|[\hat{\rho}^{\,}_{0},\hat{V}^{\,}_{\mathrm{QXYC}}]\|^{\,}_{2}}{\|\hat{\rho}^{\,}_{0}\|^{\,}_{2}}
=
\frac{\|[\hat{\rho}^{\,}_{0},\hat{V}^{\,}_{x}]\|^{\,}_{2}}{\|\hat{\rho}^{\,}_{0}\|^{\,}_{2}}
\equiv \delta V^{\,}_{x}.
\nonumber
\end{align}
Now repeat the TFIC steps with the bond-flip operators
\(\hat B^{\,}_j\equiv X^{\,}_{j}X^{\,}_{j+1}\), so that \(\hat V^{\,}_x=-J\sum_j \hat B^{\,}_{j}\).
Exactly as before, only identical bonds contribute to both traces in Eq.~\eqref{eq:deltaV2_unnormalized}.
The analog of \(Q^{\,}_{N-1}(2\beta)\) becomes
\begin{align}
\tilde{Q}(2\beta)
\equiv
\sum_{\{s^{\,}_{j}=\pm1\}}
\exp\!\Big(K\sum_{j=1}^N s^{\,}_{j}s_{j+1}\Big)\,
\exp\!\Big(K\sum_{j=1}^N \tilde s^{\,}_{j} \tilde s_{j+1}\Big),
\label{eq: define Qtilde QXYC}
\end{align}
where \(\{\tilde s^{\,}_{j}\}\) differs from \(\{s^{\,}_{j}\}\) by flipping two adjacent spins, e.g.,
\(\tilde s^{\,}_{1}=-s^{\,}_{1}\), \(\tilde s^{\,}_{2}=-s^{\,}_{2}\), and \(\tilde s^{\,}_{j}=s^{\,}_{j}\) for \(j\ge 3\).
A direct computation shows that
\begin{align}
\sum_{j=1}^N s^{\,}_{j}s_{j+1}+\sum_{j=1}^N \tilde s^{\,}_{j}\tilde s_{j+1}
=
2s^{\,}_{1}s^{\,}_{2}
+
2\sum_{j=3}^{N-1}s^{\,}_{j}s_{j+1},
\nonumber
\end{align}
i.e., the two bonds connecting the flipped dimer to the rest cancel, while all remaining bonds are doubled to \(2K\).
Hence, the sum factorizes, giving
\begin{align}
\tilde{Q}(2\beta)
=
\sum_{s^{\,}_{1},s^{\,}_{2}=\pm1}e^{2K s^{\,}_{1}s^{\,}_{2}}
\sum_{s^{\,}_{3},\ldots,s^{\,}_{N}=\pm1}
\exp\!\Big(2K\sum_{j=3}^{N-1}s^{\,}_{j}s_{j+1}\Big)
=
\big(4\cosh(2K)\big)\times
2\big(2\cosh(2K)\big)^{N-3}
=
4\big(2\cosh(2K)\big)^{N-2}.
\nonumber
\end{align}
Consequently, one obtains the same \(\delta V\) as in the TFIC case [Eq.~\eqref{eq:deltaV_TFIC_closed}]:
$
\delta V^{\,}_{\rm QXYC}
=
\delta V^{\,}_{\rm TFIC}.
$

\subsection{Evaluation of $\chi^{\,}_{\mathrm{F}}$}

Next, we compute the fidelity susceptibility $\chi^{\,}_{\mathrm{F}}$ using Eq.~\eqref{eq: define fidelity susceptiblity simple form app} for both
the TFIC and the QXYC models:
\begin{align}
\chi^{\,}_{\mathrm F}
=
\frac{2}{Z^{\,}_{0}(2\beta)}
\sum_{\substack{m,n\\ m\neq n}}
(p^{\,}_m-p^{\,}_n)^2
\frac{|\langle E_m^{(0)}|\hat V|E_n^{(0)}\rangle|^2}{(E_m^{(0)}-E_n^{(0)})^2},
\qquad p^{\,}_n:=e^{-\beta E_n^{(0)}}.
\label{eq: define fidelity susceptiblity simple form app2}
\end{align}
Here and throughout this section, we exclude degenerate pairs with $E_m^{(0)}=E_n^{(0)}$ (i.e., the $\Delta E=0$ channels),
consistent with using the non-degenerate perturbative form of $\chi^{\,}_{\mathrm F}$.

\subsubsection{Evaluation of $\chi^{\,}_{\mathrm{F}}$ in TFIC}

For the TFIC drive given by
$
\hat H_0=-J\sum_{j=1}^N Z^{\,}_{j}Z^{\,}_{j+1}
$ and
$
\hat V_{\rm TFIC}=-J\sum_{j=1}^N X^{\,}_{j}
$ 
with
$(J>0),
$
an eigenbasis of \(\hat H_0\) is the \(Z\)-product basis
\(\ket{\mathbf s}:=\ket{s_1,\dots,s_N}\) with \(s^{\,}_{j}=\pm 1\),
and eigenenergies
$
E(\mathbf s):=-J\sum_{j=1}^N s^{\,}_{j}s_{j+1}.
$
The operator \(X^{\,}_{j}\) flips spin \(j\). Denote by \(\mathbf s^{(j)}\) the configuration obtained from \(\mathbf s\)
by flipping \(s^{\,}_{j}\to -s^{\,}_{j}\).
Then, the operator $\hat V_{\rm TFIC}$ connects $|\mathbf{s}\rangle$ only to the configuration \(|\mathbf s^{(j)}\rangle\) obtained by flipping $j$, with matrix element
\begin{align}
\langle \mathbf s|\hat V_{\rm TFIC}|\mathbf s^{(j)}\rangle=-J,
\qquad
\langle \mathbf s|\hat V_{\rm TFIC}|\mathbf s'\rangle=0
\quad \text{if } \mathbf s'\neq \mathbf s^{(j)}\ \forall j.
\nonumber
\end{align}
Hence, the sum in \eqref{eq: define fidelity susceptiblity simple form app2} reduces to single-spin-flip pairs.

The energy difference associated with flipping spin \(j\) is
$
\Delta E_j(\mathbf s)
\equiv E(\mathbf s)-E(\mathbf s^{(j)})
=-2J\,s^{\,}_{j}(s_{j-1}+s_{j+1}).
$
Therefore, \(\Delta E_j=-4Js^{\,}_{j}s^{\,}_{j+1}\) when \(s_{j-1}=s_{j+1}\), while \(\Delta E_j=0\) when \(s_{j-1}=-s_{j+1}\).
The latter contribution is excluded in the sum in Eq.~\eqref{eq: define fidelity susceptiblity simple form app2}.
Hence, we may restrict to the case \(s_{j-1}=s_{j+1}\).

Using \(|\langle \mathbf s|\hat V_{\rm TFIC}|\mathbf s^{(j)}\rangle|^2=J^2\) and \((\Delta E_j)^2=(4J)^2\) on the contributing configurations,
and writing the sum over \((m,n)\) as a sum over \((\mathbf{s},j)\),
Eq.~\eqref{eq: define fidelity susceptiblity simple form app2} becomes
\begin{align}
\chi^{\rm TFIC}_{\mathrm F}
&=
\frac{2}{Z^{\,}_{0}(2\beta)}
\sum_{\mathbf{s}}\sum_{j=1}^N
I^{\,}_{j}(\mathbf{s})
\big(p(\mathbf{s})-p(\mathbf{s}^{(j)})\big)^2
\frac{J^2}{16J^2},
\nonumber
\end{align}
where \(I^{\,}_{j}(\mathbf{s})\) denotes the indicator function:
\begin{align}
I^{\,}_{j}(\mathbf{s})
:=
\frac{1+s_{j-1}s_{j+1}}{2}
=
\begin{cases}
1,& s_{j-1}=s_{j+1},\\
0,& s_{j-1}=-s_{j+1}.
\end{cases}
\nonumber
\end{align}

Now use
$
\big(p(\mathbf{s})-p(\mathbf{s}^{(j)})\big)^2
=
p(\mathbf{s})^2\Big(1-e^{-\beta(E(\mathbf{s}^{(j)})-E(\mathbf{s}))}\Big)^2
=
e^{-2\beta E(\mathbf{s})}\Big(1-e^{+\beta\Delta E_j(\mathbf{s})}\Big)^2,
$
so that the entire sum becomes an expectation value in the \(2\beta\) thermal average:
\begin{align}
\chi^{\rm TFIC}_{\mathrm F}
&=
\frac{1}{8Z^{\,}_{0}(2\beta)}
\sum_{\mathbf{s}}\sum_{j=1}^N
I^{\,}_{j}(\mathbf{s})
e^{-2\beta E(\mathbf{s})}\Big(1-e^{+\beta\Delta E_j(\mathbf{s})}\Big)^2
\nonumber\\
&=
\frac{1}{8}
\sum_{j=1}^N
\Big\langle
I^{\,}_{j}(\mathbf{s})
\big(1-e^{+\beta\Delta E_j(\mathbf{s})}\big)^2
\Big\rangle_{2\beta},
\qquad
\langle \cdots\rangle_{2\beta}:=
\frac{1}{Z^{\,}_{0}(2\beta)}\sum_{\mathbf{s}} e^{-2\beta E(\mathbf{s})}(\cdots).
\nonumber
\end{align}
Due to translation invariance, we may replace \(\sum^{N}_{j=1}\) by \(N\) times a representative site:
\begin{align}
\chi^{\rm TFIC}_{\mathrm F}
=
\frac{N}{8}
\Big\langle
I^{\,}_{1}(\mathbf{s})\Big(1-e^{-4\beta J\,s^{\,}_{1}s_{2}}\Big)^2
\Big\rangle_{2\beta}.
\label{eq:chiF_TFIC_to_classical_sum}
\end{align}

We evaluate Eq.~\eqref{eq:chiF_TFIC_to_classical_sum} using the transfer-matrix method.
Let \(K^{\,}_{2}\equiv 2\beta J\). Then
\begin{align}
Z^{\,}_{0}(2\beta)=\sum_{\mathbf s} e^{K^{\,}_{2}\sum_{i=1}^N s^{\,}_is^{\,}_{i+1}}
=\Tr(T^N)=\lambda^{N}_{+}+\lambda^{N}_{-},
\qquad
T^{\,}_{s,s'}:=e^{K^{\,}_{2}ss'}.
\nonumber
\end{align}
The eigenvalues of the transfer matrix \(T\) are
$
\lambda^{\,}_+=2\cosh K^{\,}_{2}
$
and
$
\lambda^{\,}_-=2\sinh K^{\,}_{2}.
$
Define
\begin{align}
S^{\,}_1:=
\sum_{\mathbf s}
e^{K^{\,}_{2}\sum_i s_is_{i+1}}
I^{\,}_{1}(\mathbf{s})\Big(1-e^{-2K^{\,}_{2}\,s^{\,}_{1}s_{2}}\Big)^2,
\nonumber
\end{align}
so that Eq.~\eqref{eq:chiF_TFIC_to_classical_sum} can be written as
\begin{align}
\chi^{\rm TFIC}_{\mathrm F}
=
\frac{NS^{\,}_{1}}{8\,Z^{\,}_{0}(2\beta)}.
\label{eq:chiF_TFIC_to_classical_sum N times}
\end{align}
We are left with evaluating $S^{\,}_{1}$ using the transfer-matrix method.

Writing \(e^{K^{\,}_{2}\sum_i s_is_{i+1}}=\prod_{i=1}^N T_{s_i,s_{i+1}}\) and noting that
\(I^{\,}_{1}(\mathbf s)=1\) enforces \(s_N=s_2\), we obtain
\begin{align}
S_1
=
\sum_{s_1,s_2=\pm 1}
T_{s_1,s_2}\,(T^{N-2})_{s_2,s_2}\,T_{s_2,s_1}\,
\Big(1-e^{-2K^{\,}_{2} s_1s_2}\Big)^2.
\nonumber
\end{align}
The prefactor simplifies as
\begin{align}
T_{s_1,s_2}T_{s_2,s_1}\Big(1-e^{-2K^{\,}_{2} s_1s_2}\Big)^2
&=
e^{2K^{\,}_{2} s_1s_2}\Big(1-e^{-2K^{\,}_{2} s_1s_2}\Big)^2
=
\Big(e^{K^{\,}_{2} s_1s_2}-e^{-K^{\,}_{2} s_1s_2}\Big)^2
=
4\sinh^2 K^{\,}_{2},
\nonumber
\end{align}
which is independent of \(s_1,s_2\). Therefore,
\begin{align}
S^{\,}_1
&=
4\sinh^2 K^{\,}_{2}\sum_{s_1,s_2}(T^{N-2})_{s_2,s_2}
=
8\sinh^2 K^{\,}_{2}\;\Tr(T^{N-2})
=
8\sinh^2 K^{\,}_{2}\;\big(\lambda_+^{N-2}+\lambda_-^{N-2}\big).
\label{eq:S1_closed}
\end{align}
Substituting \eqref{eq:S1_closed} into \eqref{eq:chiF_TFIC_to_classical_sum N times} gives
$
\chi^{\rm TFIC}_{\mathrm F}
=
N\sinh^2 K^{\,}_{2}\,
\frac{\lambda_+^{N-2}+\lambda_-^{N-2}}{\lambda_+^{N}+\lambda_-^{N}}.
$
Using \(\lambda_+=2\cosh K^{\,}_{2}\), \(\lambda_-=2\sinh K^{\,}_{2}\), and \(\lambda_-/\lambda_+=\tanh K^{\,}_{2}\), we obtain
\begin{align}
\chi^{\rm TFIC}_{\mathrm F}
=
\frac{N}{4}\,\tanh^2(2\beta J)\,
\frac{1+\tanh^{N-2}(2\beta J)}{1+\tanh^{N}(2\beta J)}.
\label{eq:chiF_TFIC_final}
\end{align}

\subsubsection{Evaluation of $\chi^{\,}_{\mathrm{F}}$ in QXYC}

For the QXYC drive,
$
\hat V^{\,}_{\rm QXYC}
=
-J\sum_{j=1}^N\Big(X^{\,}_{j}X^{\,}_{j+1}-Z^{\,}_{j}Z^{\,}_{j+1}\Big)
\equiv \hat V^{\,}_x+\hat V^{\,}_z,
$
with
$
\hat V^{\,}_x=-J\sum_{j=1}^NX^{\,}_{j}X^{\,}_{j+1}
$
and
$
\hat V^{\,}_z=+J\sum_{j=1}^NZ^{\,}_{j}Z^{\,}_{j+1}.
$
In Eq.~\eqref{eq: define fidelity susceptiblity simple form app2}, only off-diagonal matrix elements \(\langle m|\hat V|n\rangle\) with \(m\neq n\) contribute.
Since \(\hat V_z\) is diagonal in the \(Z\)-basis, it does not contribute. Hence, \(\chi^{\rm QXYC}_{\mathrm F}=\chi^{\,}_{\mathrm F}[\hat V_x]\).

Now \(X^{\,}_{j}X^{\,}_{j+1}\) flips the adjacent spins \(j\) and \(j+1\), mapping \(\mathbf{s}\mapsto \mathbf{s}^{(j,j+1)}\),
with matrix element \(\langle \mathbf{s}|\hat V_x|\mathbf{s}^{(j,j+1)}\rangle=-J\).
The corresponding energy difference under \(\hat H_0\) is
$
\Delta E_{j,j+1}(\mathbf{s})
:=
E(\mathbf{s})-E(\mathbf{s}^{(j,j+1)})
=
-2J\Big(s_{j-1}s^{\,}_{j}+s_{j+1}s_{j+2}\Big),
$
so \(\Delta E_{j,j+1}\in\{0,\pm4J\}\).
As above, the $\Delta E_{j,j+1}(\mathbf s)=0$ sector is excluded in the sum in Eq.~\eqref{eq: define fidelity susceptiblity simple form app2}, so we restrict to the $\Delta E=\pm 4J$ channels.
The computation proceeds exactly as in the TFIC case, with the only change that the relevant indicator is now expressed in terms of the bond variables
\(b_j:=s^{\,}_{j}s_{j+1}\) (note that \(E(\mathbf{s})=-J\sum_j b_j\)).
In these variables, \(X^{\,}_{j}X^{\,}_{j+1}\) flips precisely two bonds, \(b_{j-1}\to -b_{j-1}\) and \(b_{j+1}\to -b_{j+1}\),
while leaving the rest unchanged. This leads to the same transfer-matrix evaluation as above, and yields the same closed form:
$
\chi^{\rm QXYC}_{\mathrm F} = \chi^{\rm TFIC}_{\mathrm F}.
$

\subsection{Threshold driving rate $\Gamma^{\,}_{\mathrm{th}}$}

Substituting Eqs.~\eqref{eq:deltaV_TFIC_closed} and \eqref{eq:chiF_TFIC_final} into Eq.~\eqref{eq: define threshold driving rate}, we obtain for both the TFIC and QXYC
\begin{align}
\Gamma^{\,}_{\mathrm{th}}
&=
\frac{4\sqrt{2}J}{\sqrt{N}}\,
\coth(2\beta J)
\left(
\frac{1+\tanh^N(2\beta J)}{1+\tanh^{N-2}(2\beta J)}
\right)^{1/2}\alpha
\equiv
\Gamma^{\,}_{N} f^{\,}_{N}(\beta),
\nonumber
\end{align}
with
\[
\Gamma^{\,}_{N}\equiv\frac{4\sqrt{2}J}{\sqrt{N}}\alpha,
\qquad
f^{\,}_{N}(\beta)\equiv
\coth(2\beta J)
\left(
\frac{1+\tanh^N(2\beta J)}{1+\tanh^{N-2}(2\beta J)}
\right)^{1/2}.
\]

For any fixed finite $\beta>0$, $\tanh(2\beta J)<1$ and thus
\[
f(\beta):=\lim_{N\to\infty} f^{\,}_{N}(\beta)=\coth(2\beta J),
\qquad
f(\beta)=
\begin{cases}
1+2e^{-4\beta J}+\cdots, & \beta\to\infty,\\[2pt]
(2\beta J)^{-1}+\cdots, & \beta\to 0.
\end{cases}
\]

It is instructive to reverse the order of the thermodynamic limit and the low-temperature limit.
Taking the extreme-temperature limits at fixed $N$ gives instead
\begin{align}
\lim_{\beta\to\infty} f^{\,}_{N}(\beta)
&=
1+e^{-4\beta J}
+\left(N-\frac{1}{2}\right)e^{-8\beta J}
+\left(N-\frac{1}{2}\right)e^{-12\beta J}
-\left(\frac{1}{3}N^3-\frac{3}{2}N^2+\frac{7}{6}N-\frac{3}{8}\right)e^{-16\beta J}
+\mathcal{O}\!\big(e^{-20\beta J}\big),
\nonumber\\
\lim_{\beta\to0} f^{\,}_N(\beta)
&=
\frac{1}{2\beta J}+\frac{2\beta J}{3}-\frac{(2\beta J)^{3}}{45}
+\mathcal{O}\!\big((\beta J)^{5}\big),
\nonumber
\end{align}
showing that the thermodynamic and low-temperature limits do not commute.

\section{Mixed-field Ising chain: exact results via the transfer-matrix method}
\label{sec:mfic_transfer_matrix}

We consider the mixed-field Ising chain (MFIC) with periodic boundary conditions.
The Hamiltonian is written as $\hat H(t)=\hat H_{0,\mathrm{MFIC}}+\lambda(t)\hat V_{\mathrm{MFIC}}$, where the
unperturbed part contains a longitudinal field in addition to the Ising coupling:
\begin{align}
\hat{H}^{\,}_{0,\mathrm{MFIC}}
&=
\sum^{N}_{j=1}\Big(-JZ^{\,}_{j}Z^{\,}_{j+1}+B Z^{\,}_{j}\Big),
\qquad 
\hat{V}^{\,}_{\mathrm{MFIC}}=-J\sum^{N}_{j=1}X^{\,}_{j},
\qquad
\lambda(t)=\frac{h(t)}{J},
\nonumber
\end{align}
with $h(0)=0$.

We assume an initial thermal state at inverse temperature $\beta$,
$
\hat{\rho}^{\,}_{0}(\beta)=\frac{1}{Z^{\,}_{0}(\beta)}e^{-\beta\hat{H}^{\,}_{0,\mathrm{MFIC}}}
$
with
$
Z^{\,}_{0}(\beta)= \Tr\!e^{-\beta\hat{H}^{\,}_{0,\mathrm{MFIC}}}.
$
Since $\hat H_{0,\mathrm{MFIC}}$ is diagonal in the $Z$-product basis, all traces reduce to classical Ising sums
and can be evaluated via the transfer-matrix method. Throughout this section, it is convenient to work at inverse temperature $2\beta$.
Define
\begin{align}
K:=2\beta J,\qquad H:=2\beta B.
\nonumber
\end{align}
Introduce the symmetric $2\times2$ transfer matrix
\begin{align}
T^{(B)}_{s,s'}(2\beta)
=
\exp\!\left(
K\,ss' -\frac{H}{2}(s+s')
\right),
\qquad s,s'=\pm1.
\nonumber
\end{align}
Equivalently,
\begin{align}
T^{(B)}(2\beta)
=
\begin{pmatrix}
e^{K-H} & e^{-K}\\
e^{-K} & e^{K+H}
\end{pmatrix}
=
\begin{pmatrix}
e^{2\beta(J-B)} & e^{-2\beta J}\\
e^{-2\beta J} & e^{2\beta(J+B)}
\end{pmatrix}.
\nonumber
\end{align}
The partition function at inverse temperature $2\beta$ is
$
Z_0(2\beta)=\Tr\!\big[(T^{(B)})^N\big]=\Lambda_+^N+\Lambda_-^N,
$
with eigenvalues
\begin{align}
\Lambda_{\pm}
=
e^{K}\cosh H
\pm
\sqrt{e^{2K}\sinh^2 H+e^{-2K}}.
\label{eq:MFIC_lambdapm}
\end{align}

We will use the following two-eigenvalue trace identity later:
Let $T$ be a $2\times2$ matrix with eigenvalues $\Lambda_\pm$ with $\Lambda_+\neq\Lambda_-$. Then, for any $2\times2$ matrix $M$ and any integer $n\ge0$,
\begin{align}
\Tr(T^n M)=a^{\,}_+(M)\Lambda_+^n+a^{\,}_-(M)\Lambda_-^n,
\qquad
a^{\,}_+(M):=\frac{\Tr(TM)-\Lambda_-\,\Tr(M)}{\Lambda_+-\Lambda_-},
\qquad
a^{\,}_-(M):=\frac{\Lambda_+\,\Tr(M)-\Tr(TM)}{\Lambda_+-\Lambda_-}.
\label{eq:twoeig_trace_identity}
\end{align}

\subsection{Evaluation of $\delta V$ in MFIC}
\label{subsec:deltaV_MFIC_transfer}

We compute $\delta V$ [Eq.~\eqref{eq: QSL mixed states simplified}] in the form
\begin{align}
\delta V
= \frac{\|[\hat{\rho}^{\,}_{0},\hat{V}]\|^{\,}_{2}}{\|\hat{\rho}^{\,}_{0}\|^{\,}_{2}}
= \sqrt{2}\,
\sqrt{\frac{\Tr\!\left(\hat{\rho}^{2}_{0}\hat{V}^2\right)
-\Tr\!\left[\left(\hat{\rho}^{\,}_{0}\hat{V}\right)^2\right]}
{\Tr\!\left(\hat{\rho}^{2}_{0}\right)}}\, .
\nonumber
\end{align}
As in the TFIC computation [recall Eq.~\eqref{eq:deltaV2_unnormalized}], it is convenient to rewrite $\delta V$
in terms of unnormalized weights:
\begin{align}
(\delta V)^2
=
\frac{2}{Z^{\,}_{0}(2\beta)}
\Bigg[
\Tr\!\left(e^{-2\beta\hat H_{0,\mathrm{MFIC}}}\hat V^2\right)
-
\Tr\!\left(e^{-\beta\hat H_{0,\mathrm{MFIC}}}\hat V\,e^{-\beta\hat H_{0,\mathrm{MFIC}}}\hat V\right)
\Bigg],
\label{eq:deltaV2_unnormalized_MFIC}
\end{align}
where $Z_0(2\beta)=\Tr(e^{-2\beta\hat H_{0,\mathrm{MFIC}}})$.

\paragraph*{Step 1: Compute $\Tr(e^{-2\beta\hat H_{0,\mathrm{MFIC}}}\hat V^2)$.}
With $\hat V_{\mathrm{MFIC}}=-J\sum_{j=1}^N X_j$ and $(X^{\,}_{j})^{2}=\mathbb{I}$, we have
\begin{align}
\hat V_{\mathrm{MFIC}}^2
=J^2\sum_{j,k=1}^N X_j X_k
=
NJ^{2} \mathbb{I} +J^2\sum_{\substack{j,k\\ j\neq k}}^N X_j X_k.
\nonumber
\end{align}
Since $e^{-2\beta\hat H_{0,\mathrm{MFIC}}}$ is diagonal in the $Z$-basis, only $j=k$ contributes to the trace,
giving
\begin{align}
\Tr\!\left(e^{-2\beta\hat H_{0,\mathrm{MFIC}}}\hat V_{\mathrm{MFIC}}^2\right)
=NJ^2\,Z_0(2\beta).
\label{eq:MFIC_term1}
\end{align}

\paragraph*{Step 2: Compute $\Tr(e^{-\beta\hat H_{0,\mathrm{MFIC}}}\hat V\,e^{-\beta\hat H_{0,\mathrm{MFIC}}}\hat V)$.}
Using the same $Z$-basis argument (net spin flips must return to the original configuration),
only identical sites contribute:
\begin{align}
A^{(B)}:=\Tr\!\left(e^{-\beta\hat H_{0,\mathrm{MFIC}}}\hat V\,e^{-\beta\hat H_{0,\mathrm{MFIC}}}\hat V\right)
=
J^2\sum_{k=1}^N
\Tr\!\left(e^{-\beta\hat H_{0,\mathrm{MFIC}}}X_k\,e^{-\beta\hat H_{0,\mathrm{MFIC}}}X_k\right).
\nonumber
\end{align}
In the \(Z\)-product basis \(\ket{\mathbf s}=\ket{s_1,\dots,s_N}\), \(X_k\) flips \(s_k\), and one finds
\[
\langle \mathbf s|e^{-\beta \hat H_{0,\mathrm{MFIC}}}X_k e^{-\beta \hat H_{0,\mathrm{MFIC}}}X_k|\mathbf s\rangle
=
e^{-\beta E(\mathbf s)}e^{-\beta E(\mathbf s^{(k)})},
\]
where \(E(\mathbf s)=-J\sum_{j=1}^N s_js_{j+1}+B\sum_{j=1}^N s_j\).
By translation invariance it suffices to take \(k=1\).
A direct cancellation shows that the bonds $(N,1)$ and $(1,2)$ and the field term $Bs_1$ drop out in
$E(\mathbf s)+E(\mathbf s^{(1)})$, yielding
\begin{align}
E(\mathbf s)+E(\mathbf s^{(1)})
=
-2J\sum_{j=2}^{N-1}s_js_{j+1}
+2B\sum_{j=2}^{N}s_j.
\nonumber
\end{align}
Hence, $s_1$ completely decouples and produces an overall factor $2$, so that
\begin{align}
A^{(B)}
=
2NJ^2\,Q^{(B)}_{N-1}(2\beta),
\label{eq:MFIC_deltaV_term2}
\end{align}
where \(Q^{(B)}_{N-1}(2\beta)\) is the partition function of an \emph{open} Ising chain of length \(N-1\)
at inverse temperature \(2\beta\) with longitudinal field \(B\):
\begin{align}
Q^{(B)}_{N-1}(2\beta)
:=
\sum_{s_2,\dots,s_N=\pm1}
\exp\!\left(
K\sum_{j=2}^{N-1}s_js_{j+1}
-H\sum_{j=2}^{N}s_j
\right), 
\qquad (K\equiv 2\beta J,\ H\equiv 2\beta B).
\nonumber
\end{align}
Introduce the boundary vector
$
\mathbf u:=
\begin{pmatrix}
e^{-H/2} ,e^{+H/2}
\end{pmatrix}^{\mathsf T},
$
and use the same transfer matrix at inverse temperature \(2\beta\), i.e. \(T^{(B)}\equiv T^{(B)}(2\beta)\).
Then
\begin{align}
Q^{(B)}_{N-1}(2\beta)
=
{\mathbf u}^{\mathsf T}\,(T^{(B)})^{N-2}\,{\mathbf u}
=
\Tr\!\big[(T^{(B)})^{N-2}U^{(B)}\big],
\qquad
U^{(B)}:={\mathbf u}{\mathbf u}^{\mathsf T}.
\nonumber
\end{align}

Applying Eq.~\eqref{eq:twoeig_trace_identity} with $T=T^{(B)}$, $M=U^{(B)}$, and $n=N-2$ yields
\begin{align}
Q^{(B)}_{N-1}(2\beta)
&=
c_+^{(B)}\Lambda_+^{N-2}+c_-^{(B)}\Lambda_-^{N-2},
\qquad
c_\pm^{(B)}:=a_\pm(U^{(B)}),
\nonumber
\end{align}
where \(\Lambda^{\,}_{\pm}\) are given by Eq.~\eqref{eq:MFIC_lambdapm}, and
\begin{align}
\Tr(U^{(B)})&={\mathbf u}^{\mathsf T}{\mathbf u}=2\cosh H,
\qquad
\Tr(T^{(B)}U^{(B)})={\mathbf u}^{\mathsf T}T^{(B)}{\mathbf u}=2e^{K}\cosh(2H)+2e^{-K},
\nonumber
\\
\Lambda_{+}-\Lambda_{-}
&=2\sqrt{e^{2K}\sinh^2 H+e^{-2K}}.
\nonumber
\end{align}

\paragraph*{Step 3: Assemble $\delta V$.}
Substituting Eqs.~\eqref{eq:MFIC_term1} and \eqref{eq:MFIC_deltaV_term2} into
Eq.~\eqref{eq:deltaV2_unnormalized_MFIC} gives
\begin{align}
\delta V_{\mathrm{MFIC}}
=
\sqrt{2N}J
\left(
1-\frac{2Q^{(B)}_{N-1}(2\beta)}{Z_0(2\beta)}
\right)^{1/2}.
\label{eq:deltaV_MFIC_final_transfer}
\end{align}

In particular, for fixed finite \(\beta\) and large \(N\) [so that $(\Lambda^{\,}_{-}/\Lambda^{\,}_{+})^{N}\to 0$],
Eq.~\eqref{eq:deltaV_MFIC_final_transfer} simplifies to
\begin{align}
\delta V_{\rm MFIC}\xrightarrow[N\to\infty]{}
\sqrt{2N}\,J\,
\left(
1-\frac{2c^{(B)}_{+}}{\Lambda^{2}_{+}}
\right)^{1/2},
\label{eq:deltaV_MFIC_largeN}
\end{align}
with
\begin{align}
c^{(B)}_+
=
\frac{e^{K}\sinh^2H+e^{-K}}{\sqrt{e^{2K}\sinh^2 H+e^{-2K}}}+\cosh H,
\qquad 
\Lambda^{\,}_{+}= e^{K}\cosh H+
\sqrt{e^{2K}\sinh^2 H+e^{-2K}}.
\nonumber
\end{align}
Setting \(B=0\) reduces \(c^{(B)}_{+}\) to \(2\) and $\Lambda^{\,}_{+}$ to $2\cosh K$,
and hence reproduces the TFIC result
$\delta V_{\rm TFIC}\to\sqrt{2N}\,J\,\tanh(2\beta J)$ [Eq.~\eqref{eq:deltaV_TFIC_closed large N}].

\subsection{Evaluation of $\chi^{\,}_{\mathrm{F}}$ in MFIC}
\label{subsec:chiF_MFIC_transfer}

We next compute the fidelity susceptibility $\chi^{\,}_{\mathrm{F}}$ using Eq.~\eqref{eq: define fidelity susceptiblity simple form app2}
with $\hat H_0=\hat H_{0,\mathrm{MFIC}}$ and $\hat V=\hat V_{\mathrm{MFIC}}$.
Throughout this subsection we assume $B\neq 0,\pm 2J$, so that $B-J(s_{j-1}+s_{j+1})\neq 0$ for all local configurations.

In the $Z$-product basis $\ket{\mathbf s}=\ket{s_1,\dots,s_N}$ ($s_j=\pm 1$), the unperturbed energies are
\begin{align}
E(\mathbf s)=-J\sum_{j=1}^N s_js_{j+1}+B\sum_{j=1}^N s_j,
\qquad (s_{N+1}\equiv s_1),
\nonumber
\end{align}
and $\hat V_{\mathrm{MFIC}}=-J\sum_{j=1}^N X_j$ flips a single spin:
$\langle \mathbf s|\hat V_{\mathrm{MFIC}}|\mathbf s^{(j)}\rangle=-J$, where $\mathbf s^{(j)}$ is obtained
from $\mathbf s$ by $s_j\to -s_j$.
The corresponding energy difference is
$
\Delta E_j(\mathbf s)
:=E(\mathbf s)-E(\mathbf s^{(j)})
=
2s_j\Big(B-J(s_{j-1}+s_{j+1})\Big).
$

Starting from Eq.~\eqref{eq: define fidelity susceptiblity simple form app2} and using
$|\langle \mathbf s|\hat V_{\mathrm{MFIC}}|\mathbf s^{(j)}\rangle|^2=J^2$, we obtain
\begin{align}
\chi^{\rm MFIC}_{\mathrm F}
&=
\frac{2J^2}{Z_0(2\beta)}
\sum_{\mathbf s}\sum_{j=1}^N
\frac{\big(p(\mathbf s)-p(\mathbf s^{(j)})\big)^2}{\big(\Delta E_j(\mathbf s)\big)^2},
\qquad
p(\mathbf s):=e^{-\beta E(\mathbf s)}.
\nonumber
\end{align}
Using
$\big(p(\mathbf s)-p(\mathbf s^{(j)})\big)^2=e^{-2\beta E(\mathbf s)}\big(1-e^{+\beta\Delta E_j(\mathbf s)}\big)^2$
and $(\Delta E_j)^2=4\big(B-J(s_{j-1}+s_{j+1})\big)^2$, we can rewrite $\chi^{\,}_{\mathrm F}$ as an expectation value
in the $2\beta$ thermal average:
\begin{align}
\chi^{\rm MFIC}_{\mathrm F}
=
\frac{N J^2}{2}\,
\Bigg\langle
\frac{\big(1-e^{+\beta\Delta E_1(\mathbf s)}\big)^2}{\big(B-J(s_{N}+s_{2})\big)^2}
\Bigg\rangle_{2\beta},
\quad
\langle \cdots\rangle_{2\beta}
:=
\frac{1}{Z_0(2\beta)}
\sum_{\mathbf s}
\exp\!\Big(K\sum_{j=1}^N s_js_{j+1}-H\sum_{j=1}^N s_j\Big)\,(\cdots).
\label{eq:chiF_MFIC_expectation}
\end{align}
Here, we used translation invariance to replace $\sum_{j=1}^N$ by $N$ times $j=1$, and
$\beta\Delta E_1(\mathbf s)=s_1\big(H-K(s_N+s_2)\big)$.

\paragraph*{Transfer-matrix evaluation.}
Define
\begin{align}
S^{(B)}_1
:=
\sum_{\mathbf s}
\exp\!\Big(K\sum_{j=1}^N s_js_{j+1}-H\sum_{j=1}^N s_j\Big)\,
\frac{\big(1-e^{s_1(H-K(s_N+s_2))}\big)^2}{\big(B-J(s_N+s_2)\big)^2},
\nonumber
\end{align}
so that Eq.~\eqref{eq:chiF_MFIC_expectation} becomes
\begin{align}
\chi^{\rm MFIC}_{\mathrm F}
=
\frac{N J^2}{2}\,\frac{S^{(B)}_1}{Z_0(2\beta)}.
\label{eq:chiF_MFIC_S1def}
\end{align}
We can write
\begin{align}
S^{(B)}_1
&=
\sum_{a,b,c=\pm 1}
T^{(B)}_{a,b}\,T^{(B)}_{b,c}\,(T^{(B)})^{N-2}_{c,a}\;
\frac{\big(1-e^{b(H-K(a+c))}\big)^2}{\big(B-J(a+c)\big)^2},
\label{eq:S1_MFIC_local_sum}
\end{align}
where $a=s^{\,}_N$, $b=s^{\,}_1$, and $c=s^{\,}_2$.
Introduce the $2\times 2$ matrix $\mathcal{M}^{(B)}$ with entries
\begin{align}
\mathcal{M}^{(B)}_{a,c}
:=
\sum_{b=\pm 1}
T^{(B)}_{a,b}\,T^{(B)}_{b,c}\;
\frac{\big(1-e^{b(H-K(a+c))}\big)^2}{\big(B-J(a+c)\big)^2},
\qquad a,c=\pm 1.
\label{eq:M_MFIC_def}
\end{align}
Since $T^{(B)}$ is symmetric, Eq.~\eqref{eq:S1_MFIC_local_sum} becomes
$
S^{(B)}_1
=
\Tr\!\Big((T^{(B)})^{N-2}\mathcal{M}^{(B)}\Big).
$

\paragraph*{Explicit form of $\mathcal{M}^{(B)}$.}
Evaluating the $b$-sum in Eq.~\eqref{eq:M_MFIC_def} gives the symmetric matrix
\begin{align}
\mathcal{M}^{(B)}
=
\begin{pmatrix}
\mathcal{M}^{(B)}_{+,+} & \mathcal{M}^{(B)}_{+,-}\\
\mathcal{M}^{(B)}_{+,-} & \mathcal{M}^{(B)}_{-,-}
\end{pmatrix},
\nonumber
\end{align}
with
\begin{subequations}
\label{eq:M_MFIC_entries}
\begin{align}
\mathcal{M}^{(B)}_{+,+}
&=
\frac{8e^{-H}\,\sinh^{2}\!\big(\frac{H-2K}{2}\big)}{(B-2J)^2}
=
\frac{8e^{-2\beta B}\,\sinh^{2}\!\big(\beta(B-2J)\big)}{(B-2J)^2},
\nonumber
\\
\mathcal{M}^{(B)}_{-,-}
&=
\frac{8e^{+H}\,\sinh^{2}\!\big(\frac{H+2K}{2}\big)}{(B+2J)^2}
=
\frac{8e^{+2\beta B}\,\sinh^{2}\!\big(\beta(B+2J)\big)}{(B+2J)^2},
\nonumber
\\
\mathcal{M}^{(B)}_{+,-}
&=
\frac{8\sinh^{2}\!\big(\frac{H}{2}\big)}{B^2}
=
\frac{8\sinh^{2}\!\big(\beta B\big)}{B^2}.
\nonumber
\end{align}
\end{subequations}

\paragraph*{Closed form.}
Using Eq.~\eqref{eq:twoeig_trace_identity} with $T=T^{(B)}$, $M=\mathcal{M}^{(B)}$, and $n=N-2$, we obtain
\begin{align}
S^{(B)}_1
=
d_+^{(B)}\Lambda_+^{N-2}+d_-^{(B)}\Lambda_-^{N-2},
\qquad
d_\pm^{(B)}:=a_\pm(\mathcal{M}^{(B)}).
\nonumber
\end{align}
Explicitly,
\begin{align}
\Tr(\mathcal{M}^{(B)})
&=
8e^{-H}
\frac{\sinh^{2}\!\big(\frac{H-2K}{2}\big)}{(B-2J)^2}
+
8e^{+H}
\frac{\sinh^{2}\!\big(\frac{H+2K}{2}\big)}{(B+2J)^2},
\nonumber\\
\Tr\!\big(T^{(B)}\mathcal{M}^{(B)}\big)
&=
8e^{K-2H}\frac{\sinh^{2}\!\big(\frac{H-2K}{2}\big)}{(B-2J)^2}
+
8e^{K+2H}\frac{\sinh^{2}\!\big(\frac{H+2K}{2}\big)}{(B+2J)^2}
+
16e^{-K}\frac{\sinh^{2}\!\big(\frac{H}{2}\big)}{B^2},
\nonumber\\
\Lambda_{+}-\Lambda_{-}
&=2\sqrt{e^{2K}\sinh^2 H+e^{-2K}}.
\nonumber
\end{align}
Finally,
\begin{align}
\chi^{\rm MFIC}_{\mathrm F}
=
\frac{N J^2}{2}\,
\frac{d_+^{(B)}\Lambda_+^{N-2}+d_-^{(B)}\Lambda_-^{N-2}}{\Lambda_+^{N}+\Lambda_-^{N}}.
\label{eq:chiF_MFIC_final}
\end{align}

In the thermodynamic limit at fixed $\beta$ (so that $(\Lambda_-/\Lambda_+)^{N}\to 0$), Eq.~\eqref{eq:chiF_MFIC_final} simplifies to
\begin{align}
\chi^{\rm MFIC}_{\mathrm F}
\xrightarrow[N\to\infty]{}
\frac{N J^2}{2}\,\frac{d_+^{(B)}}{\Lambda_+^{2}},
\nonumber
\end{align}
where $d_+^{(B)}=a_+(\mathcal{M}^{(B)})$ and $\Lambda_\pm$ are given by Eq.~\eqref{eq:MFIC_lambdapm}.

\subsection{Threshold driving rate $\Gamma^{\,}_{\mathrm{th}}$ in MFIC}
\label{subsec:GammaTh_MFIC_transfer}

Using the definition of the threshold driving rate [Eq.~\eqref{eq: define threshold driving rate}]
together with Eq.~\eqref{eq:deltaV_MFIC_final_transfer} for $\delta V_{\mathrm{MFIC}}$ and
Eq.~\eqref{eq:chiF_MFIC_final} for $\chi^{\rm MFIC}_{\mathrm F}$, we obtain
\begin{align}
\Gamma^{\,}_{\mathrm{th,MFIC}}
&=
\frac{2\sqrt{2}\,\alpha}{\sqrt{N}}\,
\frac{Z_0(2\beta)}{J\,S^{(B)}_1}\,
\left(
1-\frac{2Q^{(B)}_{N-1}(2\beta)}{Z_0(2\beta)}
\right)^{1/2}
\nonumber\\
&=
\frac{2\sqrt{2}\,\alpha}{\sqrt{N}}\,
\frac{\Lambda_+^{N}+\Lambda_-^{N}}{J\big(d_+^{(B)}\Lambda_+^{N-2}+d_-^{(B)}\Lambda_-^{N-2}\big)}\,
\left(
1-\frac{2\big(c_+^{(B)}\Lambda_+^{N-2}+c_-^{(B)}\Lambda_-^{N-2}\big)}{\Lambda_+^{N}+\Lambda_-^{N}}
\right)^{1/2}.
\nonumber
\end{align}

Taking $\beta\to\infty$, the dominant single-spin-flip gap is $2J+|B|$, and
$S^{(B)}_1/Z_0(2\beta)\to(2J+|B|)^{-2}$ while $Q^{(B)}_{N-1}(2\beta)/Z_0(2\beta)\to 0$.
Hence the zero-temperature threshold rate is
\begin{align}
\Gamma^{\,}_{N}
:=
\lim_{\beta\to\infty}\Gamma^{\,}_{\mathrm{th,MFIC}}
=
\frac{2\sqrt{2}\,\alpha}{\sqrt{N}}\,
\frac{(2J+|B|)^2}{J}.
\nonumber
\end{align}
We normalize the temperature-dependent threshold rate by
\begin{align}
f^{\,}_{N}(\beta)
:=
\frac{\Gamma^{\,}_{\mathrm{th,MFIC}}}{\Gamma^{\,}_{N}}
=
\frac{\Lambda_+^{N}+\Lambda_-^{N}}{(2J+|B|)^2\big(d_+^{(B)}\Lambda_+^{N-2}+d_-^{(B)}\Lambda_-^{N-2}\big)}\,
\left(
1-\frac{2\big(c_+^{(B)}\Lambda_+^{N-2}+c_-^{(B)}\Lambda_-^{N-2}\big)}{\Lambda_+^{N}+\Lambda_-^{N}}
\right)^{1/2}.
\nonumber
\end{align}
By construction, $f^{\,}_{N}(\beta)\to 1$ as $\beta\to\infty$.

In the thermodynamic limit at fixed $\beta$,
\begin{align}
f(\beta)
:=
\lim_{N\to\infty} f^{\,}_{N}(\beta)
=
\frac{\Lambda_+^{2}}{(2J+|B|)^2\,d_+^{(B)}}\,
\left(
1-\frac{2c_+^{(B)}}{\Lambda_+^{2}}
\right)^{1/2},
\label{eq:fbeta_MFIC_thermo}
\end{align}
where $c_+^{(B)}$ is the thermodynamic-limit coefficient appearing in Eq.~\eqref{eq:deltaV_MFIC_largeN}.
We observe that $f(\beta)$ [Eq.~\eqref{eq:fbeta_MFIC_thermo}] is not guaranteed to be monotonic in $\beta$.
Expanding \(f(\beta)\) in the low- and high-temperature regimes, one finds
\begin{align}
f(\beta)\simeq
\begin{cases}
1+e^{-2\beta(2J+|B|)}, & \beta\to\infty,
\\[4pt]
\displaystyle
\frac{\sqrt{2+(B/J)^2}}{\sqrt{2}\,\bigl(2+|B|/J\bigr)^2}\;\frac{1}{\beta J},
& \beta\to 0,
\end{cases}
\nonumber
\end{align}
which matches the asymptotic temperature dependence stated in
Theorem~\ref{thm:GammaThTempScaling}, with
\[
c^{\,}_{1}=1,
\qquad
\Delta=2(2J+|B|),
\qquad
c^{\,}_{2}=
\frac{\sqrt{2+(B/J)^2}}{\sqrt{2}\,\bigl(2+|B|/J\bigr)^2}\;\frac{1}{J}.
\]

\end{document}